%% file: DaPoZa_SAF.tex
\documentclass[12pt,letterpaper]{article}
\input{Z_PACKAGES}
\input{Z_COMMANDS}

\begin{document}
\def\spacingset#1{\renewcommand{\baselinestretch}%
	{#1}\small\normalsize} \spacingset{1}

\title{\renewcommand{\thefootnote}{\fnsymbol{footnote}}\vspace{-1.5cm}\textbf{Sparse Approximate Factor Estimation for High-Dimensional Covariance Matrices}\footnotemark[1]
}
\renewcommand{\thefootnote}{\fnsymbol{footnote}}
\footnotetext[1]{Financial support by the Graduate School of Decision Sciences (GSDS), the German Science Foundation (DFG) and the German Academic Exchange Service (DAAD) is gratefully acknowledged. For helpful comments on an earlier draft of the paper we would like to thank Lyudmila Grigoryeva and Karim Abadir.
	The usual disclaimer applies.}

\author{
	\renewcommand{\thefootnote}{\alph{footnote}}
	\Large{Maurizio Daniele}\footnotemark[1]
	\\ University of Konstanz, GSDS
	\and
	\renewcommand{\thefootnote}{\alph{footnote}}
	\Large{Winfried Pohlmeier}\footnotemark[2]
	\\University of Konstanz, CoFE, RCEA
	\and 
	\renewcommand{\thefootnote}{\alph{footnote}}
	\Large{Aygul Zagidullina}\footnotemark[2]
	\\University of Konstanz, QEF
}
\renewcommand{\thefootnote}{\alph{footnote}}
\footnotetext[1]{Department of Economics, Universit\"atsstra\ss e 1, D-78457 Konstanz, Germany. Phone: +49-7531-88-2657, email: Maurizio.Daniele@uni-konstanz.de.}
\footnotetext[2]{Department of Economics, Universit\"atsstra\ss e 1, D-78457 Konstanz, Germany.}
\maketitle

\begin{abstract}
	We propose a novel estimation approach for the covariance matrix based on the $l_1$-regularized approximate factor model.
	Our sparse approximate factor (SAF) covariance estimator allows for the existence of weak factors and hence relaxes the pervasiveness assumption generally adopted 
	for the standard approximate factor model. We prove consistency of the covariance matrix estimator under the Frobenius norm  
	as well as the consistency of the factor loadings and the factors.
	
	\indent Our Monte Carlo simulations reveal that the SAF covariance estimator has superior properties in finite samples 
	for low and high dimensions and different designs of the covariance matrix.  
	Moreover, in an out-of-sample portfolio forecasting application the estimator uniformly outperforms alternative portfolio strategies  
	based on alternative covariance estimation approaches and modeling strategies including the $1/N$-strategy.
\end{abstract}

\noindent%
\vfill

%
\spacingset{1.37}
\setlength{\bibsep}{5pt plus 0.3ex}

\section{Introduction} \label{sec:intro}
The estimation of high-dimensional covariance matrices and their inverses (precision matrices) has recently received a great attention.
In economics and finance, it is central for portfolio allocation, risk measurement, asset pricing and 
graphical network analysis. The list of important applications from other areas of research includes, for example,  the analysis of climate data, gene classification 
and  image classification.  What appears to be a trivial estimation problem for a large sample size $T$  and a low dimensional vector of covariates,
turns out to be demanding, if $N$ is of the same order of magnitude or even larger than $T$.
In these cases, the sample covariance matrix becomes nearly singular and estimates the population covariance matrix poorly.
Moreover, assumptions of standard asymptotic theory with $T \to \infty $, holding $N $ fixed, turns out to be inappropriate and have to be replaced 
by assumptions allowing for both, $T$ and $N$,  approaching infinity.  

In recent years numerous studies proposed alternative estimation approaches for high-dimensional covariance matrices, which differ in the way of bounding the dimensionality problem. Two major approaches are 
factor models imposing a lower dimensional factor structure for the underlying multivariate process 
and regularization strategies for the parameters of the covariance matrix or its eigenvalues (see \citemain{Fan/Liao/Liu2016} for a recent survey on the estimation of large 
covariances and precision matrices).  
In this paper, we present an effective novel approach to the estimation of high-dimensional covariances, which 
profits from both branches of the literature. Our sparse approximate  
factor (SAF) approach to the estimation of high-dimensional covariance matrices is based on  $l_1$-regularization 
of the factor loadings and thereby is able to account for weak factors and shrinks elements in the covariance matrix towards zero.

Approaches to obtain consistent estimators by imposing a sparse structure on the covariance matrix directly include\nocitemain{Bickel2008}\nocitemain{BickelLevina2008}
\citeauthor{Bickel2008} (\citeyear{Bickel2008}, \citeyear{BickelLevina2008}),  \citemain{Cai/Liu2011} and \citemain{Cai/Zhou2012}.
These thresholding approaches  are shrinking small elements in the covariance matrix exactly to zero. While this may be 
a reasonable strategy, e.g. for genetic data, this assumption may not be appropriate for economic or financial data, where variables are driven
by common underlying factors. Such a feature may be more appropriately captured by covariance matrices based on factor representations.

In the literature on factor based covariance estimation \citemain{Fan2008a} consider the case of a strict factor representation with observed factors. 
This approach requires knowledge of additional observable variables (e.g. the Fama-French factors in the asset pricing framework), which may be an additional source 
of misspecification. Moreover, strict factor model representations impose the overly strong assumption of strictly uncorrelated idiosyncratic errors.  
This assumption was relaxed in \citemain{Fan2011a} and \citemain{FanLiaoMincheva2013}, who propose a covariance estimator based on an approximate factor model representation.
While \citemain{Fan2011a} shrink the entries of the covariance matrix of the idiosyncratic errors to zero using the adaptive thresholding technique by   
\citemain{Cai/Liu2011}, the approach proposed in \citemain{FanLiaoMincheva2013} rests on the more general principal
orthogonal complement thresholding method (POET) to allow for sparsity in the covariance matrix of the idiosyncratic errors. 

Our SAF covariance matrix estimator extends the existing framework on factor based approaches  
by imposing sparsity on both, the factor loadings and the covariance matrix of the idiosyncratic errors. Unlike imposing sparsity for the covariance matrix directly by thresholding 
or $l_1$-norm regularization, the $l_1$-regularization of the factor loadings does not necessarily imply zero entities of the covariance matrix, but simply reduces the dimensionality problem in the estimation of the factor driven part of the covariance matrix.  
Moreover, the sparsity in the matrix of factor loadings allows for weak factors, which only affect a subset of the observed variables.  Thus the SAF-approach 
relaxes the identifying assumption on the pervasiveness of the factors in the standard framework. This further implies that the eigenvalues of the covariance matrix corresponding to the common component are allowed to diverge at a slower rate than commonly considered (i.e. slower than $\mathcal{O}(N)$).   

The recent paper by 
\citemain{FanLiuWang2018}
claims that the relaxation of the pervasiveness assumption in the approximate factor model framework is the next major concern which should be addressed in future research. Hence, in this paper we focus exactly on this issue and build a bridge between the standard factor model and a relaxed pervasiveness assumption.

The weaker conditions on the eigenvalues  allow us to derive the consistency for the SAF covariance matrix estimator under the average Frobenius norm 
under rather mild regularity conditions. To our knowledge this convergence result is new. Because of the fast diverging eigenvalues for estimators based on the approximate 
factor model, convergence has  only be shown under the weaker weighted quadratic norm but not for the more general Frobenius norm (see, e.g. \citemain{FanLiaoMincheva2013}). 
As a byproduct of our proof for the SAF covariance matrix estimator, we also prove the consistency for the estimators of the sparse factor loadings, 
the factors and the covariance matrix of the idiosyncratic errors.

The favorable asymptotic properties of the SAF covariance matrix estimator are well supported by our Monte Carlo study based on different dimensions and alternative designs of the population covariance matrix. More precisely, the SAF covariance matrix estimator yields the lowest difference in the Frobenius norm to the true underlying covariance matrix compared to several competing estimation strategies. 

Finally, in an empirical study on the portfolio allocation problem, we show that the SAF covariance matrix estimator 
is a superior choice to construct the weights of the Global Minimum Variance Portfolio (GMVP) for low and large dimensional portfolios.  
Based on returns data from the S\&P 500  the estimator uniformly outperforms portfolio strategies based on alternative covariance estimation approaches 
and modeling strategies including the $1/N$-strategy in terms of different popular out-of-sample portfolio performance measures. 

The rest of the paper is organized as follows. In Section \ref{sec:model} we introduce the approximate factor model approach and show how sparsity can be obtained with respect to the factor loadings matrix by  $l_{1}$-regularization. Section \ref{sec:theo} discusses the theoretical setup and provides the convergence results. In Section \ref{sec:sim}, we present Monte-Carlo evidence on the finite sample properties of our new covariance estimator, while in Section \ref{sec:pf} we show the performance of our approach when applied to the empirical portfolio allocation problem.
Section \ref{sec:conclusions} summarizes the main findings and gives an outlook on future research.

Throughout the paper we will use the following notation: $\pi_{\max}(\bA)$ and $\pi_{\min}(\bA)$ are the maximum and minimum eigenvalue of a matrix $\bA$. Further, $\spec{\bA}$, $\frob{\bA}$ and $\lone{\bA}$ denote the spectral, Frobenius 
and the $l_1$-norm of $\bA$, respectively. They are defined as $\spec{\bA} = \sqrt{\pi_{\max}(\bA'\bA)}$,  $\frob{\bA} = \sqrt{\trace{\bA'\bA}}$ and $ \lone{\bA}= \max_{j} \sum_i |a_{ij} |$. For some constant $c > 0$ and a non-random sequence $b_N$, we use the notation $b_N = \mathcal{O}(N)$, if $N^{-1} b_N \to c$, for $N \to \infty$. Moreover, $b_N = o(N)$, if $N^{-1} b_N \to 0$, for $N \to \infty$. Similarly, for a random sequence $d_N$, we say $d_N = \mathcal{O}_p(N)$, if $N^{-1} d_N \overset{p}{\to} c$, for $N \to \infty$ and $d_N = o_p(N)$, if $N^{-1} d_N \overset{p}{\to} 0$, for $N \to \infty$, where $\overset{p}{\to}$ denotes convergence in probability.

\section{Factor Model Based Covariance Estimation}\label{sec:model}
\subsection{The Approximate Factor Model}\label{sub_sec:app_f_model}
\noindent
The following analysis is based on the approximate factor model (AFM) proposed by \citemain{Chamberlain1983}  to obtain a lower dimensional representation of a possibly 
high-dimensional covariance matrix. Let  
$x_{it}$ be the $i$-th observable variable at time $t$ for $i = 1, \dots, N$ and $t = 1, \dots, T$, such that 
$N$ and $T$ denote the sample size in the cross-section and in the time dimension, respectively.
The AFM is given by:
\begin{align}
	x_{it} = \blam_i' \bfa_t + u_{it} \label{approx_model} \, ,
\end{align}
where $\blam_i$ is a $(r \times 1)$-dimensional vector of factor loadings for variable $i$ and $\bfa_t$ is a $(r \times 1)$-dimensional vector of latent  factors at time $t$, where $r$ denotes the number of factors common to all variables in the model. 
Typically, we assume that $r$ is much smaller than the number of variables $N$. Finally, the idiosyncratic component $u_{it}$ accounts for variable-specific shocks, which are not captured by the common component $\blam_i' \bfa_t$. 
The AFM allows for weak serial and cross-sectional correlations among the idiosyncratic components with a dense covariance matrix of the idiosyncratic error term vector, $\bSigma_u = \cova{\left[(u_{1t}, u_{2t}, \ldots u_{Nt})'\right]}$. 
In matrix notation, \eqref{approx_model} can be written as:
\begin{align}
	\bX= \bLam \bFa'  + \bu \, ,		\label{approx_factor} 
\end{align}
where $\bX$ denotes a $(N \times T)$ matrix containing $T$ observations for $N$ weakly stationary time series. It is assumed that the time series are demeaned and standardized. $\bFa = (\bfa_1, \dots, \bfa_T)'$ is referred to as a $(T \times r)$-dimensional matrix of unobserved factors, $\bLam = (\blam_1, \dots, \blam_N)'$ is a $N \times r$ matrix of corresponding factor loadings and $\bu$ is a $(N \times T)$-dimensional matrix of idiosyncratic shocks.

There are several estimation approaches for a factor model as given by \eqref{approx_factor}. The principal component analysis (PCA)\footnote{See e.g., \citemain{Bai2002} for a detailed treatment of the PCA in approximate factor models.} and the quasi-maximum likelihood estimation (QMLE) under normality (see i.e. \citemain{Bai2016a}) are the two most popular ones.  In the following, we pursue estimating the factor model by QMLE. This allows us to introduce sparsity 
in the factor loadings by penalizing the likelihood function. Moreover, contrary to PCA,  all model parameters including the covariance matrix 
$\bSigma_u$ can be estimated jointly, while PCA-based second stage estimates of $\bSigma_{u}$ require consistent estimation of $\bLam$ and $\bFa$
in the first stage. This, however, may be problematic for the case of a relatively small $N$, because $\bFa$ can no longer be estimated consistently (\citemain{Bai2016}). 

The negative quasi log-likelihood function for the data in the AFM is defined as:
\begin{align}
	\mathcal{L}(\bLam, \bSigma_{F}, \bSigma_{u}) = \log\left|\det\left(\bLam \bSigma_{F}\bLam' + \bSigma_{u}\right)\right| + \traces{\bS_{x}\left(\bLam \bSigma_{F}\bLam' + \bSigma_{u}\right)^{-1}}, \label{quasi_lik_initial}
\end{align}
where $\bS_{x} = \frac{1}{T} \sum_{t = 1}^{T} \bx_t\bx_t'$ denotes the sample covariance matrix based on the observed data.  $\bSigma_{F}$ is the low dimensional covariance matrix of the factors. 
Within the framework of an AFM, the estimation of a full $\bSigma_u$ is cumbersome, as the number of parameters to estimate is $\frac{N(N+1)}{2}$ which may exceed the sample size $T$.
In order to overcome this problem, we treat $\bSigma_{u}$ as a diagonal matrix in the first step and define $\bPhi_u = \text{diag}\left(\bSigma_{u}\right)$ denoting a diagonal matrix that contains only the elements of the main diagonal of $\bSigma_{u}$. Furthermore, we restrict the covariance matrix of the factors to $\bSigma_{F} = \bI_r$.

Imposing these restrictions has the advantage that the estimation of the covariance matrix of the factors
becomes redundant.  Hence, our objective function reduces to:
\begin{align}
	\mathcal{L}(\bLam, \bPhi_{u}) = \log\left|\det\left(\bLam\bLam' + \bPhi_{u}\right)\right| + \traces{\bS_{x}\left(\bLam \bLam' + \bPhi_{u}\right)^{-1}} \, .  \label{quasi_lik}
\end{align}
As the true covariance matrix of $\bu_t$ allows for correlations of general form, but the previous objective function incorporates the error term structure of a strict factor model, \eqref{quasi_lik} may be seen as a quasi-likelihood. \citemain{Bai2016a} show that the QML estimator based on \eqref{quasi_lik} yields consistent parameter estimates. Hence, the consistency of $\bPhi_{u}$ is not affected by the general form of cross-section and serial correlations in $\bu_t$.

The factors $\bfa_t$ can be estimated by generalized least squares (GLS):
\begin{align}\label{gls_factors}
	\hat{\bfa}_t = \left(\hat{\bLam}' \hat{\bPhi}_{u}^{-1} \hat{\bLam}\right)^{-1}\hat{\bLam}'\hat{\bPhi}_{u}^{-1}\bx_t \, ,
\end{align}
where the estimates $\hat{\bLam}$ and $\hat{\bPhi}_{u}$ are the ones obtained from the optimization of the objective function in \eqref{quasi_lik}.

\subsection{The Sparse Approximate Factor Model} \label{factor_lasso_model}
\noindent The sparse approximate factor (SAF) model allows for sparsity in the factor loadings matrix $\bLam$ by shrinking 
single elements of $\bLam$ to zero. This is obtained by the $l_1$-norm penalized
MLE of \eqref{quasi_lik} based on the following optimization problem:
\begin{align}
	\underset{\bLam, \bPhi_{u}}{\min}\;\left[\log\left|\det\left(\bLam\bLam' + \bPhi_{u}\right)\right| + \traces{\bS_{x}\left(\bLam \bLam' + \bPhi_{u}\right)^{-1}} + \mu \sum_{k = 1}^{r} \sum_{i = 1}^{N}\left|\lambda_{ik}\right| \right], \label{sparse_lasso}
\end{align}
where $\mu \geq 0$ denotes a regularization parameter.  Note that the number of factors $r$ is predetermined and assumed to be fixed. 
Sparsity is obtained by shrinking some elements of $\bLam$ to zero, such that not all $r$ factors load on each $x_{it}$. Hence, this framework allows for weak factors (see, e.g. \citemain{Onatski2012}) that affect only a subset of the $N$ time series.

It is well known that the factors and factor loadings in AFM model in \eqref{approx_factor} are only identified up to an arbitrary non-singular rotation matrix $\bP$. This follows from the fact that $\bX = \bLam \bP \bP^{-1'}\bFa' + \bu = \bLam^{*} \bFa^{*'} + \bu$, with $\bLam^{*} = \bLam \bP$ and $\bFa^{*'} = \bP^{-1'}\bFa'$. \\
In contrast to the standard AFM model, which needs additional restrictions to identify $\bP$, our SAF model with embedded $l_1$-norm penalty function ensures the identification of the factors and factor loadings up to a unitary generalized permutation matrix $\bP$.\footnote{A short demonstration of the fact that $\bP$ can only be a unitary generalized permutation matrix for the $l_1$-norm is given in Section \ref{subsec:unitary_permutation} in the Supplement, as well as in \citemain{HornJohnson2012}.} 

Hence, by fixing the ordering of columns, e.g. by sorting the columns of the factor loadings matrix according to their respective sparsity, and assuming that the SAF estimator $\hat{\bLam}$ has identical column signs as the true factor loadings $\bLam_0$, as part of the identification conditions, the SAF model is fully identified. However, it should be noted that the identification of the SAF model only holds if the $l_1$-norm penalty on $\bLam$ enters the penalized optimization problem \eqref{sparse_lasso}, i.e. for $\mu > 0$. For $\mu = 0$, we are in the standard ML setting for the AFM and solely for this case we identify the model, following \citemain{Lawley1971}, by imposing the identification restriction that
$\bLam'\bPhi_{u}^{-1}\bLam$ is diagonal, with distinct diagonal entries that are arranged in a decreasing order.

In contrast to the weak factor assumption introduced in the following,  the pervasiveness assumption conventionally made for standard approximate factor models
(e.g. \citemain{Bai2002}, \citemain{stock2002forecasting}), implies that the $r$ largest eigenvalues of $\bLam' \bLam$ diverge at the rate $\mathcal{O}(N)$. Intuitively, this means that all factors are strong and the entire set of time series is affected. Consequently, the sparsity in the factor loadings matrix introduced 
in Assumption \ref{assum_lam} below considerably relaxes the conventional pervasiveness assumption.
\begin{assumption}[Weakness of the Factors]\label{assum_lam}
	\leavevmode\\
	There exists a constant $c > 0$ such that, for all $N$,
	\begin{align*}
		c^{-1} < \pi_{\min}\left(\frac{\bLam'\bLam}{N^{\beta}}\right) \leq \pi_{\max}\left(\frac{\bLam'\bLam}{N^{\beta}}\right) < c, \text{ where } 1/2\leq \beta \leq 1.\footnotemark
	\end{align*}
\end{assumption}
\footnotetext{The lower limit 1/2 for $\beta$ is necessary to consistently estimate the factors. See \lemref{lem_est_factor} in Section \ref{subsubsec:consistency_lam} in the Supplement.}
Assumption \ref{assum_lam} implies that the $r$ largest eigenvalues of $\bLam' \bLam$ diverge with the rate $\mathcal{O}\left( N^{\beta}\right)$,
which can be much slower than in the standard AFM. Furthermore, the parameter $\beta$ can take on different values for each of the eigenvalues of $\bLam' \bLam$. Hence, the eigenvalues can diverge at different rates. On the other hand, the special case of $\beta = 1$, implies the standard AFM framework with strong factors (i.e. \citemain{FanLiaoMincheva2013}, \citemain{Bai2016}). Hence, our sparse approximate factor model offers a convenient generalization of the standard one. Furthermore, Assumption \ref{assum_lam} has a direct implication on the sparsity of $\bLam$. In fact, this can be deduced by upper bounding the spectral norm of $\bLam$ according to the following expressions: 
\begin{align}\label{ine_spars}
	\lone{\bLam} \leq \sqrt{N} \spec{\bLam} = \mathcal{O}\left( N^{\left(1 + \beta\right)/2}\right) \text{ and } \lone{\bLam} \geq \spec{\bLam} = \mathcal{O}\left(N^{\beta/2}\right).
\end{align}
This result shows that imposing the weak factor assumption limits the amount of affected time series across all factors and hence requires a non-negligible amount of zero elements in each column of the factor loadings matrix. Nevertheless, the number of zero factor loadings can be arbitrarily small as $\beta$ increases. Note, that the lower bound of equation \eqref{ine_spars} restricts the number of zero elements in each column of $\bLam$, so that we can disentangle the common component from the idiosyncratic one.

The pervasiveness assumption imposed by the standard AFM, further implies a clear separation of the eigenvalues of the data covariance matrix into two groups, corresponding to the diverging eigenvalues of the common component and the bounded eigenvalues of the covariance matrix of the idiosyncratic errors. These characteristics can be observed in Figure \ref{fig_eig_strong}, where both panels illustrate the eigenvalue structure of datasets, that are simulated only based on strong factors for $T=450$ and different $N$.  
\begin{figure}[H]
	\begin{minipage}{.45\textwidth}
		\centering
		\captionsetup{width=1\linewidth}
		\subfloat[Eigenvalues for simulated data with 1 strong factor with $T = 450$]{\includegraphics[width=1\linewidth]{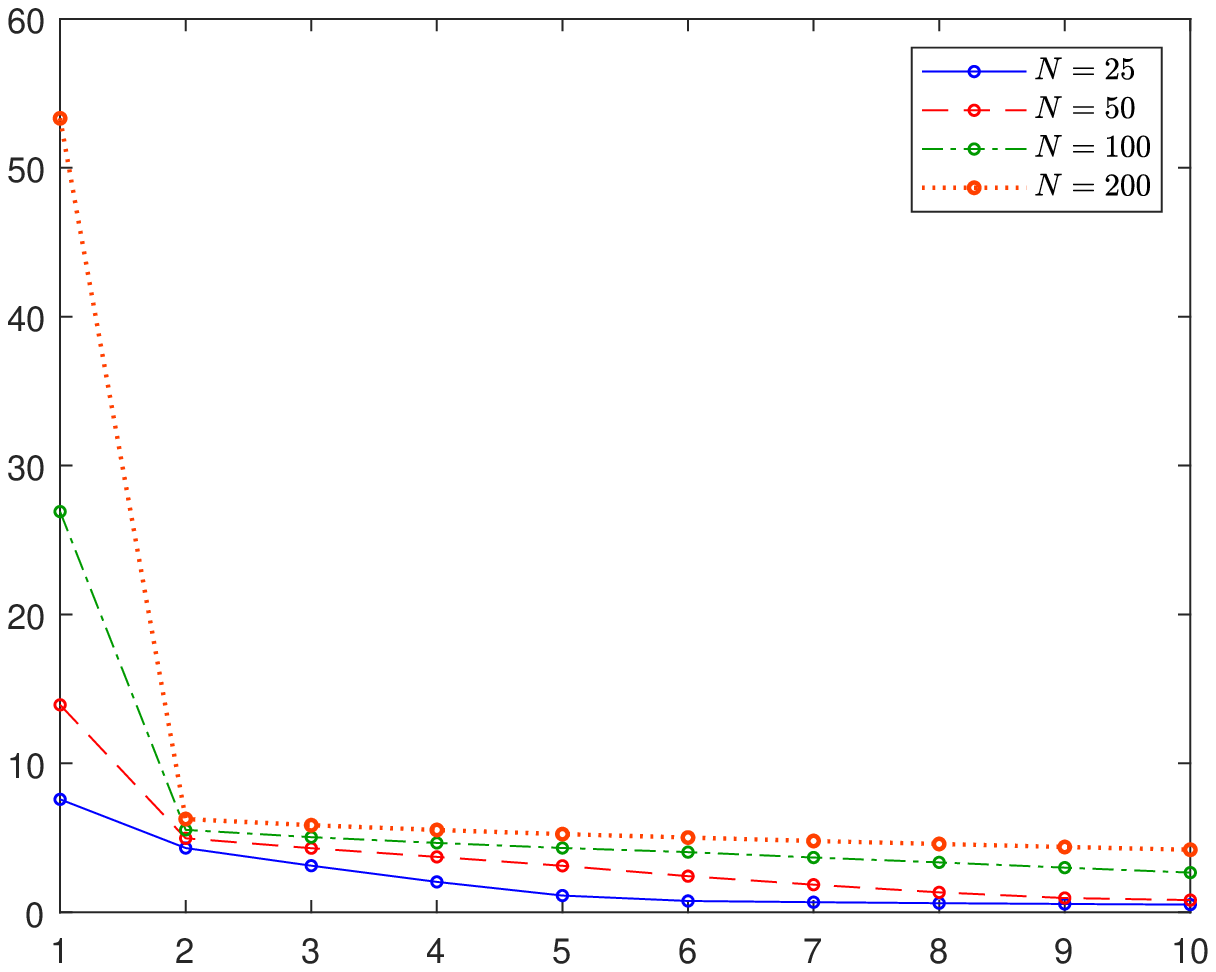}}
	\end{minipage}\hspace*{1cm}
	\begin{minipage}{.45\textwidth}
		\centering
		\captionsetup{width=1\linewidth}
		\subfloat[Eigenvalues for simulated data with 4 strong factors with $T = 450$]{\includegraphics[width=1\linewidth]{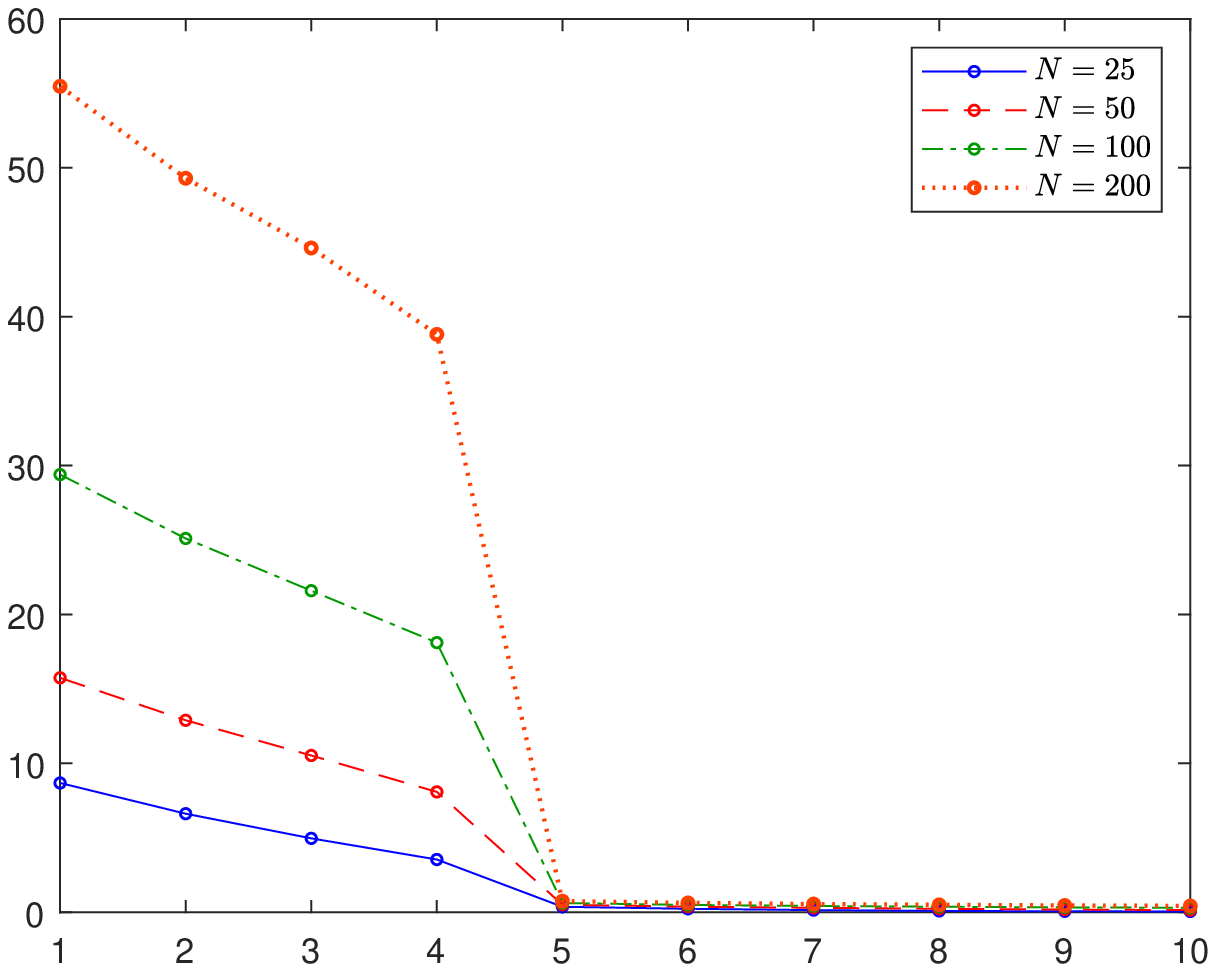}}
	\end{minipage}
	\caption{Structure of the eigenvalues based on strong factors}\label{fig_eig_strong}
\end{figure}
The panels differ solely in the number of factors included, where the left panel includes one strong factor and the right panel depicts the case of four strong factors. Both graphs reveal a clear partition in their respective eigenvalue structures, into sets of eigenvalues that diverge with the sample size $N$ corresponding to the number of included strong factors and sets of bounded eigenvalues associated to the idiosyncratic components.
\begin{figure}[H]
	\begin{minipage}{.45\textwidth}
		\centering
		\captionsetup{width=1\linewidth}
		\includegraphics[width=1\linewidth]{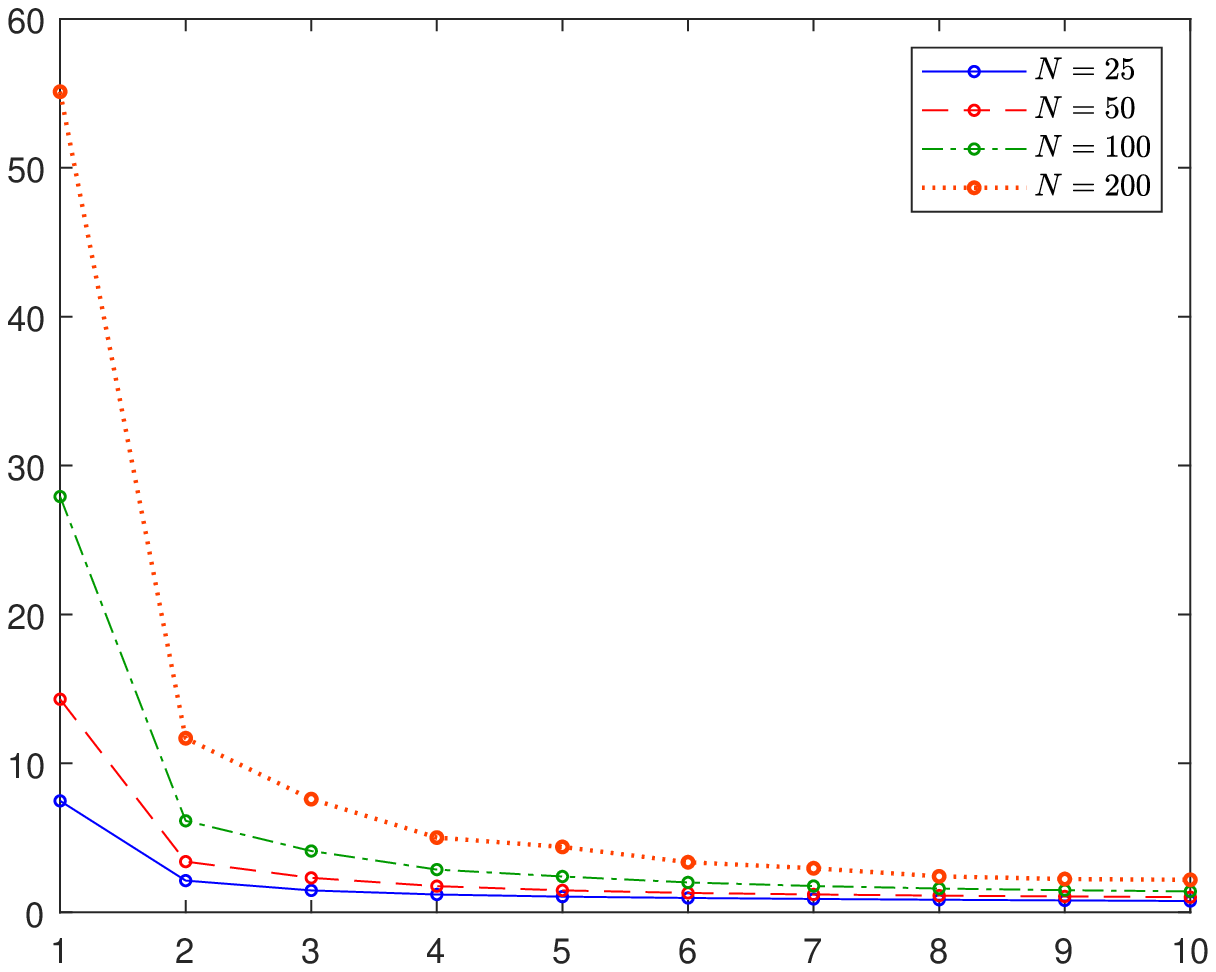}
		\caption{Eigenvalues for stock returns of stocks constituents of the S\&P 500 index with $T = 450$}\label{fig_sp500_eig}
	\end{minipage}\hspace*{1cm}
	\begin{minipage}{.45\textwidth}
		\centering
		\captionsetup{width=1\linewidth}
		\includegraphics[width=1\linewidth]{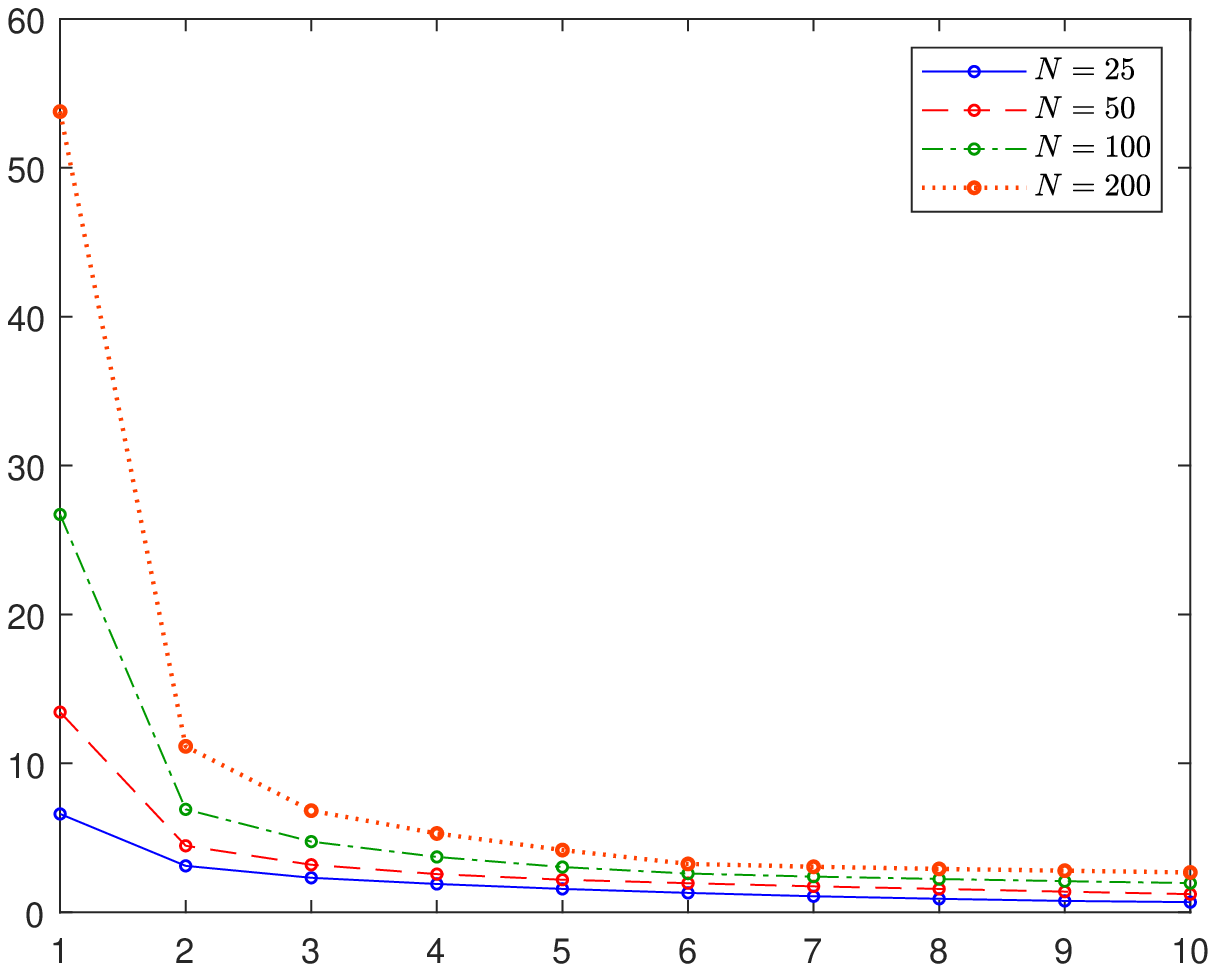}
		\caption{Eigenvalues for simulated data with 1 strong factor and 3 weak factors with $T = 450$}\label{fig_weak_eig}
	\end{minipage}
\end{figure}
However, such a clear separation in the eigenvalue structure of the covariance matrix cannot typically be found in real datasets. An example offers a dataset that contains the monthly asset returns of stocks constituents of the S\&P 500 stock index available for the entire period of 450 months,\footnote{The same dataset is also used in our empirical application and is described in more detail in Section \ref{sec:pf}.} whose eigenvalue distribution is illustrated in Figure \ref{fig_sp500_eig}. The graph shows a clear distinction between the first eigenvalue and the remaining eigenvalues.  However, the remaining eigenvalues diverge at a slower rate and a clear separation between the common and idiosyncratic component as implied by the standard AFM is impossible. Hence, the weak factor framework that allows for a slower divergence rate in the eigenvalues of the common component is more realistic for modeling the eigenvalue structure of real datasets. Furthermore, the weak factor assumption supports the well-documented empirical evidence that the eigenvalues of the sample covariance matrix of asset returns diverge at different rates (see, e.g. \citemain{Ross1976} and \citemain{Trzcinka1986}). Figure \ref{fig_weak_eig} depicts the eigenvalue structure of a dataset, which is generated by one strong factor and three weak factors. This model with weak factors nicely mimics the decaying eigenvalue structure we observe for the S\&P 500 asset returns.
\subsection{Estimation of the idiosyncratic error covariance matrix $\bSigma_{u}$} \label{error_factor}
In order to relax the imposed diagonality assumption on $\bSigma_{u}$ in the first step of our estimation, we re-estimate the covariance matrix 
of the idiosyncratic error term by means of the principal orthogonal complement thresholding (POET) estimator by \citemain{FanLiaoMincheva2013}. 
The POET estimator is based on soft-thresholding the off-diagonal elements of the sample covariance matrix of the residuals obtained from the estimation of an approximate factor model. Hence, it introduces sparsity in the idiosyncratic covariance matrix and offers a solution to the non-invertibility problem, generated using the sample covariance estimator, especially in high dimensional settings, where $N$ is close or even larger than $T$. More specifically, the estimated idiosyncratic error covariance matrix $\hat{\bSigma}_{u}^{\tau}$ based on the POET method is defined as:
\begin{align*}
	\hat{\bSigma}_{u}^{\tau} = \left(\hat{\sigma}_{ij}^{\tau}\right)_{N\times N}, \quad \hat{\sigma}_{ij}^{\tau} = \left\{\begin{array}{ll}
		\hat{\sigma}_{u,ii}, & i = j\\
		\mathcal{S}(\hat{\sigma}_{u,ij}, \tau), & i \neq j 
	\end{array}\right.
\end{align*}
where $\hat{\sigma}_{u,ij}$ is the $ij$-th element of the sample covariance matrix \\
$\bS_u = \frac{1}{T} \sum_{t = 1}^{T} (\bx_{t} - \hat{\bLam}\hat{\bfa}_t)(\bx_{t} - \hat{\bLam}\hat{\bfa}_t)'$ of the estimated factor model residuals, $\tau = \frac{1}{\sqrt{N}}+\sqrt{\frac{\log(N)}{T}}$ is a threshold\footnote{The threshold $\tau$ is based on the convergence rate of the idiosyncratic error covariance estimator specified in \lemref{lem_idio}. in Section \ref{subsubsec:rate_conv_sigma_u} in the Supplement.} and $\mathcal{S}(\cdot)$ denotes the soft-thresholding operator defined as:
\begin{align}\label{soft_t}
	\mathcal{S}(\sigma_{u,ij}, \tau) = \text{sign}(\sigma_{u,ij})(|\sigma_{u,ij}| - \tau)_+  \, .
\end{align}
In contrast to \citemain{FanLiaoMincheva2013}, who use the residuals of a static factor model based on the PCA estimator, 
our estimates are based on the residuals obtained from our sparse factor model. 

Thus, we follow the approach of \citemain{FanLiaoMincheva2013} and use a two-step procedure, where at first step we identify the common part, however, unlike in the PCA framework we allow for the weak factors; and in the second step, we model the general covariance structure for the idiosyncratic component.
By using a two-step procedure, we control for the sparsity patterns in $\bLam$ and $\bSigma_u$ separately and hence, this ensures that the sparsity in the loadings matrix is not distorted by the sparsity in the idiosyncratic error covariance matrix. 

Moreover, the joint estimation of two high-dimensional matrices with embedded $l_1$-norms, would become computationally burdensome and lead to considerable numerical instabilities. By separating the joint estimation into our two-step procedure we obtain a numerical stable optimization method that is computationally time-efficient. 

\subsection{SAF covariance matrix estimation}\label{fm_cov_est}
The estimator of the data covariance matrix based on the approximate factor model is obtained according to $ \bSigma = \cov{\bX} =  \bLam \bSigma_{F} \bLam' + \bSigma_{u}$.
Hereby, we first estimate the factors $\bfa_t$ and the factor loadings $\bLam$ according to our sparse factor model introduced in Section \ref{factor_lasso_model}. 
Consistent estimates of $\bLam$ and $\bfa_t$ are obtained by MLE and GLS as given by \eqref{quasi_lik} and \eqref{gls_factors}, respectively.
This yields the estimates of the common and idiosyncratic components of the AFM defined in \eqref{approx_model}. The latter one is used as input to estimate $\bSigma_u$ by the POET estimator   
introduced in Section \ref{error_factor}. Hence, our SAF covariance matrix estimator is given by:
\begin{align}\label{est_rfcov}
	\hat{\bSigma}_\text{SAF} = \hat{\bLam} \bS_{\hat{F}} \hat{\bLam}' + \hat{\bSigma}_{u}^{\tau}, 
\end{align}
where $\bS_{\hat{F}}$ denotes the sample estimator for the covariance matrix of the estimated factors, which is positive definite because the number of observations exceeds the number of factors. Further, using the convergence rate of the idiosyncratic error covariance matrix for the threshold $\tau$ also guarantees that $\hat{\bSigma}_{u}^{\tau}$ is positive definite with probability tending to one according to \citemain{Bickel2008}. Hence, the covariance matrix estimator $\hat{\bSigma}_\text{SAF}$ is positive definite by construction.

The implementations issues, the choice of the number of factors and the selection of the tuning parameter $\mu$ are described in Section \ref{sec:implem} in the Supplement.


\section{Large Sample Properties}\label{sec:theo}
In order to establish the consistency of the factor loadings matrix $\bLam$ and the data covariance matrix $\bSigma$ estimators, we adapt the following standard assumptions:
\begin{assumption}[Data generating process]\label{data_assum}
	\leavevmode
	\begin{enumerate}[label=(\roman*)]
		\item $\left\{\bu_t, \bfa_t\right\}_{t\geq 1}$ is strictly stationary. $\E{u_{it}} = \E{u_{it}f_{kt}} = 0$, $\forall i \leq N$, $k \leq r$ and $t \leq T$.\label{as11}
		\item There exist $r_1, r_2 > 0$ and $b_1, b_2 > 0$, such that for any $s > 0$, $i \leq N$ and $k \leq r$,\label{ass_exp}
		\begin{align*}
			\Prob{|u_{it}| > s} \leq \exp(-(s/b_1)^{r_1}), \quad \Prob{|f_{kt}| > s} \leq \exp(-(s/b_2)^{r_2}).
		\end{align*}
		\item Define the mixing coefficient: $\alpha(T) \coloneqq \sup_{A\in \mathcal{F}_{-\infty}^0, B\in \mathcal{F}_{T}^{\infty}} \left|\Prob{A}\Prob{B} - \Prob{AB}\right|,$
		where $\mathcal{F}_{-\infty}^0$ and $\mathcal{F}_{T}^{\infty}$ denote the $\sigma$-algebras generated by $\{(\bfa_t, \bu_t): -\infty \leq t \leq 0\}$ and $\{(\bfa_t, \bu_t): T \leq t \leq \infty\}$.\\
		Strong mixing: There exist $r_3 > 0$ and $C > 0$ s.t.: $\alpha(T) \leq \exp(-CT^{r_3}),$ $\forall T \in \mathcal{Z}^+$.
		\label{ass_sm}
		
		\item There exist constants $c_1, c_2 > 0$ such that $c_2 \leq \pi_{\min}\left(\bSigma_{u0}\right) \leq \pi_{\max}\left(\bSigma_{u0}\right) \leq c_1$. \label{assum_sig}
	\end{enumerate}
\end{assumption}
The assumptions in \ref{data_assum} impose regularity conditions on the data generating process and are identical to those imposed by \citemain{Bai2016}. Condition \textit{\ref{as11}} imposes strict stationarity for $\bu_t$ and $\bfa_t$ and requires that both terms are not correlated. Condition \textit{\ref{ass_exp}} requires exponential-type tails, which allows to use the large deviation theory for $\frac{1}{T} \sum_{t = 1}^{T} u_{it} u_{jt} - \sigma_{u, ij}$ and $\frac{1}{T} \sum_{t = 1}^{T} f_{jt} u_{it}$. In order to allow for weak serial dependence, we impose a strong mixing condition specified in Condition \textit{\ref{ass_sm}}. Further, Condition \textit{\ref{assum_sig}} implies bounded eigenvalues of the idiosyncratic error covariance matrix, which is a common identifying assumption in the factor model framework.

\begin{assumption}[Sparsity]\label{sparsity_assum}
	\leavevmode
	\begin{enumerate}[label=(\roman*)]
		\item $L_N = \sum_{k = 1}^{r} \sum_{i = 1}^{N}\1\left\{\lambda_{ik} \neq 0 \right\} = \mathcal{O}\left(N\right)$, \label{spars1}
		\item $S_N = \max_{i \leq N} \sum_{j = 1}^{N} \1\left\{\sigma_{u,ij} \neq 0 \right\}$, $S_N^2 d_T = o(1)$ and $S_N \mu = o(1)$\label{spars2},
	\end{enumerate}
	where $\1\{\cdot\}$ defines an indicator function that is equal to one if the boolean argument in braces is true, $d_T = \frac{\log N^{\beta}}{N} + \frac{1}{N^{\beta}}\frac{\log N}{T}$ and $\mu$ denotes the regularization parameter.
\end{assumption}
Assumptions \ref{sparsity_assum} imposes sparsity conditions on $\bLam$ and $\bSigma_{u}$, where condition \textit{\ref{spars1}} defines the quantity $L_N$ that reflects the number of non-zero elements in the factor loadings matrix $\bLam$. As the number of factors $r$ is assumed to be fixed,\textit{ \ref{spars1}} restricts the number of non-zero elements in each column of $\bLam$ to be upper bounded by $N$. At the same time, this assumption allows for a sparse factor loadings matrix with less than $N$ non-zero elements. Condition \textit{\ref{spars2}} specifies $S_N$ that quantifies the maximum number of non-zero elements in each row of $\bSigma_{u}$, following the definition of \citemain{Bickel2008}.  Furthermore, it restricts the number of zero elements in each row of $\Sigma_u$. Hence, it requires that $\Sigma_u$ is not too dense.
\subsection{Consistency of the Sparse Approximate Factor Model Estimator} \label{sec:theo:con_saf}
\begin{theorem}[Consistency of the Sparse Approximate Factor Model Estimator]\label{theo_consistency_lam}
	\leavevmode\\
	Under Assumptions \ref{assum_lam}, \ref{data_assum} and \ref{sparsity_assum} the sparse factor model in \eqref{sparse_lasso} satisfies the following properties, as $T$ and $N \to \infty$ and for $1/2 \leq \beta \leq 1$:
	\begin{align*}
		\frac{1}{N} \frob{\hat{\bLam} - \bLam_0}^2 &= \mathcal{O}_p\left(\mu^2 + \frac{\log N^{\beta}}{N} + \frac{1}{N^{\beta}}\frac{\log N}{T}\right),\\
		\frac{1}{N} \frob{\hat{\bPhi}_{u} -\bPhi_{u 0}}^2 &= \mathcal{O}_p\left(\frac{\log N^{\beta}}{N} + \frac{\log N}{T}\right).
	\end{align*}
	Hence, for $\, \log(N) = o(T)$ and the regularization parameter $\mu = o(1)$, we have:
	\begin{align*}
		\frac{1}{N} \frob{\hat{\bLam} - \bLam_0}^2 &= o_p(1), \quad \frac{1}{N} \frob{\hat{\bPhi}_{u} -\bPhi_{u0}}^2 = o_p(1) \quad \text{and} \quad \spec{\hat{\bfa}_t - \bfa_t} = o_p(1), \forall t \leq T.
	\end{align*}
	For the covariance matrix estimator of the idiosyncratic errors in the second step, specified in Section \ref{error_factor}, we get:
	\begin{align*}
		\spec{\hat{\bSigma}_u^{\tau} - \bSigma_u} = \mathcal{O}_p\left(S_N\sqrt{\mu^2 + \frac{N}{L_N} d_T}\right), \text{ for } d_T = \frac{\log N^{\beta}}{N} + \frac{1}{N^{\beta}}\frac{\log N}{T}.
	\end{align*}
	Hence, for $S_N^2 d_T = o(1)$ and $S_N \mu = o(1)$, this yields:
	$\spec{\hat{\bSigma}_u^{\tau} - \bSigma_u} = o_p(1).$
\end{theorem}
The proof of Theorem \ref{theo_consistency_lam} is given in the Sections \ref{subsubsec:consistency_lam} and \ref{subsubsec:rate_conv_sigma_u} in the Supplement. 
Under the given regularity conditions this theorem establishes the average consistency in the Frobenius norm of the estimators for the factor loadings matrix and idiosyncratic error covariance matrix based on our sparse factor model. More specifically, $\bLam$ and $\bPhi$ can be estimated consistently, regardless of the diagonality restriction on $\bSigma_{u}$ in the first step of our estimation procedure. Consequently, the factors $\bfa_t$ estimated based on GLS are as well consistent. The lower limit $1/2$ on $\beta$ is a necessary condition to achieve consistency. Intuitively this means that the factors should not be too weak such that there is still a clear distinction between the common and idiosyncratic component. Furthermore, the second step estimator of $\bSigma_u$ can be consistently estimated under the spectral norm. 

\subsection{Consistency of the Covariance Matrix Estimator}
Finally, in this section we take a closer look on the asymptotic properties of the SAF covariance matrix estimator, given in Section \ref{fm_cov_est}. 
The following theorem gives the convergence rates of the covariance matrix estimator and of its inverse under different matrix norms. 
\begin{theorem}[Convergence Rates for the Covariance Matrix Estimator]\label{theorem_cov}
	\leavevmode\\	
	Under Assumptions \ref{assum_lam}, \ref{data_assum} and \ref{sparsity_assum}, the covariance matrix estimator based on the SAF model in equation \eqref{est_rfcov} satisfies the following properties, as $T$, $N \to \infty$ and $1/2 \leq \beta \leq 1$:
	\begin{align}
		\frac{1}{N} \spec{\hat{\bSigma}_{\normalfont{\text{SAF}}} - \bSigma}_{\bSigma}^2 &= \mathcal{O}_p\left(\left[\mu^2 + d_T\right]^2 + \left[\frac{N^{\beta}}{N} + \frac{S_N^2}{N}\right]\left[\mu^2 + d_T\right]\right),\label{con_weigh}\\
		\frac{1}{N} \frob{\hat{\bSigma}_{\normalfont{\text{SAF}}} - \bSigma}^2 &= \mathcal{O}_p\left(N \left[\mu^2 + d_T\right]^2 + \left[N^{\beta} + S_N^2\right]\left[\mu^2 + d_T\right]\right), \label{fro_cov_con}\\
		\frac{1}{N}\frob{\hat{\bSigma}_{\normalfont{\text{SAF}}}^{-1} - \bSigma^{-1}}^2 &= \mathcal{O}_p\left(\left[\frac{1}{N^{\beta}} + S_N^2\right]\left[\mu^2 + d_T\right]\right),\label{cov_inv_con}
	\end{align}
	where $d_T = \frac{\log N^{\beta}}{N} + \frac{1}{N^{\beta}}\frac{\log N}{T}$ and $\spec{\bA}_{\bSigma} = \frac{1}{\sqrt{N}} \frob{\bSigma^{-1/2} \bA \bSigma^{-1/2}}$ denotes the weighted quadratic norm introduced by \citemain{Fan2008a}.
\end{theorem}
The proof of Theorem \ref{theorem_cov} is given in Section \ref{subsubsec:riskbounds} in the Supplement. Similar as for Theorem \ref{theo_consistency_lam}, we assume that the regularization parameter $\mu = o(1)$ and $\log(N) = o(T)$. Equation \eqref{con_weigh} in Theorem \ref{theorem_cov} shows that the covariance matrix estimator based on the sparse factor model in equation \eqref{est_rfcov} is consistent if we consider the weighted quadratic norm for the entire set of possible values for $\beta$.

Generally, convergence under the average Frobenius norm is hard to achieve because of the too fast diverging eigenvalues of the common component (see \citemain{FanLiaoMincheva2013}). However, according to equation \eqref{fro_cov_con} our SAF covariance matrix estimator is consistent, if $\mu = o\left(N^{-\beta/2}\right)$ and $1/2 \leq \beta \lessapprox 9/10$. Hence, the relaxation of the pervasiveness assumption in the standard approximate factor model to allow for weak factors leads to convergence of the covariance estimator under the average Frobenius norm. The upper bound for $\beta$ follows from the expression $\frac{N^{\beta} \log N^{\beta}}{N}$ in Equation \eqref{fro_cov_con} of Theorem \ref{theorem_cov}.\footnote{A closed form solution for the upper bound of $\beta$ is not feasible, hence we numerically approximate the maximum value of $\beta$ in the neighbourhood of one such that the expression $\frac{N^{\beta} \log N^{\beta}}{N}$ converges to zero.}  Further, Equation \eqref{cov_inv_con} of Theorem \ref{theorem_cov} shows that the inverse of $\bSigma_{\normalfont{\text{SAF}}}$ is consistently estimated under the average Frobenius norm.

\section{Monte Carlo Evidence}\label{sec:sim}
In the following, we present Monte Carlo evidence on the finite sample properties of our new covariance estimator. In particular, we  focus on the accuracy of the covariance matrix estimates depending on the dimensionality as well as on the strength of correlations in the true covariance matrix to be estimated. The simulation results for the SAF estimator are compared to the ones obtained from eight competing estimators that are popular in the literature.

\subsection{Monte Carlo Designs}
For our Monte Carlo experiments we use three different designs of the true covariance matrix $\bSigma$. In the first case, we consider the uniform covariance matrix design used in \citemain{AbadirDistasoZikes2014}, which takes the following form:
\begin{align}
	\sigma^\text{u}_{ii} = 1 \text{ and } \sigma^\text{u}_{ij} = \eta \;\mathcal{U}_{(0,1)}, \; \text{ for } i \neq j, \label{first_d}
\end{align}
where $\mathcal{U}_{(0,1)}$ denotes a standard uniform random variable, and we set $\eta \in \{0.025, 0.05, 0.075\}$. In this setting, $\eta$ controls for the correlations among the variables,
where an increase in $\eta$ amplifies the strength of the dependencies among the covariates.

For the second design, we use the sparse covariance matrix suggested by \citemain{BienTibshiraniothers2011}, which contains zero entries for the off-diagonals with a certain probability.
More specifically, the $ij$-th element of the covariance matrix $\sigma_{ij} = \sigma_{ji}$ is assigned to be non-zero with probability $p$, where $p \in \{0.05, 0.075, 0.1\}$. Similar as in the uniform design, the diagonal elements are set to 1. The non-zero off-diagonal elements are independently drawn from the uniform distribution $\mathcal{U}_{(0,0.2)}$.

Finally, the last design we consider is based on a generalized spiked covariance model as in \citemain{BaiYao2012}. More precisely, we use the following definition:
\begin{align}
	\bSigma_s =  \text{diag}\left(r_1, r_2, r_3, r_4, 0, \cdots, 0 \right) + \bSigma_u, \label{spiked_form}
\end{align}
where $r_1 - r_4$ correspond to four spiked eigenvalues and $\bSigma_u$ is a covariance matrix based on the uniform design in equation \eqref{first_d}. As this covariance matrix design complies with the approximate factor model framework, estimation approaches that are based on a factor model specification might benefit from this setting. More precisely, the first part of equation \eqref{spiked_form} is in accordance with the eigenvalue distribution of the common component in an AFM with four factors, whereas the second part in \eqref{spiked_form} corresponds to the covariance matrix of the idiosyncratic component and allows for weak correlations among the errors. In the simulation, we consider the following specification for the spiked eigenvalues: $r_1 = r_2 = N, r_3 = N^{0.8}, r_4 = N^{0.5}$. This design is in line with the weak factor framework, where the first two factors are strong and the last two correspond to weak factors.

For all three covariance matrix designs, we draw a time independent random data series $\bX$ from a multivariate normal distribution with zero population mean.\footnote{The same Monte Carlo experiments are carried out based on data from a multivariate t-distribution with five degrees of freedom. The results are rather similar to the multivariate normal setting and can be obtained upon request.}
The time dimension $T$ is set to 60, which relates to a dataset with 5 years of monthly data. The number of replications is 1000.
Further, we consider several dimensions for $\bX$ and set $N \in \{30, 50, 100, 200\}$.  
As goodness of fit criterion for the difference between the true and the estimated covariance matrix, we use the Frobenius norm.

\subsection{Alternative covariance estimation strategies} \label{sub_sec:cov_est}
Table \ref{tab:est_models} gives an overview of the methods for the covariance matrix estimation that are compared in our Monte Carlo experiments. A more detailed description of the alternative strategies is provided in Section \ref{sec:A_methods} in the Supplement.

\begin{table}[!htb]
	\footnotesize
	\setlength{\extrarowheight}{1pt}
	\centering
	\caption{Considered Models}
	\begin{tabular}{p{1.5cm}p{12cm}}
		\hline
		\hline
		1/N   & Equally Weighted Portfolio \bigstrut[t]\\
		Sample & Sample covariance matrix estimator \bigstrut[b]\\
		\hline
		\multicolumn{2}{c}{\textbf{Factor Models}} \bigstrut\\
		\hline
		SAF   & Sparse Approximate Factor Model \\
		POET  & POET covariance matrix estimator by \par \citemain{FanLiaoMincheva2013} \bigstrut[t]\\
		DFM   & Dynamic Factor Model estimated as in \par \citemain{Doz2011} \\
		SIM   & Single Index Model by \citemain{Sharpe1963}\\
		FF3F  & 3-Factor Model by \citemain{Fama1993} \bigstrut[b]\\
		\hline
		\multicolumn{2}{c}{\textbf{Covariance Matrix Shrinkage Strategies}} \bigstrut\\
		\hline
		LW    & The linear shrinkage estimator by \citemain{Ledoit2003} \bigstrut[t]\\
		KDM   & The inverse covariance matrix estimator by \par \citemain{Kourtis2012}\\
		ADZ   & The design-free covariance matrix estimator by \par \citemain{AbadirDistasoZikes2014}\\
		LW-NL & The non-linear shrinkage estimator by \citemain{LedoitWolf2018} \bigstrut[b]\\
		\hline
		\multicolumn{2}{c}{\textbf{Sparse Covariance Estimators}} \bigstrut\\
		\hline
		ST    & The soft-thresholding estimator as in \par \citemain{Rothman2009} \bigstrut[t]\\
		BT    &  The sparse covariance matrix estimator by \par \citemain{BienTibshiraniothers2011} \bigstrut[b]\\
		\hline
		\hline
	\end{tabular}
	\label{tab:est_models}
\end{table}

Since the observed factor models SIM and FF3F require additional, observable economic factors, they will only be considered in our empirical application to portfolio choice.
\subsection{Simulation results}
Table \ref{tab:sim_uni_norm} below contains the Monte Carlo results for the uniform design of the true covariance matrix,  
Table \ref{tab:sim_sparse_norm} gives the results based on the sparse covariance matrix design, while Table \ref{tab:sim_spiked_norm} shows the results for the covariance matrix design with spiked eigenvalues. Interestingly, we find a very 
similar and  clear picture. In terms of the goodness of fit, our sparse approximate factor model approach provides the smallest Frobenius norm, i.e. the SAF fits the true covariance matrix best. These results hold for all of the three rather different designs, all dimensions and degrees of correlation between the variables. Note that the advantage of the SAF model in accurately estimating the true covariance matrix is even more pronounced when $N$ increases, especially for the two high dimensional settings with $N= 100, 200$ and $T = 60$. Concerning the alternative approaches, ST, which is rather similar to our approach, performs second best in most of the scenarios.
\begin{table}[htb] 	
	\centering
	\scriptsize
	\caption{Simulation results - Uniform Covariance Matrix Design}
	\begin{threeparttable}
		\begin{tabular}{clccc|clccc}
			\hline
			\hline
			\multirow{2}[4]{*}{$N$} & \multicolumn{1}{c}{\multirow{2}[4]{*}{Model}} & \multicolumn{3}{c|}{$\eta$} & \multirow{2}[4]{*}{$N$} & \multicolumn{1}{c}{\multirow{2}[4]{*}{Model}} & \multicolumn{3}{c}{$\eta$} \bigstrut\\
			\cline{3-5}\cline{8-10}          & \multicolumn{1}{c}{} & 0.025 & 0.05  & 0.075 &       & \multicolumn{1}{c}{} & 0.025 & 0.05  & 0.075 \bigstrut\\
			\hline
			\multirow{9}[2]{*}{30} & Sample & 14.78 & 14.80 & 14.82 & \multirow{9}[2]{*}{100} & Sample & 168.25 & 167.80 & 167.00 \bigstrut[t]\\
			& \textbf{SAF} & \textbf{0.89} & \textbf{1.17} & \textbf{1.44} &       & \textbf{SAF} & \textbf{2.45} & \textbf{8.47} & \textbf{18.54} \\
			& POET  & 6.92  & 7.24  & 7.67  &       & POET  & 25.72 & 28.15 & 32.02 \\
			& DFM   & 6.41  & 6.59  & 6.84  &       & DFM   & 25.53 & 27.94 & 31.74 \\
			& LW    & 3.67  & 3.88  & 4.26  &       & LW    & 15.19 & 20.52 & 25.79 \\
			& ADZ   & 2.45  & 2.86  & 3.13  &       & ADZ   & 5.56  & 12.11 & 23.54 \\
			& LW-NL & 0.96  & 1.31  & 1.55  &       & LW-NL & 14.10 & 24.05 & 33.85 \\
			& ST    & 1.96  & 2.22  & 2.46  &       & ST    & 5.36  & 11.60 & 21.82 \\
			& BT    & 1.83  & 2.38  & 3.29  &       & BT    & 11.60 & 17.64 & 28.46 \bigstrut[b]\\
			\hline
			\multirow{9}[2]{*}{50} & Sample & 41.51 & 41.35 & 41.33 & \multirow{9}[2]{*}{200} & Sample & 674.22 & 672.98 & 671.15 \bigstrut[t]\\
			& \textbf{SAF} & \textbf{0.87} & \textbf{2.33} & \textbf{4.75} &       & \textbf{SAF} & \textbf{8.96} & \textbf{33.32} & \textbf{63.74} \\
			& POET  & 11.16 & 11.87 & 12.09 &       & POET  & 64.55 & 78.16 & 98.27 \\
			& DFM   & 11.00 & 11.69 & 12.03 &       & DFM   & 64.20 & 77.51 & 98.85 \\
			& LW    & 5.96  & 7.31  & 8.75  &       & LW    & 44.41 & 66.45 & 89.96 \\
			& ADZ   & 2.91  & 4.36  & 6.40  &       & ADZ   & 18.49 & 39.61 & 68.02 \\
			& LW-NL & 1.53  & 2.88  & 5.10  &       & LW-NL & 58.09 & 134.27 & 126.48 \\
			& ST    & 2.15  & 3.72  & 6.24  &       & ST    & 14.87 & 39.81 & 81.27 \\
			& BT    & 4.29  & 5.85  & 8.48  &       & BT    & 35.32 & 59.57 & 100.65 \bigstrut[b]\\
			\hline
			\hline
		\end{tabular}%
		\vspace*{-0.5cm}
		\begin{tablenotes}
			\footnotesize
			\singlespacing
			\item \leavevmode\kern-\scriptspace\kern-\labelsep 
			Note: The table gives the mean goodness of fit in terms of the Frobenius norm.	
		\end{tablenotes}
	\end{threeparttable}
	\label{tab:sim_uni_norm}%
\end{table}%

However, for small samples ($N = 30, 50$) it is outperformed by LW-NL, for the uniform and sparse covariance matrix designs. Furthermore, for the uniform covariance matrix design for high dimensions and very strong dependencies ($N= 100,200, \eta = 0.075$), ADZ performs slightly better than ST. 
It is also interesting to note that direct $l_1$-norm penalization of the covariance matrix as suggested by 
\citemain{ BienTibshiraniothers2011} does not do nearly as well as our approach, which profits from sparsity in the factor loadings matrix and thresholding of the covariance matrix of the idiosyncratic component. Moreover, the results for the POET estimator by \citemain{FanLiaoMincheva2013} that allows only for sparsity in the idiosyncratic error covariance matrix indicate that allowing for sparsity in the factor loadings matrix leads to a considerable improvement in the estimation accuracy.


\begin{table}[htb]
	\centering
	\scriptsize
	\caption{Simulation results - Sparse Covariance Matrix Design}
	\begin{threeparttable}
		\begin{tabular}{clccc|clccc}
			\hline
			\hline
			\multirow{2}[4]{*}{$N$} & \multicolumn{1}{c}{\multirow{2}[4]{*}{Model}} & \multicolumn{3}{c|}{$p$} & \multirow{2}[4]{*}{$N$} & \multicolumn{1}{c}{\multirow{2}[4]{*}{Model}} & \multicolumn{3}{c}{$p$} \bigstrut\\
			\cline{3-5}\cline{8-10}          & \multicolumn{1}{c}{} & 0.05  & 0.075 & 0.1   &       & \multicolumn{1}{c}{} & 0.05  & 0.075 & 0.1 \bigstrut\\
			\hline
			\multirow{9}[2]{*}{30} & Sample & 14.78 & 14.80 & 14.82 & \multirow{9}[2]{*}{100} & Sample & 168.06 & 167.84 & 167.08 \bigstrut[t]\\
			& \textbf{SAF} & \textbf{0.89} & \textbf{1.17} & \textbf{1.44} &       & \textbf{SAF} & \textbf{6.93} & \textbf{10.15} & \textbf{13.36} \\
			& POET  & 6.92  & 7.24  & 7.67  &       & POET  & 31.16 & 34.56 & 37.90 \\
			& DFM   & 6.78  & 7.10  & 7.32  &       & DFM   & 30.99 & 34.54 & 37.63 \\
			& LW    & 3.67  & 3.88  & 4.26  &       & LW    & 18.19 & 21.56 & 24.20 \\
			& ADZ   & 2.45  & 2.86  & 3.13  &       & ADZ   & 10.37 & 13.71 & 17.36 \\
			& LW-NL & 0.96  & 1.31  & 1.55  &       & LW-NL & 19.64 & 20.56 & 23.42 \\
			& ST    & 1.96  & 2.22  & 2.46  &       & ST    & 9.89  & 13.23 & 16.46 \\
			& BT    & 2.67  & 3.01  & 3.29  &       & BT    & 16.19 & 19.30 & 23.04 \bigstrut[b]\\
			\hline
			\multirow{9}[2]{*}{50} & Sample & 41.48 & 41.40 & 41.47 & \multirow{9}[2]{*}{200} & Sample & 673.93 & 673.22 & 672.50 \bigstrut[t]\\
			& \textbf{SAF} & \textbf{1.92} & \textbf{2.73} & \textbf{3.54} &       & \textbf{SAF} & \textbf{26.75} & \textbf{39.72} & \textbf{51.95} \\
			& POET  & 12.37 & 13.39 & 14.32 &       & POET  & 86.87 & 101.04 & 114.01 \\
			& DFM   & 12.24 & 13.15 & 13.86 &       & DFM   & 86.95 & 101.35 & 114.55 \\
			& LW    & 6.65  & 7.54  & 8.32  &       & LW    & 56.63 & 70.98 & 82.47 \\
			& ADZ   & 4.02  & 4.82  & 5.56  &       & ADZ   & 36.57 & 49.23 & 61.00 \\
			& LW-NL & 2.87  & 3.71  & 4.76  &       & LW-NL & 39.23 & 53.53 & 72.73 \\
			& ST    & 3.24  & 4.11  & 4.94  &       & ST    & 33.10 & 46.32 & 58.93 \\
			& BT    & 5.37  & 6.23  & 7.11  &       & BT    & 52.26 & 65.41 & 77.36 \bigstrut[b]\\
			\hline
			\hline
		\end{tabular}%
		\vspace*{-0.5cm}
		\begin{tablenotes}
			\footnotesize
			\singlespacing
			\item \leavevmode\kern-\scriptspace\kern-\labelsep 
			Note: The table gives the mean goodness of fit in terms of the Frobenius norm.	
		\end{tablenotes}
	\end{threeparttable}
	\label{tab:sim_sparse_norm}%
\end{table}%


\begin{table}[htb]
	\centering
	\scriptsize
	\caption{Simulation results - Spiked Eigenvalues Covariance Matrix Design}
	\begin{threeparttable}
		\begin{tabular}{clccc|clccc}
			\hline
			\hline
			\multirow{2}[4]{*}{$N$} & \multicolumn{1}{c}{\multirow{2}[4]{*}{Model}} & \multicolumn{3}{c|}{$\eta$} & \multirow{2}[4]{*}{$N$} & \multicolumn{1}{c}{\multirow{2}[4]{*}{Model}} & \multicolumn{3}{c}{$\eta$} \bigstrut\\
			\cline{3-5}\cline{8-10}          & \multicolumn{1}{c}{} & 0.025 & 0.05  & 0.075 &       & \multicolumn{1}{c}{} & 0.025 & 0.05  & 0.075 \bigstrut\\
			\hline
			\multirow{9}[2]{*}{30} & Sample & 357.67 & 356.99 & 364.88 & \multirow{9}[2]{*}{100} & Sample & 3685.11 & 3764.92 & 3648.53 \bigstrut[t]\\
			& \textbf{SAF} & \textbf{81.47} & \textbf{76.18} & \textbf{89.71} &       & \textbf{SAF} & \textbf{728.35} & \textbf{811.97} & \textbf{703.69} \\
			& POET  & 348.22 & 345.81 & 355.75 &       & POET  & 3217.32 & 3301.45 & 3200.86 \\
			& DFM   & 374.08 & 362.31 & 370.60 &       & DFM   & 3859.43 & 3782.30 & 3591.11 \\
			& LW    & 235.92 & 225.21 & 243.05 &       & LW    & 1943.57 & 1884.93 & 1591.67 \\
			& ADZ   & 363.31 & 335.58 & 359.26 &       & ADZ   & 37344.83 & 37130.45 & 37395.58 \\
			& LW-NL & 433.97 & 422.25 & 437.14 &       & LW-NL & 9551.24 & 9636.64 & 9581.48 \\
			& ST    & 121.28 & 120.42 & 131.18 &       & ST    & 894.49 & 1004.33 & 892.38 \\
			& BT    & 100.96 & 95.58 & 108.89 &       & BT    & 1149.95 & 1085.80 & 1073.54 \bigstrut[b]\\
			\hline
			\multirow{9}[2]{*}{50} & Sample & 964.06 & 973.54 & 972.58 & \multirow{9}[2]{*}{200} & Sample & 14918.21 & 14629.11 & 13959.27 \bigstrut[t]\\
			& \textbf{SAF} & \textbf{205.87} & \textbf{212.06} & \textbf{210.39} &       & \textbf{SAF} & \textbf{3172.47} & \textbf{3143.15} & \textbf{2687.99} \\
			& POET  & 886.14 & 896.26 & 897.94 &       & POET  & 12729.59 & 12471.16 & 11878.79 \\
			& DFM   & 937.97 & 983.28 & 925.26 &       & DFM   & 14816.18 & 15167.83 & 13272.48 \\
			& LW    & 598.89 & 563.39 & 543.15 &       & LW    & 7039.80 & 5968.86 & 5044.31 \\
			& ADZ   & 9257.10 & 9300.04 & 9273.15 &       & ADZ   & 12587.36 & 12504.54 & 12258.25 \\
			& LW-NL & 1605.19 & 1573.07 & 1604.67 &       & LW-NL & 52593.67 & 52959.12 & 53066.79 \\
			& ST    & 268.97 & 286.27 & 294.92 &       & ST    & 4208.99 & 4141.47 & 3595.39 \\
			& BT    & 241.11 & 245.81 & 241.74 &       & BT    & 6402.21 & 6317.62 & 6306.77 \bigstrut[b]\\
			\hline
			\hline
		\end{tabular}%
		\vspace*{-0.5cm}
		\begin{tablenotes}
			\footnotesize
			\singlespacing
			\item \leavevmode\kern-\scriptspace\kern-\labelsep 
			Note: The table gives the mean goodness of fit in terms of the Frobenius norm.	
		\end{tablenotes}
	\end{threeparttable}
	\label{tab:sim_spiked_norm}%
\end{table}%
\FloatBarrier
\section{An Application to Portfolio Choice}\label{sec:pf}
Empirical portfolio models, particularly when applied to large asset spaces, suffer from a high degree of instability. 
The estimation of $N $ mean and $N(N + 1)/2$  variance-covariance parameters yields extremely noisy estimates of  portfolio
weights with large standard errors. It is well-documented that these estimated portfolios show poor out-of-sample
performance, extreme short positions and no diversification (e.g. \citemain{Jobson/Korkie1980} and \citemain{Michaud1989}).
In order to mitigate these shortcomings and to improve portfolio estimates against extreme
estimation noise, a range of alternative strategies have been proposed including the shrinkage estimation of the covariance matrix of asset returns
(\citemain{Ledoit2003}; \citemain{LedoitWolf2018} and \citemain{Kourtis2012}). 

In the following, we investigate to what extent the SAF model can be used to obtain robustified
estimates of high-dimensional covariance matrices of asset returns as input for empirical portfolio  models.
In an out-of-sample portfolio forecasting experiment, we compare the performance of the global minimum variance portfolio (GMVP) strategy based on 
a covariance matrix estimated by our sparse factor model to popular alternative portfolio strategies with regularized covariance estimators. 
As in many other studies, we restrict our analysis to the GMVP, because its vector of portfolio weights, 
$ \bomega = \frac{\bSigma^{-1}  \mathbf{1}_N}{\mathbf{1}_N' \bSigma^{-1} \mathbf{1}_N}$, is solely a function of the covariance matrix of the asset returns. 
Thus, for estimating the GMVP  the mean vector of asset returns is redundant and its empirical performance only depends on the quality of 
the covariance matrix estimator.

In a first step, we theoretically analyze the properties of the GMVP weights based on the SAF estimator. The results are summarized in the following proposition:  
\begin{prop}\label{prop_eig_cov}
	Based on the general definition of the covariance matrix of an approximate factor model given in Section \ref{fm_cov_est}, we obtain:
	\begin{align*}
		\sum_{k = 1}^r \pi_k\left(\bLam\bLam'\right) &= \trace{\bLam\bLam'} = \sum_{i = 1}^{N} \sum_{k = 1}^r \lambda_{ik}^2,\\
		\sum_{i = 1}^{N} \pi_i\left(\bSigma^{-1}\right) &\leq \sum_{i = 1}^{N} \pi_i\left(\bI_N\right) - \frac{\sum_{i = 1}^{N}\pi_i\left(\bLam\bLam'\right)}{N + \sum_{i = 1}^{N}\pi_i\left(\bLam\bLam'\right)}.
	\end{align*}
\end{prop}
The proof is given in Section \ref{subsec:gmv_weights} in the Supplement. \propref{prop_eig_cov} shows that allowing for sparsity in the factor loadings matrix leads to shrinking the eigenvalues of the precision matrix towards the ones of an identity matrix. Hence, the portfolio weights based on our SAF model are shrunken towards those of the $1/N$ portfolio. This result makes intuitively sense as it is reasonable to invest in the equally weighted portfolio in the case of great estimation instabilities regarding the covariance matrix.   

\subsection{Data and Design of the Forecasting Experiment}\label{emp_desc}
The dataset comprises the monthly excess returns of stocks of the S\&P 500 index, that were constituents of the index in December, 2016. The excess returns are obtained by subtracting the corresponding one-month Treasury bill rate from the asset returns. We consider the time period from January, 1980 until December, 2016, which yields $T = 443$ monthly returns for each of the 205 available stocks.\footnote{The return data are retained from Thompson Reuters Datastream.} In order to check the performance of our estimator with respect to the dimensionality of the asset space, we consider the following portfolio sizes: $N \in \{30, 50, 100, 200\}$. Out of the 205 stocks, we select at random individual subsets from the overall number of assets and work with the selected assets for the entire forecasting experiment.

Since by construction, a theoretical portfolio built on a subset of assets from a larger portfolio cannot outperform the larger one, 
an observed inferiority  of the larger empirical portfolio can only be the consequence of higher estimation noise 
due to the larger dimensionality, which overcompensates for the ex-ante theoretical superiority.
Therefore, this selection strategy provides us with insights into the impact of estimation noise on the performance of empirical portfolios.

In order to estimate the portfolio weights for each strategy, we apply a rolling window approach with $h=60$ months, corresponding to 5 years of historic data. Thus, at time $t$ we use 
the last 60 months from $t-59$ until $t$ for our estimation. Using the estimated portfolio weights, we compute the out-of-sample portfolio return $\hat{r}^p_{t+1}(s) = \hat{\bomega}(s)' \pmb{r}_{t+1}$ for the period $t+1$ 
for the 12 different estimation strategies $s=1, \ldots, 12$. All portfolios are rebalanced on a monthly basis. This generates a series of $T-h$ out-of-sample portfolio returns. 
The results are then used to estimate the mean $\mu(s)$ and variance $\sigma^2(s)$ of the portfolio returns for each strategy by their empirical counterparts:
\begin{align}
	\hat \mu(s) = \frac{1}{T} \sum_{t = h+1}^{T} \hat r^p_t(s) 
	\qquad \text{and} \qquad \hat{\sigma}^2(s) = \frac{1}{T-1} \sum_{t = h+1}^{T} \left(\hat{r}^p_{t}(s) - \hat{\mu}(s)\right)^2 \, .  \label{est_mean_var}
\end{align}
We repeat this procedure 100 times to avoid that the out-of-sample results depend on the initially randomly selected stocks. 
Hence, all results reported below are average outcomes across the 100 forecasting experiments.

The criteria for the performance evaluation are the out-of-sample standard deviation (SD), the average return (AV), the Certainty Equivalent (CE) and the Sharpe ratio (SR).

\subsection{Out-of-Sample Portfolio Performance}
Table \ref{est_res_sim_sp500} contains the annualized results of our comparative study on the out-of-sample performance of different portfolio estimation approaches.
The results represent average outcomes across the 100 different forecasting experiments for each of the four performance measures. 
Our sparse approximate factor model (SAF) yields the lowest out-of-sample portfolio standard deviation for all portfolio dimensions, i.e. 
it is performing best for the performance criterion the GMVP-strategy is designed for.

In theory, the GMVP-strategy may not necessarily outperform the $1/N$-strategy in terms of the remaining three performance criteria, since it 
completely disregards optimization with respect to the expected portfolio return. Nevertheless, our SAF model also outperforms the $1/N$-strategy 
and the other estimation approaches in terms of AV, CE and SR, which depend on the expected return. 
In the portfolio forecasting experiment, our regularization method does best for the expected out-of-sample portfolio return.

It is of utmost importance to note that the superiority of our approach does not only hold for different performance measures, but also for all 
portfolio dimensions. The SAF model performs best for low, but also for high dimensional portfolios, for which the sample size is much smaller than the portfolio dimension, i.e. $T \ll N$. This indicates, at least for this specific application, that the selection of the penalty parameter is reasonable.

\begin{table}[htb]
	
	\footnotesize
	
	\caption{Estimation results for the Portfolio Application}
	
	\centering
	\resizebox{\textwidth}{!}{
		\begin{threeparttable}
			
			\begin{tabular}{c|cccccccccccc}
				\hline
				\hline
				Model & 1/N & Sample & \textbf{SAF} & POET  & DFM   & SIM   & FF3F  & LW    & KDM   & ADZ   & LW-NL & BT \bigstrut\\
				\hline
				\multicolumn{13}{c}{N = 30} \bigstrut\\
				\hline
				SD    & 0.1572 & 0.2184 & \textbf{0.1538} & 0.1694 & 0.1661 & 0.1574 & 0.1557 & 0.1638 & 0.1680 & 0.1618 & 0.1620 & 0.1571 \bigstrut[t]\\
				AV    & 0.1002 & 0.0970 & \textbf{0.1005} & 0.0920 & 0.0954 & 0.1005 & 0.0955 & 0.0958 & 0.0967 & 0.0938 & 0.0976 & 0.0996 \\
				CE    & 0.0878 & 0.0731 & \textbf{0.0887} & 0.0776 & 0.0816 & 0.0881 & 0.0834 & 0.0824 & 0.0826 & 0.0807 & 0.0845 & 0.0873 \\
				SR    & 0.6372 & 0.4431 & \textbf{0.6534} & 0.5417 & 0.5747 & 0.6386 & 0.6138 & 0.5838 & 0.5761 & 0.5786 & 0.6024 & 0.6335 \bigstrut[b]\\
				\hline
				\multicolumn{13}{c}{N = 50} \bigstrut\\
				\hline
				SD    & 0.1543 & 0.3812 & \textbf{0.1501} & 0.1654 & 0.1619 & 0.1545 & 0.1519 & 0.1603 & 0.1590 & 0.1543 & 0.1585 & 0.1610 \bigstrut[t]\\
				AV    & 0.0996 & 0.1041 & \textbf{0.1006} & 0.0971 & 0.0949 & 0.0999 & 0.0936 & 0.1008 & 0.0966 & 0.0994 & 0.0927 & 0.1012 \\
				CE    & 0.0876 & 0.0312 & \textbf{0.0893} & 0.0834 & 0.0817 & 0.0879 & 0.0820 & 0.0880 & 0.0839 & 0.0875 & 0.0801 & 0.0882 \\
				SR    & 0.6452 & 0.2725 & \textbf{0.6703} & 0.5876 & 0.5871 & 0.6463 & 0.6160 & 0.6296 & 0.6079 & 0.6444 & 0.5847 & 0.6284 \bigstrut[b]\\
				\hline
				\multicolumn{13}{c}{N = 100} \bigstrut\\
				\hline
				SD    & 0.1525 & -     & \textbf{0.1456} & 0.1593 & 0.1580 & 0.1527 & 0.1489 & 0.1558 & 0.1644 & 0.1534 & 0.1522 & 0.1600 \bigstrut[t]\\
				AV    & 0.0999 & -     & \textbf{0.1042} & 0.0972 & 0.0907 & 0.1003 & 0.0913 & 0.1021 & 0.0931 & 0.0999 & 0.1001 & 0.0993 \\
				CE    & 0.0883 & -     & \textbf{0.0936} & 0.0845 & 0.0783 & 0.0886 & 0.0802 & 0.0900 & 0.0796 & 0.0881 & 0.0885 & 0.0865 \\
				SR    & 0.6556 & -     & \textbf{0.7153} & 0.6105 & 0.5755 & 0.6567 & 0.6127 & 0.6560 & 0.5666 & 0.6514 & 0.6573 & 0.6212 \bigstrut[b]\\
				\hline
				\multicolumn{13}{c}{N = 200} \bigstrut\\
				\hline
				SD    & 0.1505 & -     & \textbf{0.1410} & 0.1534 & 0.1559 & 0.1507 & 0.1465 & 0.1496 & 0.1539 & 0.1455 & 0.1459 & 0.1472 \bigstrut[t]\\
				AV    & 0.0993 & -     & \textbf{0.1071} & 0.0998 & 0.0904 & 0.0996 & 0.0893 & 0.1037 & 0.0934 & 0.1014 & 0.1026 & 0.0984 \\
				CE    & 0.0879 & -     & \textbf{0.0972} & 0.0880 & 0.0782 & 0.0882 & 0.0786 & 0.0925 & 0.0815 & 0.0909 & 0.0920 & 0.0876 \\
				SR    & 0.6596 & -     & \textbf{0.7600} & 0.6505 & 0.5801 & 0.6608 & 0.6094 & 0.6932 & 0.6065 & 0.6971 & 0.7036 & 0.6689 \bigstrut[b]\\
				\hline
				\hline
			\end{tabular}%
			\vspace*{-0.5cm}
			\begin{tablenotes}
				
				\footnotesize
				\singlespacing
				\item \leavevmode\kern-\scriptspace\kern-\labelsep 
				Note: The sparse approximate factor model (SAF) in the third column is compared to the equally weighted portfolio ($1/N$), the GMVP based on the sample covariance matrix (Sample), the POET estimator by \citemain{FanLiaoMincheva2013} (POET), the Dynamic Factor Model (DFM), the Single Factor Model by \citemain{Sharpe1963} (SIM), the Three-Factor Model by \citemain{Fama1993} (FF3F), the estimators by \citemain{Ledoit2003} (LW), \citemain{Kourtis2012} (KDM), \citemain{AbadirDistasoZikes2014} (ADZ), \citemain{LedoitWolf2018} (LW-NL) and \citemain{BienTibshiraniothers2011} (BT).
				
			\end{tablenotes}
			
	\end{threeparttable}}

	\label{est_res_sim_sp500}
	
\end{table}%

As mentioned earlier, increasing the portfolio dimension does not necessarily improve the out-of-sample performance 
of an empirical portfolio as the theoretical gains maybe overcompensated by the increase in estimation noise due to the increase in the number of parameters to be estimated.
It is not too surprising that this phenomenon is most dramatically pronounced for the plug-in  estimator of the GMVP, but we also find it to some extent for the DFM.  
Moreover, for the SIM and FF3F, we do not find a strict monotonicity between portfolio dimension and portfolio performance,
while the performance of our SAF model strictly increases with $N$.

While in the portfolio forecasting experiment for any performance measure and any portfolio dimension our sparse factor model shows the best performance, there is no clear further ranking 
regarding the other approaches. FF3F is performing second best in terms of the minimization of portfolio risk for all portfolio dimensions, but it is outperformed by other estimation approaches when performance measures other than the portfolio risk are considered.

Our comparative study  also confirms the findings of  \citemain{DeMiguel2009a} that the $1/N$ portfolio is a strong competitor for many alternative portfolio strategies.  For low dimensions ($N = 30$ and $N = 50$), we can see that, apart from our estimator only the single factor model generates a higher average SR compared to the equally weighted portfolio, although it is very close to it. 
In terms of the portfolio risk, only our method and FF3F reveal performance  superior to the $1/N$ portfolio for low dimensions of the asset space. 
The picture slightly changes, when higher asset dimensions ($N > 50$) are considered. For higher dimensions, the method by \citemain{AbadirDistasoZikes2014} 
is a serious competitor to the $1/N$ portfolio. This mirrors our finding from the simulation study in Section \ref{sec:sim}, where the ADZ estimator performs comparatively well in high dimensional settings with strong linear dependencies.


Table \ref{est_res_sim_sp500_w} in  Appendix \ref{sec:A_tables} provides additional insights into the quality of the weight estimates.
The summary statistics indicate that the outstanding performance of the SAF model results from effectively stabilizing the estimated portfolio weights by avoiding extreme positions  (moderate minima and maxima in the weight estimates) and by the low
standard deviations. Furthermore, the results show that the weights of our SAF estimator shrink towards the weights of the equally weighted portfolio as $N$ increases. This is in line with the theoretical results in Proposition \ref{prop_eig_cov}. The relative good performance of  SIM and FF3F result from very low variation in the portfolio weights, which come for the SIM with $N=200$ close to the constant weights of the equally weighted portfolio.

In order to check the robustness of our findings, which are based on data from January 1980 until December 2016, we also consider forecasts based on subperiods. We restrict our attention to the standard deviation of the out-of-sample portfolio returns and consider how a gradual increase of the evaluation sample affects the performance of the competing estimators. 
The results are illustrated in Figure \ref{sd_sub} in Appendix \ref{sec:A_figures}, where the portfolio standard deviation at time $t$ incorporates the out-of-sample portfolio returns until $t$ (e.g. the out-of-sample portfolio standard deviation in January 2005 incorporates the out-of-sample portfolio returns from January 1985 until January 2005).  Special attention is given  to the periods before and after the financial crisis in 2007. 
The graphs indicate that the SAF estimator also provides for different subperiods the lowest portfolio standard deviation compared to FF3F and LW-NL. Note, that the difference is more pronounced when the recent financial crisis period is included. Hence, in comparison to our SAF model both,  FF3F and LW-NL,  fail to pick up the changing risk during the crisis and, as a result, they provide more volatile portfolio estimates. 


\section{Conclusions}\label{sec:conclusions}
In this paper, we propose a novel approach for the estimation of high-dimensional covariance matrices based on a sparse approximate factor model. The estimator allows for sparsity in the factor loadings matrix by shrinking single elements of the factor loadings matrix to zero. 
Hence, this setting reduces the number of parameters to be estimated and therefore leads to a reduction in estimation noise. 
Furthermore, the sparse factor model framework allows for weak factors, which only affect a subset of the available time series. 
Thus, our framework offers a convenient generalization to the pervasiveness assumption in the standard approximate factor model that solely leads to strong factors.

We prove average consistency under the Frobenius norm for the factor loadings matrix estimator and consistency in the spectral norm for the idiosyncratic component covariance matrix estimator based on our sparse approximate factor model. The factors estimated using the GLS method are also shown to be consistent. Furthermore, we derive average consistency for our factor model based covariance matrix estimator under the Frobenius norm for a particular rate of divergence for the eigenvalues of the covariance matrix corresponding to the common component. To the best of our knowledge, this result has not been shown in the existing literature because of the fast diverging eigenvalues. Additionally, we provide consistency results of our covariance matrix estimator under the weighted quadratic norm.

In our Monte Carlo study, we analyze the finite sample properties of our covariance matrix estimator for different simulation designs for the true underlying covariance matrix. The results show that our estimator offers the lowest difference in Frobenius norm to the true covariance matrix compared to the competing estimators. Further, the benefit of the covariance matrix estimator based on our sparse factor model is even more pronounced if the dimensionality of the problem increases.

In an out-of-sample portfolio forecasting experiment, we compare the performance of the global minimum variance portfolio based on the covariance matrix estimator of our sparse approximate factor model to alternative estimation approaches frequently used in the literature. The forecasting results reveal that our estimator yields the lowest average out-of-sample portfolio standard deviation across different portfolio 
dimensions.  At the same time, it generates the highest Certainty Equivalent and Sharpe Ratio compared to all considered portfolio strategies. The performance gains of our SAF model are especially pronounced during the recent financial crisis. Hence, our estimator has a stabilizing impact on the portfolio weights, especially during highly volatile periods.

The results of our out-of-sample portfolio forecasting study show a substantial reduction of the portfolio standard deviation of the dynamic factor model compared to the standard approximate factor model, especially for small asset dimensions. Hence, it would be interesting to analyze if a possible extension of our SAF model by considering dynamic factors, would as well lead to a more efficient estimation of the covariance matrix. We leave this for future research.


\section*{Supplement}
The supplementary material provides the omitted proofs and additional related results. Furthermore, it discusses the implementations issues, the choice of the number of factors and the selection of the tuning parameter $\mu$.

\bibliographystylemain{econometrica}
\bibliographymain{DaPoZaFLasso}

\newpage
\begin{appendix}
	\section*{Appendix}
	\label{sec:appendix}
	\section{Tables}
	\label{sec:A_tables}
	 \begin{table}[H]
		
		\footnotesize
		
		\caption{Summary Statistics for the Estimated Portfolio Weights}
		
		\centering
		\resizebox{\textwidth}{!}{
			\begin{threeparttable}
				
				\begin{tabular}{c|cccccccccccc}
					\hline
					\hline
					Model & 1/N & Sample & \textbf{SAF} & POET  & DFM   & SIM   & FF3F  & LW    & KDM   & ADZ   & LW-NL & BT \bigstrut\\
					\hline
					\multicolumn{13}{c}{N = 30} \bigstrut\\
					\hline
					Min   & 0.0333 & -0.2676 & 0.0049 & -0.1150 & -0.0870 & 0.0325 & 0.0155 & -0.1163 & -0.0423 & -0.0572 & -0.0670 & -0.0080 \bigstrut[t]\\
					Max   & 0.0333 & 0.2898 & 0.0529 & 0.1765 & 0.1255 & 0.0369 & 0.0547 & 0.1498 & 0.1064 & 0.1505 & 0.1467 & 0.0660 \\
					SD    & 0.0000 & 0.1336 & 0.0122 & 0.0685 & 0.0502 & 0.0011 & 0.0094 & 0.0635 & 0.0356 & 0.0513 & 0.0523 & 0.0179 \\
					MAD   & 0.0000 & 0.1051 & 0.0097 & 0.0533 & 0.0389 & 0.0008 & 0.0073 & 0.0496 & 0.0280 & 0.0408 & 0.0415 & 0.0141 \bigstrut[b]\\
					\hline
					\multicolumn{13}{c}{N = 50} \bigstrut\\
					\hline
					Min   & 0.0200 & -0.5199 & -0.0043 & -0.1045 & -0.0769 & 0.0195 & 0.0035 & -0.1136 & -0.0231 & 0.0150 & -0.0686 & -0.0590 \bigstrut[t]\\
					Max   & 0.0200 & 0.5124 & 0.0353 & 0.1378 & 0.0934 & 0.0225 & 0.0393 & 0.1212 & 0.0626 & 0.0251 & 0.1224 & 0.0888 \\
					SD    & 0.0000 & 0.2219 & 0.0092 & 0.0510 & 0.0365 & 0.0007 & 0.0078 & 0.0503 & 0.0184 & 0.0023 & 0.0425 & 0.0324 \\
					MAD   & 0.0000 & 0.1745 & 0.0073 & 0.0397 & 0.0284 & 0.0005 & 0.0061 & 0.0393 & 0.0144 & 0.0018 & 0.0338 & 0.0256 \bigstrut[b]\\
					\hline
					\multicolumn{13}{c}{N = 100} \bigstrut\\
					\hline
					Min   & 0.0100 & -     & -0.0159 & -0.0776 & -0.0577 & 0.0097 & -0.0030 & -0.0903 & -0.0488 & -0.0471 & -0.0418 & -0.0684 \bigstrut[t]\\
					Max   & 0.0100 & -     & 0.0224 & 0.0935 & 0.0601 & 0.0115 & 0.0248 & 0.0865 & 0.0686 & 0.0917 & 0.0822 & 0.0878 \\
					SD    & 0.0000 & -     & 0.0077 & 0.0315 & 0.0224 & 0.0003 & 0.0054 & 0.0331 & 0.0233 & 0.0274 & 0.0245 & 0.0311 \\
					MAD   & 0.0000 & -     & 0.0061 & 0.0246 & 0.0175 & 0.0002 & 0.0043 & 0.0258 & 0.0185 & 0.0217 & 0.0195 & 0.0247 \bigstrut[b]\\
					\hline
					\multicolumn{13}{c}{N = 200} \bigstrut\\
					\hline
					Min   & 0.0050 & -     & -0.0165 & -0.0539 & -0.0419 & 0.0049 & -0.0038 & -0.0612 & -0.0368 & -0.0324 & -0.0336 & -0.0376 \bigstrut[t]\\
					Max   & 0.0050 & -     & 0.0145 & 0.0587 & 0.0368 & 0.0059 & 0.0150 & 0.0563 & 0.0510 & 0.0620 & 0.0660 & 0.0593 \\
					SD    & 0.0000 & -     & 0.0063 & 0.0183 & 0.0136 & 0.0002 & 0.0034 & 0.0197 & 0.0159 & 0.0166 & 0.0173 & 0.0171 \\
					MAD   & 0.0000 & -     & 0.0050 & 0.0142 & 0.0106 & 0.0001 & 0.0027 & 0.0154 & 0.0126 & 0.0132 & 0.0136 & 0.0135 \bigstrut[b]\\
					\hline
					\hline
				\end{tabular}%
				\vspace*{-0.5cm}
				\begin{tablenotes}
					
					\footnotesize
					\singlespacing
					\item \leavevmode\kern-\scriptspace\kern-\labelsep 
					Note: The sparse approximate factor model (SAF) in the third column is compared to the equally weighted portfolio ($1/N$), the GMVP based on the sample covariance matrix (Sample), the POET estimator by \citemain{FanLiaoMincheva2013} (POET), the Dynamic Factor Model (DFM), the Single Factor Model by \citemain{Sharpe1963} (SIM), the Three-Factor Model by \citemain{Fama1993} (FF3F), the estimators by \citemain{Ledoit2003} (LW), \citemain{Kourtis2012} (KDM), \citemain{AbadirDistasoZikes2014} (ADZ), \citemain{LedoitWolf2018} (LW-NL) and \citemain{BienTibshiraniothers2011} (BT).
					
				\end{tablenotes}
				
		\end{threeparttable}}
		
		\label{est_res_sim_sp500_w}
		
	\end{table}%

	\section{Figures}
	\label{sec:A_figures}
	
	\begin{figure}[H]
		\begin{minipage}{.5\textwidth}
			\centering
			\captionsetup{width=1\linewidth}
			\subfloat[$N = 30$]{\includegraphics[width=1\linewidth]{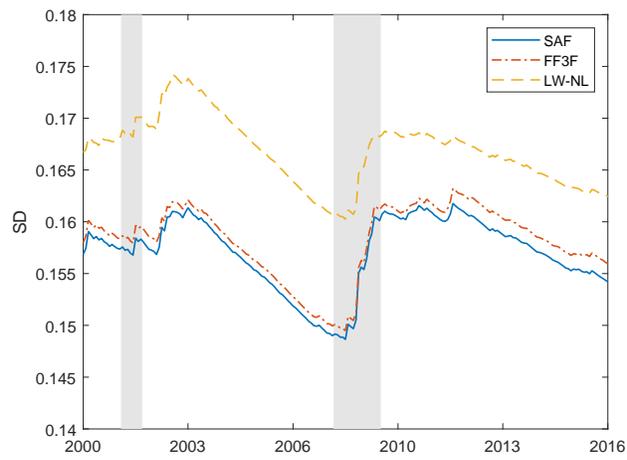}}
		\end{minipage}\hspace*{1cm}
		\begin{minipage}{.5\textwidth}
			\centering
			\captionsetup{width=1\linewidth}
			\subfloat[$N = 50$]{\includegraphics[width=1\linewidth]{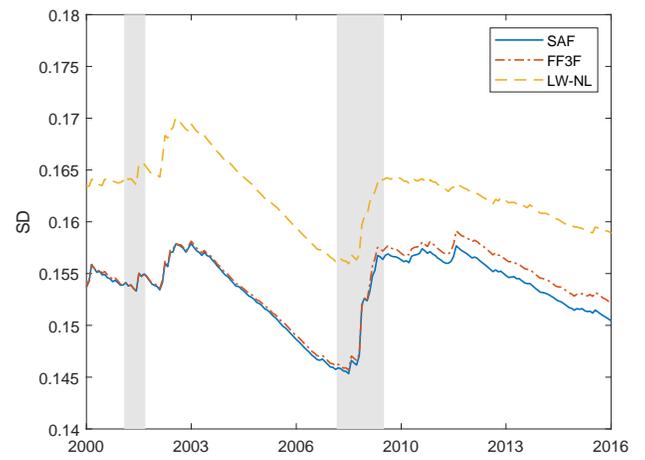}}
		\end{minipage}\\\vspace{0.5cm}
		
		\begin{minipage}{.5\textwidth}
			\centering
			\captionsetup{width=1\linewidth}
			\subfloat[$N = 100$]{\includegraphics[width=1\linewidth]{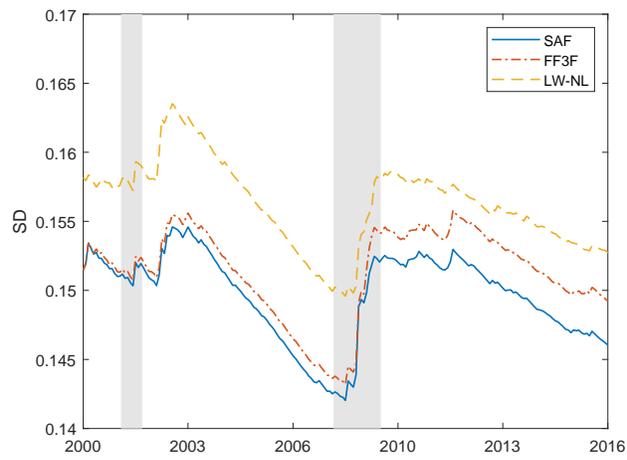}}
		\end{minipage}\hspace*{1cm}
		\begin{minipage}{.5\textwidth}
			\centering
			\captionsetup{width=1\linewidth}
			\subfloat[$N = 200$]{\includegraphics[width=1\linewidth]{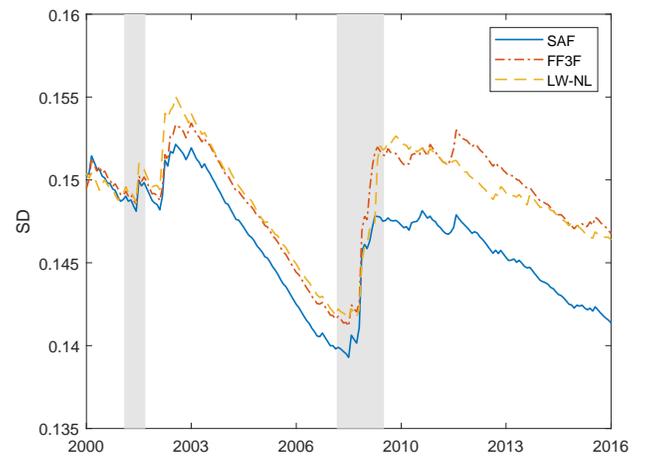}}
		\end{minipage}
		\caption{SD for different subperiods}\label{sd_sub}
	\end{figure}
\end{appendix}
\newpage
\setcounter{page}{1}
\section*{Supplement to "Sparse Approximate Factor Estimation for High-Dimensional Covariance Matrices"}
This supplement provides the omitted proofs and additional related results. Furthermore, it discusses the implementations issues, the choice of the number of factors and the selection of the tuning parameter $\mu$.
\vspace*{-0.5cm}
\setcounter{section}{0}
\setcounter{equation}{0}
\renewcommand\theequation{S.\arabic{equation}}
\renewcommand\thesection{S.\arabic{section}}
\section{Proofs}
\label{sec:A_proofs}
\subsection{Consistency of the Sparse Approximate Factor Model Estimator} \label{subsubsec:consistency_lam}
\textbf{Proof.} Theorem \ref{theo_consistency_lam} (Consistency of the Sparse Approximate Factor Model Estimator)

Define the penalized log-likelihood
\begin{align}
	\mathcal{L}_p(\bLam, \bSigma_{u}) = Q_1(\bLam, \bSigma_{u}) + Q_2(\bLam, \bSigma_{u}) + Q_3(\bLam, \bSigma_{u}), \label{const_lik}
\end{align}
where
\begin{align*}
	Q_1(\bLam, \bSigma_{u}) &= \frac{1}{N} \log\left|\bSigma_{u}\right| + \frac{1}{N} \trace{\bS_{u} \bSigma_{u}^{-1}} - \frac{1}{N} \log\left|\bSigma_{u 0}\right| - \frac{1}{N} \trace{\bS_u \bSigma_{u 0}^{-1}}\\
	&\quad + \frac{1}{N} \mu \sum_{k = 1}^{r} \sum_{i = 1}^{N} \left|\lambda_{ik}\right| - \frac{1}{N} \mu \sum_{k = 1}^{r} \sum_{i = 1}^{N} \left|\lambda_{ik0}\right|,	\\
	Q_2(\bLam, \bSigma_{u}) &= \frac{1}{N} \traces{\left(\bLam - \bLam_0\right)' \bSigma_{u}^{-1} \left(\bLam - \bLam_0\right) - \left(\bLam - \bLam_0\right)' \bSigma_{u}^{-1} \bLam \left(\bLam'\bSigma_{u}^{-1} \bLam\right)^{-1} \bLam'\bSigma_{u}^{-1}\left(\bLam - \bLam_0\right)},	\\
	Q_3(\bLam, \bSigma_{u}) &= \frac{1}{N} \log\left|\bLam\bLam' + \bSigma_{u}\right| + \frac{1}{N} \trace{\bS_x\left(\bLam\bLam' + \bSigma_{u}\right)^{-1}} - Q_2(\bLam, \bSigma_{u}) \\
	&\quad - \frac{1}{N} \log\left|\bSigma_{u}\right| - \frac{1}{N} \trace{\bS_u \bSigma_{u}^{-1}}.
\end{align*}
Therefore, we can see that equation \eqref{const_lik} can be written as
\begin{align}
	\begin{split}
		\mathcal{L}_p(\bLam, \bSigma_{u}) &= \frac{1}{N} \log\left|\bLam\bLam' + \bSigma_{u}\right| + \frac{1}{N} \trace{\bS_x\left(\bLam\bLam' + \bSigma_{u}\right)^{-1}} \\
		&\quad - \frac{1}{N} \log\left|\bSigma_{u 0}\right| - \frac{1}{N} \trace{\bS_{u} \bSigma_{u 0}^{-1}}\\
		&\quad + \frac{1}{N} \mu \sum_{k = 1}^{r} \sum_{i = 1}^{N} \left|\lambda_{ik}\right| - \frac{1}{N} \mu \sum_{k = 1}^{r} \sum_{i = 1}^{N} \left|\lambda_{ik0}\right|. \label{const_lik2}
	\end{split}
\end{align}
Define the set,
\begin{align*}
	\Psi_\delta = \left\{ \left(\bLam, \bSigma_{u}\right):\;\right. &\delta^{-1} < \pi_{\min}\left(\frac{\bLam'\bLam}{N^{\beta}}\right) \leq \pi_{\max}\left(\frac{\bLam'\bLam}{N^{\beta}}\right) < \delta, \\
	&\left. \delta^{-1} < \pi_{\min}\left(\bSigma_{u}\right) \leq \pi_{\max}\left(\bSigma_{u}\right) < \delta \right\}, \quad \text{for } 1/2 \leq \beta \leq 1.
\end{align*}
Further, $\bPhi_u = \text{diag}\left(\bSigma_{u}\right)$ and denotes a covariance matrix that contains only the elements of the main diagonal of $\bSigma_{u}$.\\
We impose the following sparsity assumptions on $\bLam$ and $\bSigma_{u}$:
\begin{align*}
	L_N &= \sum_{k = 1}^{r} \sum_{i = 1}^{N}\1\left\{\lambda_{ik} \neq 0 \right\} = \mathcal{O}\left(N\right),	\\
	S_N &= \max_{i \leq N} \sum_{j = 1}^{N} \1\left\{\sigma_{u,ij} \neq 0 \right\}, S_N^2 d_T = o(1) \text{ and } S_N \mu = o(1),
\end{align*} 
where $\1\{\cdot\}$ defines an indicator function that is equal to one if the boolean argument in braces is true, $d_T = \frac{\log N^{\beta}}{N} + \frac{1}{N^{\beta}}\frac{\log N}{T}$ and $\mu$ denotes the regularization parameter.

We introduce a lemma that will be necessary for the forthcoming derivations.\\

\begin{lemma}\label{lemma_fan}
	\leavevmode
	\begin{itemize}
		\item [(i)] $\max_{i,j\leq N} \left|\frac{1}{T} \sum_{t = 1}^{T} u_{it}u_{jt} - \E{u_{it}u_{jt}}\right| = \mathcal{O}_p\left(\sqrt{(\log N) / T}\right)$,
		\item[(ii)] $\max_{i \leq r, j\leq N} \left|\frac{1}{T} \sum_{t = 1}^{T} f_{it}u_{jt}\right| = \mathcal{O}_p\left(\sqrt{(\log N) / T}\right)$.
	\end{itemize}
\end{lemma}
\begin{proof}
	See Lemmas A.3 and B.1 in \citesupp{Fan2011a}.
\end{proof}
\begin{lemma} \label{lemma1}
	\begin{align*}
		\sup_{\left(\bLam, \bSigma_{u}\right) \in \Psi_\delta} \left|Q_3(\bLam, \bSigma_{u})\right| = \mathcal{O}_p\left(\frac{\log N^{\beta}}{N} + \frac{1}{N^{\beta}}\frac{\log N}{T}\right).
	\end{align*}
\end{lemma}
\begin{proof}
	The unpenalized log-likelihood
	\begin{align}
		\mathcal{L}(\bLam, \bSigma_{u}) = \frac{1}{N} \log\left|\bLam\bLam' + \bSigma_{u}\right| + \frac{1}{N} \trace{\bS_x\left(\bLam\bLam' + \bSigma_{u}\right)^{-1}}, \label{unp_lik}
	\end{align}
	can be decomposed in a similar fashion as in \textit{Lemma A.2.} in \citesupp{Bai2016}.
	
	The first term in equation \eqref{unp_lik} can be written as:
	\begin{align*}
		\frac{1}{N} \log\left|\bLam\bLam' + \bSigma_{u}\right| = \frac{1}{N}\log\left|\bSigma_{u}\right| + \frac{1}{N} \log\left|\bI_r + \bLam'\bSigma_{u}^{-1}\bLam\right|.
	\end{align*}
	Hence, we have
	\begin{align}\label{det_f}
		\frac{1}{N} \log\left|\bLam\bLam' + \bSigma_{u}\right| = \frac{1}{N} \log\left|\bSigma_{u}\right| + \mathcal{O}\left(\frac{\log N^{\beta}}{N}\right).
	\end{align}
	Now, we consider the second term $\frac{1}{N} \trace{\bS_x\left(\bLam\bLam' + \bSigma_{u}\right)^{-1}}$. Hereby, $\bS_x$ is defined as:
	\begin{align*}
		\bS_x = \frac{1}{T} \sum_{t = 1}^{T} \bx_t \bx_t' = \bLam_0 \bLam_0' + \bS_u + \bLam_0\frac{1}{T}\sum_{t = 1}^{T} \bfa_t\bu_t' + \left(\bLam_0\frac{1}{T}\sum_{t = 1}^{T} \bfa_t\bu_t'\right)',
	\end{align*}
	where $\bS_u = \frac{1}{T} \sum_{t = 1}^{T} \bu_t\bu_t'$ and the identification condition $\frac{1}{T} \sum_{t = 1}^{T} \bfa_t\bfa_t' = \bI_r$ is used.
	
	By the matrix inversion formula we have:
	\begin{align*}
		\left(\bLam \bLam' + \bSigma_{u}\right)^{-1} = \bSigma_{u}^{-1} - \bSigma_{u}^{-1}\bLam\left(\bI_r + \bLam'\bSigma_{u}^{-1} \bLam\right)^{-1} \bLam'\bSigma_{u}^{-1}.
	\end{align*}
	Hence, we get:
	\begin{align}\label{t_r}
		\begin{split}
			\frac{1}{N} \trace{\bS_x\left(\bLam\bLam' + \bSigma_{u}\right)^{-1}} &= \frac{1}{N} \trace{\bLam_0' \bSigma_{u}^{-1} \bLam_0} + \frac{1}{N} \trace{\bS_u\bSigma_{u}^{-1}}\\
			&-A_1+A_2+A_3-A_4,
		\end{split}
	\end{align}
	where $A_1 = \frac{1}{N} \trace{\bLam_0\bLam_0'\bSigma_{u}^{-1}\bLam\left(\bI_r + \bLam'\bSigma_{u}^{-1} \bLam\right)^{-1} \bLam'\bSigma_{u}^{-1}}$,\\
	$A_2 = \frac{1}{N} \trace{\frac{1}{T} \sum_{t = 1}^{T}\bLam_0\bfa_t\bu_t'\left(\bLam\bLam'+\bSigma_{u}\right)^{-1}},$
	$A_3 = \frac{1}{N} \trace{\frac{1}{T} \sum_{t = 1}^{T}\bu_t\bfa_t'\bLam_0'\left(\bLam\bLam'+\bSigma_{u}\right)^{-1}}$ \\and
	$A_4 = \frac{1}{N} \trace{\bS_u\bSigma_{u}^{-1}\bLam\left(\bI_r + \bLam'\bSigma_{u}^{-1} \bLam\right)^{-1} \bLam'\bSigma_{u}^{-1}}$.\\
	
	Subsequently, we look at the terms $A_1 - A_4$, respectively.\\
	Since $\pi_{\max}\left(\bSigma_{u}\right)$ and $\pi_{\min}^{-1}\left(\bLam'\bLam\right)$ are bounded from above uniformly in $\Psi_\delta$, we can derive the following expressions similarly as in \citesupp{Bai2016}:
	\begin{align}
		\sup_{\left(\bLam, \bSigma_{u}\right)\in \Psi_\delta} \pi_{\max}\left[\left(\bLam'\bSigma_{u}^{-1}\bLam\right)^{-1}\right] &\leq \sup_{\left(\bLam, \bSigma_{u}\right)\in \Psi_\delta} \frac{\pi_{\max}\left(\bSigma_{u}\right)}{\pi_{\min} \left(\bLam'\bLam\right)} = \mathcal{O}\left(N^{-\beta}\right), \label{eig1}\\
		\sup_{\left(\bLam, \bSigma_{u}\right)\in \Psi_\delta} \pi_{\max}\left[\left(\bI_r + \bLam'\bSigma_{u}^{-1}\bLam\right)^{-1}\right] &\leq \sup_{\left(\bLam, \bSigma_{u}\right)\in \Psi_\delta} \pi_{\max}\left[\left(\bLam'\bSigma_{u}^{-1}\bLam\right)^{-1}\right] = \mathcal{O}\left(N^{-\beta}\right). \label{eig2}
	\end{align}
	By applying the matrix inversion formula we have,
	\begin{align*}
		A_1 = &\frac{1}{N} \trace{\bLam_0'\bSigma_{u}^{-1}\bLam\left(\bLam'\bSigma_{u}^{-1}\bLam\right)^{-1}\bLam'\bSigma_{u}^{-1}\bLam_0}	\\
		&- \frac{1}{N} \trace{\bLam_0'\bSigma_{u}^{-1}\bLam\left(\bLam'\bSigma_{u}^{-1}\bLam\right)^{-1}\left(\bI_r + \bLam'\bSigma_{u}^{-1} \bLam\right)^{-1}\bLam'\bSigma_{u}^{-1}\bLam_0},
	\end{align*}
	where the second term can be bounded using \eqref{eig1} and \eqref{eig2}, by the following:
	\begin{align*}
		\frac{1}{N} &\trace{\bLam_0'\bSigma_{u}^{-1}\bLam\left(\bLam'\bSigma_{u}^{-1}\bLam\right)^{-1}\left(\bI_r + \bLam'\bSigma_{u}^{-1} \bLam\right)^{-1}\bLam'\bSigma_{u}^{-1}\bLam_0}	\\
		&\leq \frac{1}{N} \frob{\bLam_0'\bSigma_{u}^{-1}\bLam}^2 \pi_{\max}\left[\left( \bLam'\bSigma_{u}^{-1}\bLam\right)^{-1}\right] \pi_{\max}\left[\left(\bI_r + \bLam'\bSigma_{u}^{-1}\bLam\right)^{-1}\right]	\\
		&\leq r \spec{\bLam_0'\bSigma_{u}^{-1}\bLam}^2 \mathcal{O}\left(N^{-2\beta}\right) \mathcal{O}\left(\frac{1}{N}\right) = \mathcal{O}\left(\frac{1}{N}\right).
	\end{align*}
	Hence,
	\begin{align*}
		A_1 = \frac{1}{N} \trace{\bLam_0'\bSigma_{u}^{-1}\bLam\left(\bLam'\bSigma_{u}^{-1}\bLam\right)^{-1}\bLam'\bSigma_{u}^{-1}\bLam_0} + \mathcal{O}\left(\frac{1}{N}\right).
	\end{align*}
	In the following, we define $s_i(\bA)$ as the $i$-th singular value of a $(m\times n)$ matrix $\bA$. Furthermore, $s_{\max}(\bA)$ denotes the largest singular value of $\bA$. Using \lemref{lemma_fan}., Fact 9.14.3 and Fact 9.14.23 in \citesupp{Bernstein2009} and the fact that
	\begin{align*}
		\pi_{\max}\left[\left(\bLam\bLam'+\bSigma_{u}\right)^{-1}\right] \leq \pi_{\max}\left[\left(\bLam\bLam'\right)^{-1}\right] = \mathcal{O}\left(N^{-\beta}\right),
	\end{align*}
	we have:
	\begin{align*}
		\sup_{\left(\bLam, \bSigma_{u}\right)\in \Psi_\delta} |A_2|&\leq \frac{1}{N}\sum_{i = 1}^{N} s_i\left(\frac{1}{T}\sum_{t = 1}^{T} \bLam_0 \bfa_t \bu_t'\right) s_i\left(\left(\bLam\bLam' + \bSigma_{u}\right)^{-1}\right)	\\
		&\quad \leq \frac{1}{2N} \sum_{i = 1}^{r} s_i\left(\bLam_0'\bLam_0 + \frac{1}{T}\sum_{t = 1}^{T} \bfa_t \bu_t' \bu_t \bfa_t'\right)s_i\left(\left(\bLam\bLam' + \bSigma_{u}\right)^{-1}\right)	\\
		&\quad \leq \frac{r}{2N} s_{\max}\left(\bLam_0'\bLam_0 + \frac{1}{T}\sum_{t = 1}^{T} \bfa_t \bu_t' \bu_t \bfa_t'\right)s_{\max}\left(\left(\bLam\bLam' + \bSigma_{u}\right)^{-1}\right)	\\
		&\quad \leq \frac{r}{2N} \left[s_{\max}\left(\bLam_0'\bLam_0\right) + s_{\max}\left(\frac{1}{T}\sum_{t = 1}^{T} \bfa_t \bu_t' \bu_t \bfa_t'\right)\right]s_{\max}\left(\left(\bLam\bLam' + \bSigma_{u}\right)^{-1}\right)	\\
		&\quad = \frac{r}{2N} \left[\pi^{1/2}_{\max}\left(\bLam_0'\bLam_0\bLam_0'\bLam_0\right) + \pi^{1/2}_{\max}\left(\frac{1}{T}\sum_{t = 1}^{T} \bfa_t \bu_t' \bu_t \bfa_t'\bfa_t \bu_t' \bu_t \bfa_t'\right)\right]\\
		&\quad \quad \cdot \pi^{1/2}_{\max}\left(\left(\bLam\bLam' + \bSigma_{u}\right)^{-1}\left(\bLam\bLam' + \bSigma_{u}\right)^{-1}\right)	\\
		&\quad = \frac{r}{2N} \left(\spec{\bLam_0'\bLam_0} + \spec{\frac{1}{T}\sum_{t = 1}^{T} \bfa_t \bu_t' \bu_t \bfa_t'}\right)\spec{\left(\bLam\bLam' + \bSigma_{u}\right)^{-1}}\\
		&\quad \leq \frac{r}{2N} \left(\mathcal{O}(1) + \mathcal{O}\left(N^{-\beta}\right) \spec{\frac{1}{T}\sum_{t = 1}^{T} \bfa_t \bu_t'}^2\right)	\\
		&\quad \leq \frac{r}{2N} \left(\mathcal{O}(1) + \mathcal{O}\left(N^{-\beta}\right) N\cdot r\spec{\frac{1}{T}\sum_{t = 1}^{T} \bfa_t \bu_t'}^2_{\max}\right) = \mathcal{O}_p\left(\frac{1}{N} + \frac{1}{N^{\beta}} \frac{\log N}{T}\right).	\\
	\end{align*}
	Similarly, we have that $\sup_{\left(\bLam, \bSigma_{u}\right)\in \Psi_\delta} |A_3| = \mathcal{O}_p\left(\frac{1}{N} + \frac{1}{N^{\beta}} \frac{\log N}{T}\right)$.
	
	By the matrix inversion formula, we have for $A_4$ the following:
	\begin{align*}
		A_4 = &\frac{1}{N} \trace{\bS_u\bSigma_{u}^{-1}\bLam\left(\bLam'\bSigma_{u}^{-1} \bLam\right)^{-1} \bLam'\bSigma_{u}^{-1}} \\
		&- \frac{1}{N} \trace{\bS_u\bSigma_{u}^{-1}\bLam\left(\bLam'\bSigma_{u}^{-1} \bLam\right)^{-1}\left(\bI_r + \bLam'\bSigma_{u}^{-1} \bLam\right)^{-1} \bLam'\bSigma_{u}^{-1}}.
	\end{align*}
	From equations \eqref{eig1} and \eqref{eig2}, we see that the second term on the right hand side is uniformly of a smaller order than the first term. The first term of $A_4$ is bounded by:
	\begin{align*}
		A_4 &= \traces{\left(\bSigma_{u}^{-1} \bS_u \bSigma_{u}^{-1}\right)^{1/2} \bLam \left(\bLam'\bSigma_u^{-1}\bLam\right)^{-1} \bLam'\left(\bSigma_{u}^{-1} \bS_u \bSigma_{u}^{-1}\right)^{1/2}}\\
		&\quad \leq \traces{\bSigma_{u}^{-1} \bS_u \bSigma_{u}^{-1}} \pi_{\max}\left(\bLam \left(\bLam'\bSigma_u^{-1}\bLam\right)^{-1} \bLam'\right)	\\
		&\quad \leq \traces{\left(\bS_u\bSigma_{u}^{-1}\right)^{1/2} \bSigma_{u}^{-1}\left(\bS_u\bSigma_{u}^{-1}\right)^{1/2}} \mathcal{O}(1)\\
		&\quad \leq \trace{\bS_u\bSigma_{u}^{-1}} \mathcal{O}(1).
	\end{align*}
	Hence, we can bound the unpenalized log-likelihood function by:
	\begin{align*}
		\mathcal{L}(\bLam, \bSigma_{u}) &= \frac{1}{N} \trace{\bLam_0'\bSigma_{u}^{-1}\bLam_0} + \frac{1}{N}\trace{\bS_u\bSigma_{u}^{-1}} + \frac{1}{N} \log\left|\bSigma_{u}\right| \\
		&\quad - \frac{1}{N} \trace{\bLam_0' \bSigma_{u}^{-1} \bLam \left(\bLam'\bSigma_{u}^{-1}\bLam\right)^{-1}\bLam'\bSigma_{u}^{-1} \bLam_0} + \mathcal{O}_p\left(\frac{\log N^{\beta}}{N} + \frac{1}{N^{\beta}}\frac{\log N}{T}\right)	\\
		&= \frac{1}{N}\trace{\bS_u\bSigma_{u}^{-1}} + \frac{1}{N} \log\left|\bSigma_{u}\right| + \frac{1}{N} \traces{\left(\bLam - \bLam_0\right)' \bSigma_{u}^{-1} \left(\bLam - \bLam_0\right)}	\\
		&\quad - \frac{1}{N} \traces{\left(\bLam - \bLam_0\right)' \bSigma_{u}^{-1} \bLam \left(\bLam'\bSigma_{u}^{-1} \bLam\right)^{-1} \bLam'\bSigma_{u}^{-1}\left(\bLam - \bLam_0\right)} +\mathcal{O}_p\left(\frac{\log N^{\beta}}{N} + \frac{1}{N^{\beta}}\frac{\log N}{T}\right)\\
		&= \frac{1}{N} \log\left|\bSigma_{u}\right| + \frac{1}{N}\trace{\bS_u\bSigma_{u}^{-1}} +  Q_2(\bLam, \bSigma_{u}) + \mathcal{O}_p\left(\frac{\log N^{\beta}}{N} + \frac{1}{N^{\beta}}\frac{\log N}{T}\right).
	\end{align*}
	By the definition of $Q_3(\bLam, \bSigma_{u})$ we have:
	\begin{align*}
		\sup_{\left(\bLam, \bSigma_{u}\right)\in \Psi_\delta}\left|Q_3(\bLam, \bSigma_{u})\right| = \mathcal{O}_p\left(\frac{\log N^{\beta}}{N} + \frac{1}{N^{\beta}}\frac{\log N}{T}\right).
	\end{align*}
\end{proof}
\begin{lemma} For $d_T = \frac{\log N^{\beta}}{N} + \frac{1}{N^{\beta}}\frac{\log N}{T}$,
	\begin{align*}
		Q_1(\hat{\bLam}, \hat{\bSigma}_u) + Q_2(\hat{\bLam}, \hat{\bSigma}_u) = \mathcal{O}_p\left(d_T\right).
	\end{align*}
\end{lemma}
\begin{proof}
	If we consider equation \eqref{const_lik2} at the true parameter values, we get:
	\begin{align}
		\begin{split}
			\mathcal{L}_p(\bLam_0, \bSigma_{u 0}) &= \frac{1}{N} \log\left|\bLam_0\bLam_0' + \bSigma_{u 0}\right| + \frac{1}{N} \trace{\bS_x\left(\bLam_0\bLam_0' + \bSigma_{u 0}\right)^{-1}} \\
			&\quad - Q_2(\bLam_0, \bSigma_{u 0}) - \frac{1}{N} \log\left|\bSigma_{u 0}\right| - \frac{1}{N} \trace{\bS_u\bSigma_{u 0}^{-1}}\\
			&\quad + \frac{1}{N} \mu \sum_{k = 1}^{r} \sum_{i = 1}^{N} \left|\lambda_{ik0}\right| - \frac{1}{N} \mu \sum_{k = 1}^{r} \sum_{i = 1}^{N} \left|\lambda_{ik0}\right|\\
			&= Q_3(\bLam_0, \bSigma_{u 0}). \label{const_lik_true}
		\end{split}
	\end{align}
	Hence, by \eqref{const_lik} and \eqref{const_lik_true}, we have:
	\begin{align*}
		Q_1(\hat{\bLam}, \hat{\bSigma}_{u}) + Q_2(\hat{\bLam}, \hat{\bSigma}_{u}) &= \mathcal{L}_p(\hat{\bLam}, \hat{\bSigma}_{u}) -  Q_3(\hat{\bLam}, \hat{\bSigma}_{u})	\\
		& \leq \mathcal{L}_p(\bLam_0, \bSigma_{u 0})- Q_3(\hat{\bLam}, \hat{\bSigma}_{u})\\
		&= Q_3(\bLam_0, \bSigma_{u 0}) - Q_3(\hat{\bLam}, \hat{\bSigma}_{u}) \\
		&= 2 \sup \left|Q_3(\bLam, \bSigma_{u})\right|.
	\end{align*}
	Therefore, by \lemref{lemma1}. we have:
	\begin{align}
		Q_1(\hat{\bLam}, \hat{\bSigma}_{u}) + Q_2(\hat{\bLam}, \hat{\bSigma}_{u}) &\leq d_T. \label{q_1_q2}
	\end{align}
\end{proof}

\begin{lemma}\label{lemma_q2}
	\begin{align*}
		\frac{1}{N} \frob{\hat{\bPhi}_{u} -\bPhi_{u 0}}^2 = \mathcal{O}_p\left(\frac{\log N}{T} + d_T\right) = o_p(1).
	\end{align*}	
\end{lemma}
\begin{proof}
	By equation \eqref{q_1_q2} and the definition of $Q_1(\hat{\bLam}, \hat{\bSigma}_{u})$ and $Q_2(\hat{\bLam}, \hat{\bSigma}_{u})$, we get:
	\begin{align}
		B_1 + B_2 \leq d_T,\label{a1a2}
	\end{align}
	where $B_1$ and $B_2$ are defined as:
	\begin{align*}
		B_1 &= \frac{1}{N} \log\left|\hat{\bSigma}_{u}\right| + \frac{1}{N} \trace{\bS_u \hat{\bSigma}_{u}^{-1}} - \frac{1}{N} \log\left|\bSigma_{u 0}\right| - \frac{1}{N} \trace{\bS_u \bSigma_{u 0}^{-1}},\\
		B_2 &= \frac{1}{N} \traces{\left(\hat{\bLam} - \bLam_0\right)' \hat{\bSigma}_{u}^{-1} \left(\hat{\bLam} - \bLam_0\right) - \left(\hat{\bLam} - \bLam_0\right)' \hat{\bSigma}_{u}^{-1} \hat{\bLam} \left(\hat{\bLam}'\hat{\bSigma}_{u}^{-1} \hat{\bLam}\right)^{-1} \hat{\bLam}'\hat{\bSigma}_{u}^{-1}\left(\hat{\bLam} - \bLam_0\right)}	\\
		&\quad + \frac{1}{N} \mu \sum_{k = 1}^{r} \sum_{i = 1}^{N} \left|\hat{\lambda}_{ik}\right| - \frac{1}{N} \mu \sum_{k = 1}^{r} \sum_{i = 1}^{N} \left|\lambda_{ik0}\right|.	\\
	\end{align*}
	By equation \eqref{a1a2}, we can see that
	\begin{align*}
		\frac{1}{N} \log\left|\hat{\bSigma}_{u}\right| + \frac{1}{N} \trace{\bS_u \hat{\bSigma}_{u}^{-1}} - \frac{1}{N} \log\left|\bSigma_{u 0}\right| - \frac{1}{N} \trace{\bS_u \bSigma_{u 0}^{-1}} \leq d_T
	\end{align*}
	and
	\begin{align}
		\frac{1}{N} \log\left|\hat{\bPhi}_{u}\right| + \frac{1}{N} \trace{\bS_u \hat{\bPhi}_{u}^{-1}} - \frac{1}{N} \log\left|\bPhi_{u 0}\right| - \frac{1}{N} \trace{\bS_u \bPhi_{u 0}^{-1}} \leq d_T, \label{a1}
	\end{align}
	where $\bPhi_u = \text{diag}\left(\bSigma_{u}\right)$ and denotes a covariance matrix that contains only the elements of the main diagonal of $\bSigma_{u}$. Using the same argument as in the proof of \textit{Lemma B.1.} in \citesupp{Bai2016}, we get:
	\begin{align*}
		c\frob{\hat{\bPhi}_{u}^{-1} - \bPhi_{u 0}^{-1}}^2 - \mathcal{O}_p\left(\sqrt{\frac{\log N}{T}}\right)\sum_{ij}\left|\phi_{u0, ij} - \hat{\phi}_{u,ij} \right| &\leq Nd_T	\\
		c\frob{\hat{\bPhi}_{u}^{-1} - \bPhi_{u 0}^{-1}}^2 - \mathcal{O}_p\left(\sqrt{\frac{\log N}{T}}\right) \sqrt{N} \frob{\hat{\bPhi}_{u} - \bPhi_{u 0}} &\leq Nd_T	\\
		c\frob{\hat{\bPhi}_{u}^{-1} - \bPhi_{u 0}^{-1}}^2 - \mathcal{O}_p\left(\sqrt{\frac{\log N}{T}}\right) \sqrt{N} \frob{\hat{\bPhi}_{u}\left(\hat{\bPhi}_{u}^{-1} - \bPhi_{u 0}^{-1}\right) \bPhi_{u 0}} &\leq Nd_T	\\
		c\frob{\hat{\bPhi}_{u}^{-1} - \bPhi_{u 0}^{-1}}^2 - \mathcal{O}_p\left(\sqrt{\frac{\log N}{T}}\right) \sqrt{N} \spec{\hat{\bPhi}_{u}} \spec{\bPhi_{u 0}} \frob{\hat{\bPhi}_{u}^{-1} - \bPhi_{u 0}^{-1}} &\leq Nd_T
	\end{align*}
	Solving for $\frob{\hat{\bPhi}_{u}^{-1} - \bPhi_{u 0}^{-1}}$ yields:
	\begin{align*}
		\frob{\hat{\bPhi}_{u}^{-1} - \bPhi_{u 0}^{-1}} &= \mathcal{O}_p\left(\sqrt{\frac{N\log N}{T}} + \sqrt{N d_T}\right),	\\
		\frac{1}{N}\frob{\hat{\bPhi}_{u}^{-1} - \bPhi_{u 0}^{-1}}^2 &= \mathcal{O}_p\left(\frac{\log N}{T} + d_T\right) = o_p(1).
	\end{align*}
	Hence, we can conclude the proof by the following derivation:
	\begin{align*}
		\frac{1}{N}\frob{\hat{\bPhi}_{u} - \bPhi_{u 0}}^2 &=  \frac{1}{N}\frob{\hat{\bPhi}_{u}\left(\bPhi_{u 0}^{-1} - \hat{\bPhi}_{u}^{-1}\right)\bPhi_{u 0}}^2	\\
		&\leq \frac{1}{N}\spec{\hat{\bPhi}_{u}}^2\spec{\bPhi_{u 0}}^2 \frob{\hat{\bPhi}_{u}^{-1} - \bPhi_{u 0}^{-1}}^2.
	\end{align*}
\end{proof}

In the following, we establish the consistency of the factor loadings estimator. Initially, we bound the first part of $B_2$ defined in equation \eqref{a1a2}.

\begin{lemma}\label{lem_b2_1}
	\begin{align*}
		&\frac{1}{N} \traces{\left(\hat{\bLam} - \bLam_0\right)' \hat{\bSigma}_{u}^{-1} \left(\hat{\bLam} - \bLam_0\right) - \left(\hat{\bLam} - \bLam_0\right)' \hat{\bSigma}_{u}^{-1} \hat{\bLam} \left(\hat{\bLam}'\hat{\bSigma}_{u}^{-1} \hat{\bLam}\right)^{-1} \hat{\bLam}'\hat{\bSigma}_{u}^{-1}\left(\hat{\bLam} - \bLam_0\right)}\\
		&\geq \mathcal{O}_p\left(\frac{L_N}{N}\right) \max_{i \leq N} \spec{\hat{\blam}_i - \blam_{i0}}^2.
	\end{align*}
\end{lemma}
\begin{proof}	
	\begin{align*}
		&\frac{1}{N} \traces{\left(\hat{\bLam} - \bLam_0\right)' \hat{\bSigma}_{u}^{-1} \left(\hat{\bLam} - \bLam_0\right) - \left(\hat{\bLam} - \bLam_0\right)' \hat{\bSigma}_{u}^{-1} \hat{\bLam} \left(\hat{\bLam}'\hat{\bSigma}_{u}^{-1} \hat{\bLam}\right)^{-1} \hat{\bLam}'\hat{\bSigma}_{u}^{-1}\left(\hat{\bLam} - \bLam_0\right)}	\notag\\
		&\geq \frac{1}{N} \traces{\left(\hat{\bLam} - \bLam_0\right)' \left(\hat{\bLam} - \bLam_0\right)} \pi_{\min}\left(\hat{\bSigma}_{u}^{-1}\right)\\
		&\quad - \frac{1}{N} \traces{\left(\hat{\bLam} - \bLam_0\right)' \left(\hat{\bLam}- \bLam_0\right)} \pi_{\max}\left(\hat{\bSigma}_{u}^{-1} \hat{\bLam} \left(\hat{\bLam}'\hat{\bSigma}_{u}^{-1} \hat{\bLam}\right)^{-1} \hat{\bLam}'\hat{\bSigma}_{u}^{-1}\right)	\\
		&\geq \left[\mathcal{O}_p\left(\frac{1}{N}\right) + \mathcal{O}_p\left(\frac{L_N}{N}\right)\right]\max_{i \leq N} \spec{\hat{\blam}_i - \blam_{i0}}^2 = \mathcal{O}_p\left(\frac{L_N}{N}\right) \max_{i \leq N} \spec{\hat{\blam}_i - \blam_{i0}}^2.
	\end{align*}	
\end{proof}

The consistency result for $\hat{\bLam}$ is summarized in the following lemma.

\begin{lemma}\label{lem_est_load}
	\begin{align*}
		\max_{i \leq N} \spec{\hat{\blam}_i - \blam_{i0}} = \mathcal{O}_p\left(\mu + \sqrt{\frac{N d_T}{L_N}}\right).
	\end{align*}
\end{lemma}
\begin{proof}
	If we consider equation \eqref{a1a2}, \lemref{lemma_q2}. and \lemref{lem_b2_1}., we have
	\begin{align*}
		\mathcal{O}_p\left(\frac{L_N}{N}\right)\max_{i \leq N} \spec{\hat{\blam}_i - \blam_{i0}}^2 + \frac{1}{N} \mu \sum_{k = 1}^{r} \sum_{i = 1}^{N} \left|\hat{\lambda}_{ik}\right| - \left|\lambda_{ik0}\right| &\leq d_T\\
		\mathcal{O}_p\left(\frac{L_N}{N}\right)\max_{i \leq N} \spec{\hat{\blam}_i - \blam_{i0}}^2 - \frac{1}{N} \mu \sum_{k = 1}^{r} \sum_{i = 1}^{N}  \left|\lambda_{ik0}\right| - \left|\hat{\lambda}_{ik}\right| &\leq d_T	\\
		\mathcal{O}_p\left(\frac{L_N}{N}\right)\max_{i \leq N} \spec{\hat{\blam}_i - \blam_{i0}}^2 - \frac{1}{N} \mu \sum_{k = 1}^{r} \sum_{i = 1}^{N} \left|\hat{\lambda}_{ik} - \lambda_{ik0}\right| &\leq d_T	\\
		\mathcal{O}_p\left(\frac{L_N}{N}\right)\max_{i \leq N} \spec{\hat{\blam}_i - \blam_{i0}}^2 - \mathcal{O}\left(\frac{L_N}{N}\right) \mu \max_{i \leq N} \sum_{k = 1}^{r}  \left|\hat{\lambda}_{ik} - \lambda_{ik0}\right| & \leq d_T \\
		\mathcal{O}_p\left(\frac{L_N}{N}\right)\max_{i \leq N} \spec{\hat{\blam}_i - \blam_{i0}}^2 - \mathcal{O}\left(\frac{L_N}{N}\right)\mu \sqrt{r} \sqrt{\max_{i \leq N} \spec{\hat{\blam}_i - \blam_{i0}}^2} & \leq d_T
	\end{align*} 
	Solving for $\max_{i \leq N} \spec{\hat{\blam}_i - \blam_{i0}}$ yields
	\begin{align*}
		\max_{i \leq N} \spec{\hat{\blam}_i - \blam_{i0}} &\leq \mu + \sqrt{\mu^2 + \mathcal{O}_p\left(\frac{N d_T}{L_N}\right)}	\\
		&\leq \mu + \mathcal{O}_p\left(\sqrt{\frac{N d_T}{L_N}}\right).
	\end{align*}
\end{proof}
\begin{lemma}\label{lem_est_factor}
	\begin{align*}
		\frac{1}{T} \sum_{t = 1}^{T} \spec{\hat{\bfa}_t - \bfa_t}^2 = o_p(1).
	\end{align*}
\end{lemma}
\begin{proof}
	By the definition of the factor estimator in equation \eqref{gls_factors} we have:
	\begin{align}\label{diff_f}
		\hat{\bfa}_t - \bfa_t = -\left(\hat{\bLam}'\hat{\bPhi}_u^{-1}\hat{\bLam}\right)^{-1}\hat{\bLam}'\hat{\bPhi}_u^{-1}\left(\hat{\bLam} - \bLam_0\right)\bfa_t + \left(\hat{\bLam}'\hat{\bPhi}_u^{-1}\hat{\bLam}\right)^{-1}\hat{\bLam}'\hat{\bPhi}_u^{-1}\bu_t.
	\end{align}
	As $L_N = \mathcal{O}\left(N^{\beta}\right)$, the first term on the right-hand side can be bounded by:
	\begin{align}
		&\mathcal{O}_p\left(N^{-\beta}\right) \sqrt{\sum_{i = 1}^{N} \spec{\left(\hat{\bLam}'\hat{\bPhi}_{u}^{-1}\right)_i \left(\hat{\blam}_i - \blam_{i0}\right)}^2}\spec{\bfa_t} \notag\\
		\quad & \leq \mathcal{O}_p\left(N^{-\beta}\right) \sqrt{\mathcal{O}_p\left(\sum_{i = 1}^{N}\spec{\hat{\blam}_i - \blam_{i0}}^2\right)} \notag\\
		& \leq \mathcal{O}_p\left(N^{-\beta}\right) \sqrt{\mathcal{O}_p\left(L_N\max_{i \leq N}\spec{\hat{\blam}_i - \blam_{i0}}^2\right)}\notag\\
		&= \mathcal{O}_p\left(\frac{\sqrt{L_N}}{N^{\beta}}\right) o_p(1) = o_p(1).\label{f1_f_hat}
	\end{align}
	Now, we are going to bound the second term on the right-hand side of \eqref{diff_f}. For this we first analyze the term $\hat{\bLam}'\hat{\bPhi}_u^{-1}\bu_t$.
	\begin{align*}
		&\mathcal{O}_p\left(N^{-\beta}\right) \frob{\left(\hat{\bLam}'\hat{\bPhi}_u^{-1} - \bLam_0'\bPhi_{u 0}^{-1}\right)\bu_t} \leq \\
		&\quad \mathcal{O}_p\left(N^{-\beta}\right) \frob{\left(\hat{\bLam} - \bLam_0\right)'\hat{\bPhi}_u^{-1}\bu_t} +\mathcal{O}_p\left(N^{-\beta}\right) \frob{\bLam_0'\left(\hat{\bPhi}_u^{-1} - \bPhi_{u 0}^{-1}\right)\bu_t}.
	\end{align*}
	Using \lemref{lem_est_load}., the first term can be bounded by:
	\begin{align}\label{f1}
		&\mathcal{O}_p\left(N^{-\beta}\right) \sqrt{\sum_{i = 1}^{N} \spec{\left(\hat{\blam}_i - \blam_{i0}\right)\left(\hat{\bPhi}_u^{-1}\bu_t\right)_i}^2} \notag	\\
		&\quad\leq \mathcal{O}_p\left(N^{-\beta}\right) \sqrt{L_N \max_{i \leq N} \spec{\hat{\blam}_i - \blam_{i0}}^2\mathcal{O}_p(1)} \notag\\
		&\quad = \mathcal{O}_p\left(\frac{\sqrt{L_N}}{N^{\beta}}\right)o_p(1) = o_p(1).
	\end{align}
	The second term can be bounded using \lemref{lemma_q2}. according to:
	\begin{align}\label{f2}
		\mathcal{O}_p\left(N^{-\beta}\right) \frob{\bLam_0'\left(\hat{\bPhi}_u^{-1} - \bPhi_{u 0}^{-1}\right)\bu_t} &= \mathcal{O}_p\left(N^{-\beta}\right) \sqrt{\sum_{i = 1}^{N} \spec{\left(\bLam_0'\bPhi_{u}^{-1}\right)_i \left(\phi_{iu0} - \hat{\phi}_{iu}\right)\left(\hat{\bPhi}_{u}^{-1}\bu_t\right)_i}^2} \notag\\
		&\leq \mathcal{O}_p\left(N^{-\beta}\right) \sqrt{\sum_{i = 1}^{N} \spec{\hat{\phi}_{iu} - \left(\phi_{iu0} \right)}^2\spec{\left(\bLam_0'\bPhi_{u}^{-1}\right)_i}^2\spec{\left(\hat{\bPhi}_{u}^{-1}\bu_t\right)_i}^2} \notag\\
		&= \mathcal{O}_p \left(\frac{\log N}{N^{\beta}} \frob{\hat{\bPhi}_u - \bPhi_{u 0}}\right) = o_p(1).
	\end{align}
	Hence, using \eqref{f1_f_hat}, \eqref{f1} and \eqref{f2} yields:
	\begin{align*}
		\spec{\hat{\bfa}_t - \bfa_t} &= \mathcal{O}_p\left(N^{-\beta}\right) \sum_{i = 1}^{N}\spec{\left(\bLam_0'\bPhi_{u 0}^{-1}\right)_i u_{it}} + o_p(1) = \mathcal{O}_p\left(N^{-\beta/2}\right) + o_p(1) = o_p(1).
	\end{align*}
\end{proof}

\begin{lemma}\label{lem_uhat}
	\begin{align*}
		\max_{i \leq N} \frac{1}{T} \sum_{t = 1}^{T} \left|\hat{u}_{it} - u_{it}\right|^2 = \mathcal{O}_p\left(\mu^2 + \frac{N d_T}{L_N}\right).
	\end{align*}
\end{lemma}
\begin{proof}
	Since $\hat{u}_{it} - u_{it} = \left(\hat{\blam}_i - \blam_{i0}\right)\hat{\bfa}_t' + \blam_{i0}\left(\hat{\bfa}_t - \bfa_t\right)'$, we have by \lemref{lem_est_load}. and \lemref{lem_est_factor}.:
	\begin{align*}
		\max_{i \leq N} \frac{1}{T} \sum_{t = 1}^{T} \left|\hat{u}_{it} - u_{it}\right|^2 &\leq 2 \max_{i \leq N} \spec{\hat{\blam}_i - \blam_{i0}}^2 \frac{1}{T} \sum_{t = 1}^{T} \spec{\hat{\bfa}_t}^2 + 2 \max_{i \leq N} \spec{\blam_{i0}}^2\frac{1}{T} \sum_{t = 1}^{T} \spec{\hat{\bfa}_t - \bfa_t}^2	\\
		&\leq \mathcal{O}_p\left(\max_{i \leq N} \spec{\hat{\blam}_i - \blam_{i0}}^2\right) + \mathcal{O}_p\left(\frac{1}{T} \sum_{t = 1}^{T} \spec{\hat{\bfa}_t - \bfa_t}^2\right)	\\
		&= \mathcal{O}_p\left(\mu^2 + \frac{N d_T}{L_N}\right).
	\end{align*}
\end{proof}
\begin{lemma}\label{lem_sig}
	\begin{align*}
		\max_{i,j\leq N} \left|\hat{\sigma}_{ij} - \sigma_{ij}\right| =  \mathcal{O}_p\left(\sqrt{\mu^2 + \frac{N d_T}{L_N}}\right),
	\end{align*}
	where $d_T = \frac{\log N^{\beta}}{N} + \frac{1}{N^{\beta}}\frac{\log N}{T}$.
\end{lemma}
\begin{proof}
	Based on \textit{Lemma A.3.(iii)} by \citesupp{Fan2011a} we have:
	\begin{align}\label{var_decom}
		\max_{i,j\leq N} \left|\hat{\sigma}_{ij} - \sigma_{ij}\right| \leq \max_{i,j\leq N} \left|\frac{1}{T} \sum_{t = 1}^{T}u_{it} u_{jt} - \sigma_{ij}\right| + \max_{i,j\leq N} \left|\frac{1}{T} \sum_{t = 1}^{T}\hat{u}_{it} \hat{u}_{jt} - u_{it} u_{jt}\right|,
	\end{align}
	where the authors show that the first term on the right-hand side is $\mathcal{O}_p\left(\sqrt{\frac{\log N}{T}}\right)$. Now we are going to analyze the second term on the right-hand side of equation \eqref{var_decom}. In \lemref{lem_uhat}. we have shown that $\max_{i \leq N} \frac{1}{T} \sum_{t = 1}^{T} \left|\hat{u}_{it} - u_{it}\right|^2 = o_p(1)$. Hence, the result follows from Lemma A.3.(ii) by \citesupp{Fan2011a}.\\
\end{proof}

\subsection{Rate of convergence for the idiosyncratic error covariance matrix estimator}\label{subsubsec:rate_conv_sigma_u}
In what follows, we are going to determine the convergence rate of the idiosyncratic error covariance matrix estimator based on soft-thresholding.
\begin{lemma}\label{lem_idio}
	\begin{align*}
		\spec{\hat{\bSigma}_u^{\tau} - \bSigma_u} = \mathcal{O}_p\left(S_N\sqrt{\mu^2 + \frac{N d_T}{L_N}}\right).
	\end{align*}
\end{lemma}
\begin{proof}
	The result follows from \lemref{lem_sig}. and Theorem A.1. of \citesupp{FanLiaoMincheva2013}.\\
\end{proof}
\subsection{Convergence Rates for the Covariance Matrix Estimator} \label{subsubsec:riskbounds}

\textbf{Proof:} Theorem \ref{theorem_cov} (Convergence Rates for the Covariance Matrix Estimator)
\begin{align}
	\bSigma &= \bLam_0\bLam_0' + \bSigma_{u0},	\label{true_cov}\\
	\hat{\bSigma}_\text{SAF} &= \hat{\bLam}\hat{\bLam}' + \hat{\bSigma}_{u}^\tau, \label{est_cov}
\end{align}
where $\hat{\bSigma}_{u}^\tau$ corresponds to the POET estimator of \citesupp{FanLiaoMincheva2013}.
Similar as in \citesupp{FanLiaoMincheva2013}, we consider the weighted quadratic norm introduced by \citesupp{Fan2008a} and which is defined as:
\begin{align*}
	\spec{\bA}_{\bSigma} = N^{-1/2} \frob{\bSigma^{-1/2} \bA \bSigma^{-1/2}}.
\end{align*}
\begin{lemma}\label{lem_sig_norm}
	\begin{align*}
		\frac{1}{N} \spec{\hat{\bSigma}_\text{SAF} - \bSigma}_{\bSigma}^2 = \mathcal{O}_p\left(\frac{L_N^2}{N^2}\left[\mu^4 + \left(\frac{N}{L_N}d_T\right)^2\right] + \left[\frac{N^{\beta}L_N}{N^2} + \frac{S_N^2}{N}\right]\left[\mu^2 + \frac{N}{L_N} d_T\right]\right).
	\end{align*}
\end{lemma}
\begin{proof}
	The weighted quadratic norm of the difference between the estimated covariance matrix $\hat{\bSigma}_\text{SAF}$ and the true one $\bSigma$ can be expressed as:
	\begin{align}
		\spec{\hat{\bSigma}_\text{SAF} - \bSigma}_{\bSigma}^2 \leq \spec{\hat{\bLam}\hat{\bLam}' - \bLam_0\bLam_0'}_{\bSigma}^2 + \spec{\hat{\bSigma}_{u}^\tau - \bSigma_{u0}}_{\bSigma}^2. \label{fan_norm}
	\end{align}
	If we consider $\bC = \hat{\bLam} - \bLam_0$ we can introduce the following definitions:
	\begin{align*}
		\bC \bC'             &= \hat{\bLam}\hat{\bLam}' - \hat{\bLam}\bLam_0' - \bLam_0 \hat{\bLam} ' + \bLam_0\bLam_0',	\\
		\bLam_0 \bC' &= \bLam_0 \hat{\bLam}' - \bLam_0 \bLam_0',	\\
		\bC \bLam_0' &= \hat{\bLam} \bLam_0' - \bLam_0\bLam_0'.
	\end{align*}
	Using the previous definitions, we can rewrite the first term in \eqref{fan_norm} in the following form
	\begin{align*}
		\spec{\hat{\bLam}\hat{\bLam}' - \bLam_0\bLam_0'}_{\bSigma}^2 &= \spec{\bC\bC' + \bLam_0\bC' + \bC \bLam_0'}_{\bSigma}^2	\\
		&\leq \spec{\bC\bC'}^2_{\bSigma} + \spec{\bC \bLam_0'}_{\bSigma}^2 + \spec{\bLam_0 \bC'}_{\bSigma}^2.
	\end{align*}
	Hence, equation \eqref{fan_norm} can be expressed as:
	\begin{align}
		\spec{\hat{\bSigma}_\text{SAF} - \bSigma}_{\bSigma}^2 \leq \spec{\bC\bC'}^2_{\bSigma} + \spec{\bC \bLam_0'}_{\bSigma}^2 + \spec{\bLam_0 \bC'}_{\bSigma}^2 + \spec{\hat{\bSigma}_{u}^\tau - \bSigma_{u}}_{\bSigma}^2. \label{sig_est_fan}
	\end{align}
	Now we analyze each term in \eqref{sig_est_fan} separately:
	\begin{align*}
		\spec{\bLam_0 \bC'}_{\bSigma}^2 &= N^{-1} \trace{\bSigma^{-1/2}\bLam_0 \bC'\bSigma^{-1/2}\bSigma^{-1/2}\bC\bLam_0'\bSigma^{-1/2}}	\\
		&= N^{-1} \trace{\bLam_0'\bSigma^{-1}\bLam_0 \bC'\bSigma^{-1}\bC}	\\
		&\leq N^{-1} \spec{\bLam_0'\bSigma^{-1}\bLam_0} \spec{\bSigma^{-1}}\frob{\bC}^2 = \mathcal{O}_p\left(\frac{N^{\beta}}{N}\frob{\bC}^2\right).
	\end{align*}
	Similarly, we get $\spec{\bC\bLam_0'}_{\bSigma}^2 = \mathcal{O}_p\left(\frac{N^{\beta}}{N}\frob{\bC}^2\right)$. Further,
	$\spec{\bC\bC'}_{\bSigma}^2 = \frac{1}{N}	\frob{\bC}^4$.\\
	Hence, by \lemref{lem_idio}. we get:
	\begin{align*}
		\spec{\hat{\bSigma}_\text{SAF} - \bSigma}_{\bSigma}^2 &= \mathcal{O}_p\left( \frac{1}{N}	\frob{\bC}^4 + \frac{N^{\beta}}{N} \frob{\bC}^2\right) + \mathcal{O}_p\left(\spec{\hat{\bSigma}_{u}^\tau - \bSigma_{u}}_{\bSigma}^2\right)\\
		&= \mathcal{O}_p\left(\frac{L_N^2}{N}\left[\mu^4 + \left(\frac{N}{L_N}d_T\right)^2\right] + \frac{N^{\beta}L_N}{N} \left[\mu^2 + \frac{N}{L_N} d_T\right]\right) + \mathcal{O}_p\left(S_N^2\left[\mu^2+\frac{N}{L_N}d_T\right]\right)	\\
		&= \mathcal{O}_p\left(\frac{L_N^2}{N}\left[\mu^4 + \left(\frac{N}{L_N}d_T\right)^2\right] + \left[\frac{N^{\beta}L_N}{N} + S_N^2\right]\left[\mu^2 + \frac{N}{L_N} d_T\right]\right).
	\end{align*}	
\end{proof}

Under the \textbf{Frobenius norm} we have:
\begin{lemma}
	\begin{align*}
		\frac{1}{N} \frob{\hat{\bSigma}_\text{SAF} - \bSigma}^2 = \mathcal{O}_p\left(\frac{L_N^2}{N} \left[\mu^2 + \frac{N}{L_N}d_T\right]^2 + \left[\frac{N^{\beta}L_N}{N} + S_N^2\right]\left[\mu^2 + \frac{N}{L_N}d_T\right]\right).
	\end{align*}
\end{lemma}
\begin{proof}
	A similar argument as in \lemref{lem_sig_norm} leads to:
	\begin{align}\label{fro_sig}
		\frob{\hat{\bSigma}_\text{SAF} - \bSigma}^2 \leq \frob{\bC\bC'}^2 + \frob{\bLam_0 \bC'}^2 + \frob{\bC \bLam_0'}^2 + \frob{\hat{\bSigma}_{u}^\tau - \bSigma_{u}}^2,
	\end{align}
	where the second term can be bounded by
	\begin{align*}
		\frob{\bLam_0 \bC'}^2 &= \trace{\bLam_0'\bLam_0\bC'\bC}	\\
		& \leq \spec{\bLam_0}^2 \frob{\bC}^2 = \mathcal{O}_p\left(N^{\beta}\frob{\bC}^2\right).
	\end{align*}
	Furthermore, the first term in \eqref{fro_sig} has the same upper bound. Hence, again by using Lemma \ref{lem_idio} we get:
	\begin{align*}
		\frob{\hat{\bSigma}_\text{SAF} - \bSigma}^2 &\leq \mathcal{O}_p\left(\frob{\bC}^4 + N^{\beta}\frob{\bC}^2\right) + \mathcal{O}_p\left(\frob{\hat{\bSigma}_{u}^\tau - \bSigma_{u}}^2\right)	\\
		&\leq \mathcal{O}_p\left(L_N^2 \left[\mu^2 + \frac{N}{L_N}d_T\right]^2 +  N^{\beta}L_N\left[\mu^2 + \frac{N}{L_N}d_T\right]\right) + \mathcal{O}_p\left(N\left[\mu^2 + \frac{N}{L_N}d_T\right]S_N^2\right)\\
		&= \mathcal{O}_p\left(L_N^2 \left[\mu^2 + \frac{N}{L_N}d_T\right]^2 + \left[N^{\beta}L_N + N S_N^2\right]\left[\mu^2 + \frac{N}{L_N}d_T\right]\right).
	\end{align*}
\end{proof}

\textbf{Inverse of the covariance matrix}\\

Define,
\begin{align*}
	\hat{\bG} &= \left(\bI_r + \hat{\bLam}'\left(\hat{\bSigma}_{u}^{\tau}\right)^{-1}\hat{\bLam}\right)^{-1},	\\
	\bG_0     &= \left(\bI_r + \bLam_0'\bSigma_{u 0}^{-1}\bLam_0\right)^{-1}.
\end{align*}
\begin{lemma}\label{bound}
	\leavevmode
	\begin{enumerate}[label=(\roman*)]
		\item$\spec{\hat{\bG}} = \mathcal{O}_p\left(N^{-\beta}\right)$,\label{bound1}
		\item $\frob{\hat{\bG}^{-1} - \bG_0^{-1}} = \mathcal{O}_p\left(N^{\beta}\left(N^{-\beta/2}\frob{\bC} + \frob{\left(\hat{\bSigma}_{u}^{\tau}\right)^{-1} - \bSigma_{u}^{-1}} \right)\right)$.\label{bound2}
	\end{enumerate}
\end{lemma}
\begin{proof}
	
	\textit{\ref{bound1}} \lemref{lem_idio}. implies $\spec{\left(\hat{\bSigma}_{u}^{\tau}\right)^{-1}} = \mathcal{O}_p(1)$. Then, by the definition of $\hat{\bG}$ we have:
	\begin{align*}
		\spec{\hat{\bG}} &\leq \spec{\left(\hat{\bLam}'\left(\hat{\bSigma}_{u}^{\tau}\right)^{-1} \hat{\bLam}\right)^{-1}} \\
		&\leq \frac{\pi_{\max}\left(\hat{\bSigma}_{u}^{\tau}\right)}{\pi_{\min}\left(\hat{\bLam}'\hat{\bLam}\right)} = \mathcal{O}_p\left(N^{-\beta}\right).
	\end{align*}
	\textit{\ref{bound2}} By the definition of $\hat{\bG}$ and $\bG_0$, we have:
	$\hat{\bG}^{-1} - \bG_0^{-1} = \hat{\bLam}'\left(\hat{\bSigma}_{u}^{\tau}\right)^{-1}\hat{\bLam} - \bLam_0'\bSigma_{u 0}^{-1}\bLam_0$. Hence, the previous quantitiy can be decomposed according to:
	\begin{align}\label{ginv}
		\hat{\bG}^{-1} - \bG_0^{-1} = \bC' \left(\hat{\bSigma}_{u}^{\tau}\right)^{-1} \hat{\bLam} + \bLam_0'\bSigma_{u0}^{-1}\bC + \bLam_0' \left(\left(\hat{\bSigma}_{u}^{\tau}\right)^{-1} - \bSigma_{u0}^{-1}\right) \hat{\bLam}.
	\end{align}
	If we bound all three terms on the right hand side of equation \eqref{ginv}, we get:
	\begin{align*}
		\frob{\hat{\bG}^{-1} - \bG_0^{-1}} &\leq \frob{\bC} \mathcal{O}_p\left(N^{\beta/2}\right) + \frob{\left(\hat{\bSigma}_{u}^{\tau}\right)^{-1} - \bSigma_{u0}^{-1}} \mathcal{O}_p\left(N^{\beta}\right)\\
		&= \mathcal{O}_p\left(N^{\beta}\left(N^{-\beta/2}\frob{\bC} + \frob{\left(\hat{\bSigma}_{u}^{\tau}\right)^{-1} - \bSigma_{u}^{-1}} \right)\right).
	\end{align*}
\end{proof}
\begin{lemma}
	\begin{align*}
		\frac{1}{N}\frob{\hat{\bSigma}_\text{SAF}^{-1} - \bSigma^{-1}}^2 = \mathcal{O}_p\left(\frac{L_N}{N^{\beta+1}}\left[\mu^2 + \frac{N}{L_N}d_T\right] + S_N^2\left[\mu^2 + \frac{N}{L_N} d_T\right]\right).
	\end{align*}
\end{lemma}
\begin{proof}
	Using the Sherman-Morrison-Woodbury inverse formula, we get
	\begin{align*}
		\frob{\hat{\bSigma}_\text{SAF}^{-1} - \bSigma^{-1}}^2 = \sum_{i = 1}^{6} L_i,
	\end{align*}
	where
	\begin{align*}
		L_1 &= \frob{\left(\hat{\bSigma}_{u}^{\tau}\right)^{-1} - \bSigma_{u 0}^{-1}}^2,	\\
		L_2 &= \frob{\left[\left(\hat{\bSigma}_{u}^{\tau}\right)^{-1} - \bSigma_{u 0}^{-1}\right] \hat{\bLam}\left[\bI_r + \hat{\bLam}'\left(\hat{\bSigma}_{u}^{\tau}\right)^{-1}\hat{\bLam}\right]^{-1}\hat{\bLam}'\left(\hat{\bSigma}_{u}^{\tau}\right)^{-1}}^2,	\\
		L_3 &= \frob{\left[\left(\hat{\bSigma}_{u}^{\tau}\right)^{-1} - \bSigma_{u 0}^{-1} \right]\hat{\bLam}\left[\bI_r + \hat{\bLam}'\left(\hat{\bSigma}_{u}^{\tau}\right)^{-1}\hat{\bLam}\right]^{-1}\hat{\bLam}'\bSigma_{u 0}^{-1}}^2,	\\
		L_4 &= \frob{\bSigma_{u 0}^{-1} \left(\hat{\bLam} - \bLam_0\right) \left[\bI_r + \hat{\bLam}'\left(\hat{\bSigma}_{u}^{\tau}\right)^{-1}\hat{\bLam}\right]^{-1}\hat{\bLam}'\bSigma_{u 0}^{-1}}^2,	\\
		L_5 &= \frob{\bSigma_{u 0}^{-1} \left(\hat{\bLam} - \bLam_0\right) \left[\bI_r + \hat{\bLam}'\left(\hat{\bSigma}_{u}^{\tau}\right)^{-1}\hat{\bLam}\right]^{-1}\bLam_0'\bSigma_{u 0}^{-1}}^2,	\\
		L_6 &= \frob{\bSigma_{u 0}^{-1} \bLam_0\left(\left[\bI_r + \hat{\bLam}'\left(\hat{\bSigma}_{u}^{\tau}\right)^{-1}\hat{\bLam}\right]^{-1} - \left[\bI_r + \bLam_0'\bSigma_{u}^{-1}\bLam_0\right]^{-1} \right)\bLam_0'\bSigma_{u 0}^{-1}}^2.	\\
	\end{align*}
	In the following, we bound each of the six terms, separately.
	\begin{align*}
		L_2 \leq \frob{\left(\hat{\bSigma}_{u}^{\tau}\right)^{-1} - \bSigma_{u 0}^{-1}}^2 \spec{\hat{\bLam}\hat{\bG}\hat{\bLam}'}^2 \spec{\left(\hat{\bSigma}_{u}^{\tau}\right)^{-1}}^2.
	\end{align*}
	By \lemref{bound}. \textit{\ref{bound1}} follows that $L_2 \leq \mathcal{O}_p(L_1)$. Similarly, $L_3$ is also $\mathcal{O}_p(L_1)$.\\
	Further,
	\begin{align*}
		L_4 \leq \spec{\bSigma_{u 0}^{-1}}^2 \frob{\bC}^2 \spec{\hat{\bG}}^2 \spec{ \hat{\bLam}'\bSigma_{u 0}^{-1}}^2.
	\end{align*}
	Hence, also by \lemref{bound}. \textit{\ref{bound1}}
	\begin{align*}
		L_4 \leq \frob{\bC}^2 \mathcal{O}_p\left(N^{-\beta}\right) = \mathcal{O}_p\left(\frob{\bC}^2N^{-\beta}\right).
	\end{align*}
	Similarly, $L_5 = \mathcal{O}_p\left(L_4\right)$. Finally,
	\begin{align*}
		L_6 \leq \spec{\bSigma_{u 0}^{-1} \bLam_0}^4 \frob{\hat{\bG} - \bG_0}^2.
	\end{align*}
	By \lemref{bound}. \textit{\ref{bound2}} we have,
	\begin{align*}
		L_6 &\leq \mathcal{O}_p\left(N^{2\beta}\right) \frob{\hat{\bG}\left(\bG_0^{-1} - \hat{\bG}^{-1}\right)\bG_0}^2	\\
		&\leq \mathcal{O}_p\left(N^{-2\beta}\right)\frob{\bG_0^{-1} - \hat{\bG}^{-1}}^2\\
		&= \mathcal{O}_p\left(N^{-2\beta}\right) \mathcal{O}_p\left(N^{2\beta}\left(N^{-\beta}\frob{\bC}^2 + \frob{\left(\hat{\bSigma}_{u}^{\tau}\right)^{-1} - \bSigma_{u}^{-1}}^2 \right)\right)\\
		&= \mathcal{O}_p\left(N^{-\beta}\frob{\bC}^2 + \frob{\left(\hat{\bSigma}_{u}^{\tau}\right)^{-1} - \bSigma_{u}^{-1}}^2 \right).
	\end{align*}
	Adding up the terms $L_1 - L_6$ gives
	\begin{align*}
		\frac{1}{N}\frob{\hat{\bSigma}_\text{SAF}^{-1} - \bSigma^{-1}}^2 = \mathcal{O}_p\left(\frac{L_N}{N^{\beta+1}}\left[\mu^2 + \frac{N}{L_N}d_T\right] + S_N^2\left[\mu^2 + \frac{N}{L_N} d_T\right]\right).
	\end{align*}
\end{proof}
\subsection{Isometry group for the $l_p$-norm} \label{subsec:unitary_permutation}
In this section, we provide a short demonstration that only a unitary generalized permutation matrix can be an isometry for the $l_p$-norm. 

Define the $l_p$ norm of a $(n \times n)$-dimensional matrix $\bA$ as:
\begin{align*}
	|\lVert \bA\lVert|_p = \left(\sum_{i = 1}^n \sum_{j = 1}^n |a_{ij}|^p\right)^{1/p}.
\end{align*}

Consider the isometry $\bP$, which is $(n \times n)$-dimensional for the $l_p$ norm, with \\$1\leq p < 2 < q \leq \infty$ and the $(n \times n)$-dimensional isometry $\bP^*$ for the $l_q$ norm. For each standard basis vector $e_j$, we have:
\begin{align*}
	|\lVert \bP e_j\lVert|_p = |\lVert e_j\lVert|_p = 1.
\end{align*}
Hence, $\sum_{i = 1}^n |p_{ij}|^p = 1$, for each $j = 1, \dots, n$. Further, we have $|p_{ij}| \leq 1$, for all $i,j = 1, \dots, n$ and $\sum_{i = 1}^n \sum_{j = 1}^n |p_{ij}|^p = n$. Similar, considering $\bP^*$ and the $l_q$ norm, yields \\
$\sum_{i = 1}^n \sum_{j = 1}^n |p_{ij}|^q = n$, as well. However, $|p_{ij}|^q \leq |p_{ij}|^p$ holds with equality if and only if $p_{ij} = 0, 1$ or $-1$. Hence, each column of $\bP$ contains only one non-zero element and it has unit modulus. Furthermore, the non-singularity of $\bP$ ensures that each row of $\bP$ has no more than one non-zero element.

\subsection{Proof of \propref{prop_eig_cov}} \label{subsec:gmv_weights}

\begin{proof}
	As presented in Section \ref{fm_cov_est}, the general equation of the covariance matrix estimator based on an approximate factor model is given by:
	\begin{align}
		\bSigma       &= \bLam\bLam' + \bSigma_{u}.\label{cov_factor}
	\end{align}
	Correspondingly, the precision matrix is given by the inverse of the two matrices on the right-hand side of \eqref{cov_factor}:
	\begin{align}
		\bSigma^{-1}  &= \bSigma_{u}^{-1} - \bSigma_{u}^{-1}\bLam\left(\bI_r + \bLam'\bSigma_{u}^{-1}\bLam\right)^{-1}\bLam'\bSigma_{u}^{-1}. \label{prec_fac}
	\end{align}
	Further, by the definition of the factor loadings matrix $\bLam$, the first part on the right-hand side of \eqref{cov_factor} can be expressed as:
	\begin{align*}
		\bLam\bLam' &= \left(\begin{array}{cccc}
			\sum_{k = 1}^r \lambda_{1k}^2 & 	& 	 &	\bC\\
			& \sum_{k = 1}^r \lambda_{2k}^2 & &	\\
			& & \ddots	&					\\
			\bC	& & &			\sum_{k = 1}^r \lambda_{Nk}^2	 
		\end{array}\right),
	\end{align*}
	where $\bC$ denotes the upper and lower diagonal block of the matrix $\bLam\bLam'$.
	
	Hence, the sum of the eigenvalues of $\bLam\bLam'$ is calculated as:
	\begin{align} 
		\sum_{k = 1}^r \pi_k\left(\bLam\bLam'\right) = \trace{\bLam\bLam'} = \sum_{i = 1}^{N} \sum_{k = 1}^r \lambda_{ik}^2.\label{eig_lam}
	\end{align}
	From equation \eqref{eig_lam}, we can clearly see that sparsity or zeros in the factor loadings matrix corresponds to shrinking the sum of the eigenvalues of the covariance of the common component.\\
	In the next step, we want to analyze the global minimum variance portfolio weights based on the estimate of the covariance matrix of our SAF model.
	
	Without loss of generality we assume that the idiosyncratic error covariance matrix is an identity matrix, which corresponds to a high penalization of the off-diagonal elements based on the POET method. Hence, the precision matrix in \eqref{prec_fac} simplifies to: 
	\begin{align}
		\bSigma^{-1} &= \bI_N - \bLam\left(\bI_r + \bLam'\bLam\right)^{-1}\bLam' \notag\\
		&= \bI_N - \left[\bI_N + \left(\bLam\bLam'\right)^{-1}\right]^{-1}. \label{prec_f}
	\end{align}
	In the following, we have a look at the eigenvalues of the precision matrix of our SAF estimator based on equation \eqref{prec_f}:
	\begin{align}
		\sum_{i = 1}^{N} \pi_i\left(\bSigma^{-1}\right) &= \sum_{i = 1}^{N} \pi_i\left(\bI_N\right) - \sum_{i = 1}^{N} \frac{1}{1 + 1/\pi_i\left(\bLam\bLam'\right)}	\notag \\
		&= \sum_{i = 1}^{N} \pi_i\left(\bI_N\right) - \sum_{i = 1}^{N} \frac{\pi_i\left(\bLam\bLam'\right)}{\pi_i\left(\bLam\bLam'\right) + 1}	\notag\\
		&\le \sum_{i = 1}^{N} \pi_i\left(\bI_N\right) - \frac{\sum_{i = 1}^{N}\pi_i\left(\bLam\bLam'\right)}{N + \sum_{i = 1}^{N}\pi_i\left(\bLam\bLam'\right)}. \label{eig_inv}
	\end{align}
	From equations \eqref{eig_lam} and \eqref{eig_inv}, we can see that the possible sparsity in $\bLam$ allowed by our SAF model shrinks the precision matrix based on the SAF model towards an identity matrix. As the GMVP weights directly depend on an estimate of the precision matrix this implies a shrinkage of the SAF portfolio weights towards the weights of the $1/N$-portfolio.
\end{proof}


\section{Implementation Issues}\label{sec:implem}
For the implementation of the SAF model, we use a two-step estimation procedure that treats $\bSigma_u$ in the first step as a diagonal matrix, denoted as $\bPhi_u$ and re-estimates the idiosyncratic error covariance matrix in a second step by the method introduced in Section \ref{error_factor}. Theorem \ref{theo_consistency_lam} shows that this two-step procedure yields consistent estimates for $\bLam$ and $\bSigma_u$.
\subsection{Majorized Log-Likelihood Function}
The numerical minimization of the objective function  \eqref{sparse_lasso} is cumbersome as it is not globally convex.   
This problem arises because the first term in \eqref{sparse_lasso} $\log\left|\det\left(\bLam\bLam' + \bPhi_{u}\right)\right|$ is concave in $\bLam$ and $\bPhi_{u}$, whereas the second term $\traces{\bS_{x}\left(\bLam \bLam' + \bPhi_{u}\right)^{-1}}$ is convex.
For our implementation we employ the majorize-minimize EM algorithm introduced by \citesupp{BienTibshiraniothers2011}.
%
The idea of this optimization approach is to approximate the numerically unstable concave part $\log\left|\det\left(\bLam\bLam' + \bPhi_{u}\right)\right|$ by its tangent plane, which corresponds to the following expression:
\begin{align}
	\log\left|\det\left(\hat{\bLam}_m\hat{\bLam}_m' + \hat{\bPhi}_{u, m}\right)\right| + \traces{2\hat{\bLam}_m'\left(\hat{\bLam}_m\hat{\bLam}_m' + \hat{\bPhi}_{u, m}\right)^{-1}\left(\bLam - \hat{\bLam}_m\right)} \label{tanget_plane_lik},
\end{align}
where the subscript $m$ denotes the $m$-th step in an iterative procedure outlined in Section \ref{impl:pga}. Replacing the concave part in \eqref{quasi_lik} by the convex expression in \eqref{tanget_plane_lik}, yields the following majorized log-likelihood function:
\begin{align}
	\begin{split}\label{major_lik}
		\bar{\mathcal{L}}_m(\bLam) = &\log\left|\det\left(\hat{\bLam}_m\hat{\bLam}_m' + \hat{\bPhi}_{u, m}\right)\right| + \traces{2\hat{\bLam}_m'\left(\hat{\bLam}_m\hat{\bLam}_m' + \hat{\bPhi}_{u, m}\right)^{-1}\left(\bLam - \hat{\bLam}_m\right)} 	\\
		& + \traces{\bS_{x}\left(\bLam\bLam' + \hat{\bPhi}_{u, m}\right)^{-1}}.
	\end{split}
\end{align}
Augmenting the majorized log-likelihood by the $l_1$-penalty term, leads to the following optimization problem for our SAF model: 
\begin{align}
	\begin{split}\label{major_p_lik}
		\min_{\bLam} &\left\{\log\left|\det\left(\hat{\bLam}_m\hat{\bLam}_m' + \hat{\bPhi}_{u, m}\right)\right| + \traces{2\hat{\bLam}_m'\left(\hat{\bLam}_m\hat{\bLam}_m' + \hat{\bPhi}_{u, m}\right)^{-1}\left(\bLam - \hat{\bLam}_m\right)}\right. 	\\
		& \left.+ \traces{\bS_{x}\left(\bLam\bLam' + \hat{\bPhi}_{u, m}\right)^{-1}} + \mu \sum_{k = 1}^{r} \sum_{i = 1}^{N} \left|\lambda_{ik}\right|\right\}.
	\end{split}
\end{align}
As all three components in \eqref{major_p_lik} are convex, the optimization problem simplifies considerably compared to the original problem in equation \eqref{sparse_lasso}.

\subsection{Projection Gradient Algorithm} \label{impl:pga}
In order to minimize \eqref{major_p_lik} efficiently, we apply 
the fast projected gradient algorithm proposed by \citesupp{BienTibshiraniothers2011}. More specifically, we approximate the majorized log-likelihood $\bar{\mathcal{L}}_m(\bLam)$ in \eqref{major_lik} by the following expression:
\begin{align*}
	\tilde{\mathcal{L}}(\bLam) = \frac{1}{2t}\left\Vert\bLam - \hat{\bLam}_m + t\hat{\bA}\right\Vert_F^2 \, ,
\end{align*}
where $t$ is the depth of projection\footnote{We set $t = 0.01$ for all our applications below.} and
\begin{align}
	\hat{\bA} = 2\left[\left(\hat{\bLam}_m\hat{\bLam}_m' + \hat{\bPhi}_{u, m}\right)^{-1} -\left(\hat{\bLam}_m\hat{\bLam}_m' + \hat{\bPhi}_{u, m}\right)^{-1}\bS_{x}\left(\hat{\bLam}_m\hat{\bLam}_m' + \hat{\bPhi}_{u, m}\right)^{-1}\right]\hat{\bLam}_m, \label{f_diff}
\end{align}
which corresponds to the first derivative of $\bar{\mathcal{L}}(\bLam)$ with respect to $\bLam$.
Hence, our final optimization problem corresponds to:
\begin{align}
	\underset{\lambda_{ik}}{\min} \; \frac{1}{2t} \sum_{k = 1}^{r} \sum_{i = 1}^{N} \left(\lambda_{ik} - \hat{\lambda}_{ik, m} + t \hat{A}_{ik, m}\right)^2 + \mu \sum_{k = 1}^{r} \sum_{i = 1}^{N} \left|\lambda_{ik}\right|.	\label{p_lasso_app}
\end{align}
%
The minimization of the objective function  \eqref{p_lasso_app} can be carried out by computing its gradient with respect to $\lambda_{ik}$ 
and setting it to zero, which yields:
\begin{align}
	\begin{split}\label{sec:impl:grad}
		\frac{\partial}{\partial \lambda_{ik}} \;&\left[\frac{1}{2t} \sum_{k = 1}^{r} \sum_{i = 1}^{N} \left(\lambda_{ik} - \hat{\lambda}_{ik, m} + t \hat{A}_{ik, m}\right)^2 + \mu \sum_{k = 1}^{r} \sum_{i = 1}^{N} \left|\lambda_{ik}\right|\right]	\\
		&= \frac{1}{t} \sum_{k = 1}^{r} \sum_{i = 1}^{N} \left(\hat{\lambda}_{ik} - \hat{\lambda}_{ik, m} + t \hat{A}_{ik, m}\right) + \mu \sum_{k = 1}^{r} \sum_{i = 1}^{N} \nu_{ik} = 0 \, ,
	\end{split}
\end{align}
where $\nu_{ik}$ denotes the subgradient of $\left|\lambda_{ik}\right|$. Solving \eqref{sec:impl:grad} for a specific $\hat{\lambda}_{ik}$, gives:
\begin{align}
	\hat{\lambda}_{ik} + t\cdot\mu \nu_{ik} &= \hat{\lambda}_{ik, m} - t\hat{A}_{ik, m}	\notag\\
	\hat{\lambda}_{ik} &= \mathcal{S}\left(\hat{\lambda}_{ik, m} - t\hat{A}_{ik, m}, \, t \cdot \mu\right),\label{lasso_sol}
\end{align}
where $\mathcal{S}$ denotes the soft-thresholding operator defined in equation \eqref{soft_t}.
Equation \eqref{lasso_sol} can be used to update the estimated factor loadings $\hat{\lambda}_{ik, m+1}$ given the estimate from the previous step $\hat{\lambda}_{ik, m}$.

In order to obtain an update for the estimate of the covariance matrix of the idiosyncratic errors $\bPhi_{u}$, we use the EM algorithm 
suggested by \citesupp{BaiLi2012}:
\begin{align*}
	\hat{\bPhi}_{u, m+1} = \text{diag}\left[\bS_{x} - \hat{\bLam}_{m+1} \hat{\bLam}_{m}' \left(\hat{\bLam}_{m}\hat{\bLam}_{m}' + \hat{\bPhi}_{u, m}\right)^{-1} \bS_{x}\right].
\end{align*}
Our iterative estimation procedure for the SAF model can be briefly summarized as given below.\\

\textit{Iterative Algorithm}
\begin{itemize}
	\item[\textit{Step 1:}] Obtain an initial consistent estimate for the factor loadings matrix $\bLam$ and for the diagonal idiosyncratic error covariance matrix $\bPhi_{u}$ , i.e. by using unpenalized MLE and set $m = 1$.
	\item[\textit{Step 2:}] Update $\hat{\lambda}_{ik,m-1}$, by $ \hat{\lambda}_{ik,m} = \mathcal{S}\left(\hat{\lambda}_{ik, m-1} - t \hat{A}_{ik, m-1}, \, t \cdot \mu\right) $.
	%
	\item[\textit{Step 3:}] Update $\hat{\bPhi}_{u}$ using the EM algorithm in \citesupp{BaiLi2012}, according to \\
	$ \hat{\bPhi}_{u, m} = \text{diag}\left[\bS_{x} - \hat{\bLam}_m \hat{\bLam}_{m-1}' \left(\hat{\bLam}_{m-1}\hat{\bLam}_{m-1}' + \hat{\bPhi}_{u, m-1}\right)^{-1} \bS_{x}\right] $.
	\item[\textit{Step 4:}] If $\spec{\hat{\bLam}_m - \hat{\bLam}_{m-1}}$ and $\spec{\hat{\bPhi}_{u,m} - \hat{\bPhi}_{u,m-1}}$ are sufficiently small, stop the procedure,  
	otherwise set $m = m + 1$ and return to \textit{Step 2}.
	\item[\textit{Step 5:}] Estimate the factors by $\hat{\bfa}_t = \left(\hat{\bLam}' \hat{\bPhi}_{u}^{-1} \hat{\bLam}\right)^{-1}\hat{\bLam}'\hat{\bPhi}_{u}^{-1}x_t$ ,
	where $\hat{\bLam}$ and $\hat{\bPhi}_{u}$ are the parameter estimates after convergence.
	\item[\textit{Step 6:}] Re-estimate the covariance matrix of the idiosyncratic errors based on the procedure introduced in Section \ref{error_factor}.
\end{itemize}
For the high dimensional case of $N > T$, the sample covariance matrix $\bS_x$ is not of full rank and hence leads to inconsistent parameter estimates. To overcome this problem, we adopt the solution proposed by \citesupp{BienTibshiraniothers2011}, who suggest augmenting the diagonal elements of $\bS_x$ by an arbitrarily small $\varepsilon > 0$, when $\bS_x$ is not of full rank. This augmentation stabilizes $\bS_x$ and yields a non-degenerate solution for our sparse factor model.

\subsection{Selecting the number of factors}
In order to select the number of latent factors $r$, we follow \citesupp{Onatski2010}.  
To the best of our knowledge, Onatski's method is the only one, which does not explicitly require that all factors are strong. Therefore, it is suitable for our setting, which  allows as well for weak factors.
The method uses the difference in subsequent eigenvalues and chooses the largest $\hat{r}$ such that:
\begin{align*}
	\{\hat{r} \leq r_{\max}: \pi_{\hat{r}}((\bX'\bX)/T) - \pi_{\hat{r}+1}((\bX'\bX)/T) > \xi\},
\end{align*}
where $\xi$ is a fixed positive constant, $r_{\max}$ is an upper bound for the possible number of factors and $\pi_{\hat{r}}((\bX'\bX)/T)$ denotes the $\hat{r}$-th largest eigenvalue of the covariance matrix of $\bX$.
For the choice of $\xi$, the empirical distribution of the eigenvalues of the data sample covariance matrix is taken into account.\footnote{We refer to  \citesupp{Onatski2010} for the detailed description of the determination of $\xi$.} However, the estimation of the number of factors based on the empirical distribution of the eigenvalues of the sample covariance matrix still requires a clear separation of the eigenvalues of the common and idiosyncratic component. Therefore, its selection accuracy depends on the degree of differentiability between the two components. Nevertheless, even if the selection method of \citesupp{Onatski2010} overestimates the true number of factors, the sparsity assumption in our setting would allow us to disentangle the informative factors from those that are too weak. Thus, compared to the standard approximate factor model we avoid including redundant factors that amplify the misspecification error. Moreover, to further support the above argument, we refer to \citesupp{YuSamworth2013}, who show that in the weak factor setting the true number of factors is not asymptotically overestimated.
%

\subsection{Choosing the tuning parameter}
As for any penalized estimation approach, the selection of the tuning parameter $\mu$ is crucial, as it controls the degree of sparsity in the factor loadings matrix and it affects the performance of our estimator. 
In our case we select $\mu$ based on a type of Bayesian information criterion, according to:
\begin{align}
	IC(\mu) = \mathcal{L}\left(\hat{\bLam}, \bS_{\hat{F}}, \hat{\bSigma}_{u}^{\tau}\right) + 2 \kappa_\mu \sqrt{\frac{\log N}{N} + \frac{\log N}{N\cdot T}},\label{alpha_crit}
\end{align}
where $\kappa_\mu$ denotes the number of non-zero elements in the factor loadings matrix $\hat{\bLam}$ for a given value of $\mu$ and $\mathcal{L}\left(\hat{\bLam}, \bS_{\hat{F}}, \hat{\bSigma}_{u}^{\tau}\right)$ is the value of the log-likelihood function in equation \eqref{quasi_lik_initial}, evaluated at the estimates of the factors, the factor loadings and the covariance matrix of the idiosyncratic errors. The penalty term in \eqref{alpha_crit} has the property of converging to zero as both $N$ and $T$ approach infinity. Hence, the penalization vanishes as the sample size increases and a smaller value for $\mu$ is selected. The characteristics of our information criterion are therefore convenient with respect to the asymptotic properties we require for the regularization parameter $\mu$. In fact, we need $\mu = o(1)$ in order to achieve estimation consistency, as elaborated in Section \ref{sec:theo}. The representation of the penalty term is based on the convergence rate of the factor loadings estimator in \lemref{lem_est_load}.

To select the optimal $\mu$, we estimate the criterion in \eqref{alpha_crit} for a grid of different values for $\mu$ and choose the one that minimizes our information criterion. For the grid of the shrinkage parameter we consider the interval $\mu = (0, \mu_{\text{max}})$, where $\mu_{\text{max}}$ denotes the highest value for the shrinkage parameter such that all imposed model restrictions are still fulfilled.
\section{Competing Approaches}
\label{sec:A_methods}
In this section, we summarize the estimation methods that are used in the simulation study and in the empirical application.\footnote{We also included in our extended  comparative study the approaches 
	by \citesupp{Frahm2010} and \citesupp{Pollak2011}, 
	which are based on direct shrinkage of the portfolio weights. The performance of these two models was clearly inferior, so that we refrained from giving the results here. However, they can be obtained from the authors upon request.}
\vspace*{0.3cm}
\begin{itemize}[leftmargin=*]
	\item \textit{Equally Weighted Portfolio ($1/N$)}
\end{itemize}
The equally weighted or $1/N$ portfolio strategy comprises identical portfolio weights of size $1/N$, for each of the risky assets. 
By ignoring any type of optimizing portfolio strategy it often serves as a benchmark case to be outperformed in empirical performance 
comparisons. As the weights have not to be estimated, the  $1/N$-strategy is free from any estimation risk. Moreover, the $1/N$ portfolio weights can be considered as the outcome for portfolio weights under extreme $l_2$-penalization. 
\citesupp{DeMiguel2009a} find that the mean-variance portfolio and most of its extensions cannot significantly outperform the $1/N$ portfolio.\footnote{\citesupp{Kazak/Pohlmeier2018} show, however, that conventional portfolio performance tests suffer from 
	very low power, so that the rejection of null hypothesis of equal performance of a given data-based strategy and the $1/N$-strategy is very unlikely. }
\vspace*{0.3cm}
\begin{itemize}[leftmargin=*, resume]
	\item \textit{Sample (GMVP)}
\end{itemize}
As the extreme alternative to the $1/N$-strategy, we consider the 
plug-in estimator of the GMVP based on the sample covariance matrix of the asset returns. The plug-in estimator is free from any type 
of regularization. The plug-in approach yields unbiased estimates of the true weights (\citesupp{Okhrin/Schmid2006}), but the weight estimates are extremely unstable 
when the asset space is large relative to the time series dimension. 
For some of our empirical designs with $N= 100, 200 $, the asset dimension exceeds the sample size, $T=60$.  For these cases the plug-in estimator is infeasible, because the sample covariance matrix is singular.
\vspace*{0.3cm}
\begin{itemize}[wide, labelwidth=!, labelindent=0pt]
	\setlength\itemsep{0.5cm}
	\item[\textbf{A.}]\textbf{\textit{Factor Models}}
	\begin{itemize}[leftmargin=0cm]
		\item[]\textbf{\textit{Factor models with latent factors}}
	\end{itemize}
	\begin{enumerate}[wide, labelwidth=!, labelindent=0pt]
		\setlength\itemsep{0.5cm}
		\item \textit{\citesupp{FanLiaoMincheva2013}} (POET) \label{gmv_factor}\\
		In our comparative study we include the POET estimator by \citesupp{FanLiaoMincheva2013} that is based on the standard approximate factor model with a dense factor loadings matrix and a sparse idiosyncratic error covariance matrix.
		Similar to SAF, we use the number of factors selected by \citesupp{Onatski2010}.
		\item \textit{\citesupp{Doz2011}}  (DFM) \label{gmv_dynf}\\
		To allow for some dynamics in the latent factors, we consider also a dynamic factor model originally proposed by \citesupp{Geweke1977}.
		Specifically, the dynamic factor model is represented by the following equation:
		\begin{align}
			x_{it} = \bB_i'(L) \bfa_t + \varepsilon_{it}, \label{dyn_factor_m}
		\end{align}
		where $\bB_i(L) = \left(b_{i1} + b_{i2}L + \cdots + b_{ip}L^p\right)$ and $L$ corresponds to the lag operator such that, $\forall p$, $L^p \bfa_t = \bfa_{t-p}$.
		In this setup $\bfa_t = \left(f_{1t}, f_{2t}, \dots, f_{qt}\right)'$ is a $(q \times 1)$-dimensional vector of dynamic factors 
		following  a VAR process and $b_{ij},  j = 1,\dots,p$ denote the corresponding $q$-dimensional factor loadings.
		In order to estimate the dynamic factor model in \eqref{dyn_factor_m}, we use the two step procedure of \citesupp{Doz2011}. The estimation 
		requires that the number of dynamic factors is given ex-ante. We use the consistent method by \citesupp{Bai2007} to determine $q$.
	\end{enumerate}
	\begin{itemize}[leftmargin=0cm]
		\item[]\textbf{\textit{Factor models with observable factors}}
	\end{itemize}
	In addition, we consider two factor models  that have been frequently used in the empirical finance literature. 
	Contrary to the approximate factor models, the factors in these models
	are not latent but observable time series variables. In this respect, these type of models incorporate more information than approaches, which solely use the information on the return process itself to estimate the covariance matrix of returns. 
	However, the inclusion of additional time series information may give rise to an additional  source of misspecification, if the factor specification fails to describe the true data generating process properly.
	\begin{enumerate}[leftmargin=*]
		\item \textit{The Single Index Model} (SIM) \label{sim_sharpe}
	\end{enumerate}
	The single index model by \citesupp{Sharpe1963} is based on a single observable factor, $f_{1t}$, representing the excess market return:
	\begin{align}
		x_{it} = \alpha + \beta_{i1} f_{1t} + \varepsilon_{it} \, .  \label{single_index}
	\end{align}
	In our study, we use as a proxy for the market return, the value-weighted returns of all Center for Research in Security Prices (CRSP) firms incorporated in the US and listed on the AMEX, NASDAQ, or the NYSE. 
	The one-month treasury bill rate serves as the risk free rate to construct the excess market returns. The estimator for the covariance matrix of the single index model is given by:
	\begin{align*}
		\hat{\bSigma}_{\text{SIM}} = \hat{\bbeta}_1 \hat{\sigma}_{f_{1}} \hat{\bbeta}_1' + \hat{\bD} \, ,
	\end{align*}
	where $\hat{\sigma}_{f_1}$ denotes the sample variance of the market excess returns.  $\hat{\beta}_1$ represents the OLS estimates of the factor loadings and $\hat{\bD}$ is a diagonal matrix of the OLS residual variances of regression model \eqref{single_index} assuming that the observed factor   picks up the cross-correlations of the returns completely.
	\vspace*{0.3cm}
	\begin{enumerate}[leftmargin=*, resume]
		\item \textit{Fama and French 3-Factor Model (FF3F)}
	\end{enumerate}
	The Fama and French 3-factor model offers an extension to the single index model by \citesupp{Sharpe1963} and is defined as:
	\begin{align}
		\bx_t =   \bbeta_1 f_{1t} + \bbeta_2 f_{2t} + \bbeta_3 f_{3t} + \bvar_t \, . \label{FF3}
	\end{align}
	
	The first factor $f_1$ is identical to the one of the one-factor model  in \eqref{single_index}. The second factor $f_{2t}$ often denoted by the acronym SMB is composed as the average returns on the three small portfolios minus the average returns on the three big portfolios. In particular, it defines a zero-cost portfolio that is long in stocks with a small market capitalization and short in stocks with a large market capitalization\footnote{It is important to note that securities with a long position in a portfolio are expected to rise in value and on the other hand securities with short positions in a portfolio are expected to decline in value.}. The third factor $f_{3t}$, denoted as HML,  comprises a zero-cost portfolio that is long in stocks with a high book-to-market value and short in low book-to-market stocks
	\footnote{A detailed definition of the factors can be found on the website of Kenneth R. French. See \text{http://mba.tuck.dartmouth.edu/pages/faculty/ken.french/data\_library.html}}.
	In matrix notation \eqref{FF3} is given by:
	\begin{align}
		\bX = \bbeta \bFa' + \bvar \, , 
	\end{align}
	where $\bFa = [\bfa_1, \bfa_2, \bfa_3]$ with dimension $T \times 3$ and $\bbeta = [\bbeta_1, \bbeta_2, \bbeta_3]$ with dimension $N \times 3$.\\
	The estimator for the covariance matrix for the 3-factor model by \citesupp{Fama1993} $\bSigma_{\text{FF}}$ is equal to the following equation:
	\begin{align*}
		\hat{\bSigma}_{\text{FF}} = \hat{\bbeta} \hat{\bSigma}_{F} \hat{\bbeta}' + \hat{\bD}_{\text{FF}} \, ,
	\end{align*}
	where $\hat{\bSigma}_{F}$ denotes the covariance matrix of the three factors and $\hat{\bD}_{\text{FF}}$ represents a diagonal matrix that contains the variances of the OLS residuals covariance matrix on its main diagonal.
	
	\item[\textbf{B.}]\textbf{\textit{Covariance Matrix Shrinking Strategies}}\\
	Within the class of covariance matrix shrinkage strategies, we consider the method proposed  by \citesupp{Ledoit2003}, \citesupp{Kourtis2012}, 
	the design-free estimator by \citesupp{AbadirDistasoZikes2014} and \citesupp{LedoitWolf2018}.
	
	\begin{enumerate}[wide, labelwidth=!, labelindent=0pt]
		\setlength\itemsep{0.5cm}
		\item \textit{\citesupp{Ledoit2003}} (LW)\\
		The LW approach shrinks the sample covariance matrix $\bS_{x}$ towards the covariance matrix of a single index model that is well-conditioned. This yields the following definition:  
		\begin{align*}
			\hat{\bSigma}_{\text{LW}} = \alpha^{*} \bS_{x} + (1 - \alpha^{*}) \hat{\bSigma}_\text{SIM},
		\end{align*}
		where $\alpha^{*} \in (0,1)$ is a constant, which corresponds to the shrinkage intensity. \citesupp{Ledoit2003} propose the following estimator to be used in practice $\hat{\alpha}^{*} = \frac{1}{T}\frac{\tau-\rho}{\gamma}$, where $\tau$ denotes the error on the sample covariance matrix, $\rho$ measures the covariance between the estimation errors of $\hat{\bSigma}_\text{SIM}$ and $\bS_{x}$ and $\gamma$ accounts for the misspecification of the shrinkage target $\hat{\bSigma}_\text{SIM}$.
		\item \textit{\citesupp{Kourtis2012} (KDM)}\\
		The estimation method by \citesupp{Kourtis2012} directly shrinks the inverse of the sample-based covariance matrix $\bS_{x}$ towards the identity matrix $I_N$ and the inverse of the covariance matrix resulting from a single index model by \citesupp{Sharpe1963}, according to the following equation:
		\begin{align}
			\hat{\bSigma}^{-1}_{\text{KDM}} = \zeta_1 \bS_{x}^{-1} + \zeta_2 \bI_N + \zeta_3 \hat{\bSigma}_\text{SIM}^{-1}. \label{kouris_three}
		\end{align}
		The authors show that the resulting weights are a three-fund strategy, i.e. a linear combination of the sample-based weights $\hat{\bomega}$, the equally weighted portfolio weights $\hat{\bomega}_{1/N}$ and those of the single index model model $\hat{\bomega}_\text{SIM}$.
		In order to select the optimal shrinkage coefficients in \eqref{kouris_three}, the authors suggest minimizing the out-of-sample portfolio variance using cross-validation. It is important to note that this portfolio strategy is also applicable for the case when $N > T$. In order to obtain reliable results for the inverse of $\bS_{x}$ the authors use the Moore-Penrose pseudo-inverse.
		\item \textit{\citesupp{AbadirDistasoZikes2014}} (ADZ)\\
		The design-free estimator for the covariance matrix by \citesupp{AbadirDistasoZikes2014} aims to improve the estimation of the eigenvalues $\hat{\bP}$ of $\bS_x$, that is a possible source of ill-conditioning. The authors consider the following spectral decomposition of $\bS_{x}$:
		\begin{align}
			\bS_{x} = \hat{\bGamma}\hat{\bP}\hat{\bGamma}'.
		\end{align}
		In order to obtain an improved estimator for 
		$\bP$, $\bX$ is split into two subsamples $\bX = \left(\underset{N \times n}{\bX_1}, \underset{N \times (T- n)}{\bX_2} \right)$.
		Calculating the sample covariance matrix for the first $n$ observations yields:
		\begin{align}
			\bS_1 = \frac{1}{n} \bX_1\bM_n \bX_1' = \hat{\bGamma}_1\hat{\bP}_1\hat{\bGamma}_1',
		\end{align}
		where $\bM_n = \bI_n - \frac{1}{n}\mathbf{1}_n\mathbf{1}_n'$ is the de-meaning matrix of dimension $n$ and $\mathbf{1}_n$ denotes a $n \times 1$ vector of ones. The spectral decomposition of $\bS_1$ provides the matrix of eigenvectors $\hat{\bGamma}_1$ and the diagonal matrix of eigenvalues $\hat{\bP}_1$.
		
		In the second step, an improved estimator for $\bP$ is computed from the remaining orthogonalized observations:
		\begin{align}
			\tilde{\bP} = \text{diag}\left(\cova{\left[\hat{\bGamma}_1'\bX_2\right]}\right) = \text{diag}\left(\hat{\bGamma}_1'\bS_2\hat{\bGamma}_1\right).
		\end{align}
		The new estimator for the covariance matrix is now obtained according to:
		\begin{align}
			\hat{\bSigma}_\text{AZD} = \hat{\bGamma}\tilde{\bP}\hat{\bGamma}'.
		\end{align}
		\item \textit{\citesupp{LedoitWolf2018}} (LW-NL)\\
		Another method that aims to improve on the estimation of the eigenvalues of $\bS_x$ is provided by \citesupp{LedoitWolf2018}. The covariance estimator is given by:
		\begin{align}
			\hat{\bSigma}_\text{LW-NL} = \hat{\bGamma}\hat{\bD}\hat{\bGamma}',
		\end{align}
		where $ \hat{\bGamma}$ are the sample eigenvectors of $\bS_x$ and the eigenvalues in the diagonal matrix $\hat{\bD}$ are estimated in a non-linear fashion as in Theorem 6.2. in \citesupp{LedoitWolf2018}.
		
	\end{enumerate}
	\item[\textbf{C.}]\textbf{\textit{Sparse Covariance Estimators}}\\
	The following estimators are explicitly designed to provide sparse covariance matrices. Hence, these models are appropriate for empirical settings that are reflected by our second simulation design.
	\begin{enumerate}[wide, labelwidth=!, labelindent=0pt]
		\setlength\itemsep{0.5cm}
		\item \textit{\citesupp{Rothman2009} (ST)}\\
		As a special case of the generalized thresholding estimators studied by \citesupp{Rothman2009}, we use the soft-thresholding (ST) method as a sparse covariance estimator and obtain:
		\begin{align*}
			\hat \bSigma_\text{ST} = \left(\hat{\sigma}_{\text{ST},ij}\right)_{N \times N}, \quad \hat{\sigma}_{\text{ST}, ij} = \left\{\begin{array}{ll}
				\hat{\sigma}_{s,ij},   		&i = j	\\
				\mathcal{S}(\hat{\sigma}_{s,ij}, \kappa),	    &i \neq j 
			\end{array}\right.
		\end{align*}
		where $\hat{\sigma}_{s,ij}$ is the $ij$-th element of the sample covariance matrix and $\mathcal{S}$ denotes the soft-thresholding operator defined in \eqref{soft_t}. The thresholding parameter $\kappa$ is selected by minimizing the difference between $\hat \bSigma_\text{ST}$ and $\bS_x$ in Frobenius norm based on cross-validation.
		\item \textit{\citesupp{BienTibshiraniothers2011} (BT)}\\
		The authors propose a penalized maximum likelihood estimator based on a lasso penalty in order to allow for sparsity in the covariance matrix and to reduce the effective number of parameters. More specifically, the following objective function is optimized:  
		\begin{align*}
			\min_{\bSigma \succ 0} \quad \log\det\left(\bSigma\right) + \trace{\bSigma^{-1} \bS_x} + \alpha_N \sum_{i = 1}^{N} \sum_{j = 1}^{N} \left|h_{ij} \sigma_{ij}\right|,
		\end{align*}
		where $\alpha_N$ is a regularization parameter selected based on 5-fold cross-validation. The $ij$-th element of the selection matrix $H$ is defined as $h_{ij} = \1\left\{i \neq j\right\}$ and enables an equal penalization of the off-diagonal elements and leaves the diagonal elements unaffected. Furthermore, \citesupp{BienTibshiraniothers2011} show that the estimated sparse covariance matrix is positive definite.  
	\end{enumerate}
\end{itemize}

\bibliographystylesupp{econometrica}
\bibliographysupp{DaPoZaFLasso}



\end{document}

%% file: Z_PACKAGES.tex
\usepackage[english]{babel}   			   
\usepackage[OT1]{fontenc} 	               
\usepackage[utf8]{inputenc}              
\usepackage[shortcuts]{extdash}
\usepackage{setspace} 
\usepackage{microtype}

\usepackage{amsmath,amsfonts,amsthm,amssymb,bbm}
\usepackage{marvosym,nicefrac}
\usepackage{bm} 
\usepackage{stmaryrd}
\usepackage[resetlabels]{multibib}
\usepackage[longnamesfirst,round]{natbib}
\usepackage{mathtools}    

\usepackage[hang]{footmisc}
\usepackage{rotating}
\usepackage{commath}
\usepackage{etex}
\usepackage{float}
\usepackage{subfig}
\usepackage[colorlinks,citecolor=blue,urlcolor=blue,breaklinks]{hyperref}
\usepackage{breakurl}

\usepackage{color,graphicx,epsfig,epstopdf}
\definecolor{DarkBlue}{rgb}{0.1,0,0.55}
\definecolor{DarkGreen}{RGB}{24,126,35}
\definecolor{lgrey}{RGB}{245,245,245}     
\definecolor{colKeys}{RGB}{0,0,255}       
\definecolor{colIdentifier}{RGB}{0,0,0}	  
\definecolor{colComments}{RGB}{34,139,34} 
\definecolor{colString}{RGB}{160,32,240}  

\usepackage{caption}
\usepackage{booktabs}
\usepackage{bigstrut}
\usepackage{multirow,multicol}
\usepackage{longtable, tabularx, array, pbox, rotating}
\usepackage{enumitem} 
\usepackage[flushleft]{threeparttable}
\usepackage[section]{placeins}
\usepackage{adjustbox}

\usepackage{etoolbox}

%% file: Z_COMMANDS.tex
\DeclareMathOperator*{\cova}{Cov} 
\newcommand{\cov}[1]{{\cova}\left[ {#1} \right]}

\newcommand{\1}{\mbox{$\mathrm{1\hspace*{-2.5pt}l}$\,}}

\newtoggle{withcomments}

\DeclareMathOperator{\bFa}{\pmb{F}}
\DeclareMathOperator{\bfa}{\pmb{f}}
\DeclareMathOperator{\bLam}{\pmb{\Lambda}}
\DeclareMathOperator{\blam}{\pmb{\lambda}}
\DeclareMathOperator{\bu}{\pmb{u}}
\DeclareMathOperator{\bvar}{\pmb{\varepsilon}}
\DeclareMathOperator{\bbeta}{\pmb{\beta}}
\DeclareMathOperator{\bomega}{\pmb{\omega}}
\DeclareMathOperator{\bX}{\pmb{X}}
\DeclareMathOperator{\bx}{\pmb{x}}
\DeclareMathOperator{\bSigma}{\pmb{\Sigma}}

\DeclareMathOperator{\bPhi}{\pmb{\Phi}}
\DeclareMathOperator{\bS}{\pmb{S}}
\DeclareMathOperator{\bA}{\pmb{A}}
\DeclareMathOperator{\bC}{\pmb{C}}
\DeclareMathOperator{\bG}{\pmb{G}}
\DeclareMathOperator{\bI}{\pmb{I}}
\DeclareMathOperator{\bD}{\pmb{D}}
\DeclareMathOperator{\bP}{\pmb{P}}
\DeclareMathOperator{\bB}{\pmb{B}}

\DeclareMathOperator{\bGamma}{\pmb{\Gamma}}
\DeclareMathOperator{\bM}{\pmb{M}}

\newcommand{\E}[1]{\mathbb{E}\left[#1\right]}

\newcommand{\Prob}[1]{\mathbb{P}\left(#1\right)}
\newcommand{\trace}[1]{\textup{\text{tr}}\left(#1\right)}
\newcommand{\traces}[1]{\textup{\text{tr}}\left[#1\right]}
\newcommand{\spec}[1]{\left\lVert #1\right\lVert}
\newcommand{\frob}[1]{\left\lVert #1\right\lVert_F}
\newcommand{\lone}[1]{\left\lVert #1\right\lVert_1}
\setlist[itemize]{leftmargin=2cm}
\theoremstyle{plain}
\newtheorem{assumption}{Assumption}
\newtheorem{theorem}{Theorem}
\newtheorem{lemma}[theorem]{Lemma}
\newtheorem{prop}{Proposition}
\usepackage{mathtools}           

\usepackage{changepage}
\usepackage[flushleft]{threeparttable}
\usepackage{chngcntr}
\usepackage{apptools}
\counterwithin{assumption}{section}
\counterwithin{theorem}{section}
\counterwithin{prop}{section}
\renewenvironment{proof}{{\bfseries Proof.}}{\qed}
\newcites{main}{References}
\newcites{supp}{References}